\documentclass[11pt]{article}

\usepackage{arxiv}

\usepackage[utf8]{inputenc} % allow utf-8 input
\usepackage[T1]{fontenc}    % use 8-bit T1 fonts
\usepackage{hyperref}       % hyperlinks
\usepackage{url}            % simple URL typesetting
\usepackage{booktabs}       % professional-quality tables
\usepackage{amsfonts}       % blackboard math symbols
\usepackage{nicefrac}       % compact symbols for 1/2, etc.
\usepackage{microtype}      % microtypography
\usepackage{lipsum}		% Can be removed after putting your text content
\usepackage{graphicx}
\usepackage{doi}

\usepackage{comment}
\usepackage{subfigure}
\usepackage{xcolor}
\usepackage{thmtools,thm-restate}
\usepackage{makecell}
\usepackage{amsthm}

\usepackage{amsmath}
\DeclareMathOperator*{\argmax}{arg\,max}

\newsavebox\savefigure
\newsavebox\savefigureA
\newsavebox\savefigureB
\usepackage{calligra}

\usepackage{multirow}

\newcommand{\uk}{\underline{k}_1}
\newcommand{\ok}{\overline{k}_1}
\newcommand{\olambda}{\overline{\lambda}}
\newcommand{\ulT}{\underline{\ell}_2}
\newcommand{\ukT}{\underline{k}_2}
\newcommand{\nknote}[1]{\textcolor{blue}{(NK: #1)}}

\newtheorem{theorem}{Theorem}[section]
\newtheorem{proposition}{Proposition}[section]
\newtheorem{lemma}[theorem]{Lemma}
\newtheorem{claim}{Claim}
\newtheorem{definition}{Definition}

\title{Dynamic Bipartite Matching Market\\ with Arrivals and Departures \thanks{ An extended abstract is to appear in ACM WINE 2021. } 
}

\author{ 
Naonori Kakimura
\thanks{Supported by JSPS KAKENHI Grant Numbers JP17K00028, JP18H05291, 20H05795, and 21H03397.} 
\\
	Department of Mathematics\\
	Keio University\\
	Yokohama, Japan 223-8522 \\
	\texttt{kakimura@math.keio.ac.jp} \\
	\And
	{
	Donghao Zhu \thanks{Supported by the Deutsche Forschungsgemeinschaft (DFG, German Research Foundation) – 277991500/GRK2201.} } \\
	Department of Informatics\\
	Technical University of Munich\\
	Munich, Germany 85748 \\
	\texttt{donghao.zhu@in.tum.de} \\
}

\renewcommand{\headeright}{ }
\renewcommand{\undertitle}{ }
\renewcommand{\shorttitle}{Dynamic Bipartite Matching Market with Arrivals and Departures}

\hypersetup{
pdftitle={Dynamic Bipartite Matching Market with Arrivals and Departures},
pdfsubject={NK, DZ},
pdfauthor={Naonori Kakimura, Donghao Zhu},
pdfkeywords={Bipartite Matching, Markov chain, Online algorithm},
}

\begin{document}
\maketitle

\begin{abstract}
In this paper, we study a matching market model on a bipartite network where agents on each side arrive and depart stochastically by a Poisson process.
For such a dynamic model, we design a mechanism that decides not only which agents to match, but also when to match them, to minimize the expected number of unmatched agents.
The main contribution of this paper is to achieve theoretical bounds on the performance of local mechanisms with different timing properties.
We show that an algorithm that waits to thicken the market, called the \textit{Patient} algorithm, is exponentially better than the \textit{Greedy} algorithm, i.e., an algorithm that matches agents greedily.
This means that waiting has substantial benefits on maximizing a matching over a bipartite network.
We remark that the Patient algorithm requires the planner to identify agents who are about to leave the market, and, under the requirement, the Patient algorithm is shown to be an optimal algorithm.
We also show that, without the requirement, the Greedy algorithm is almost optimal.
In addition, we consider the \textit{1-sided algorithms} where only an agent on one side can attempt to match.
This models a practical matching market such as a freight exchange market and a labor market where only agents on one side can make a decision.
For this setting, we prove that the Greedy and Patient algorithms admit the same performance, that is, waiting to thicken the market is not valuable.
This conclusion is in contrast to the case where agents on both sides can make a decision and the non-bipartite case by~[Akbarpour et al.,~\textit{Journal of Political Economy}, 2020].
\end{abstract}

% keywords can be removed

\clearpage
\section{Introduction}

Matching markets arise in many applications such as marriage and dating market~\cite{kurino2020credibility}, paired kidney exchange~\cite{akbarpour2020thickness}, and ride-hailing system~\cite{dickerson2018allocation, wang2019ridesourcing}.
In a matching market, which can be modeled as a network with agents~(vertices) and edges, 
a social planner designs a mechanism that finds an acceptable matching on the network.
In a \textit{dynamic} matching market, agents are allowed to arrive and depart over time.
A market is then changed dynamically over time, in which a social planner designs a mechanism that chooses how to match agents.
%For example, in paired kidney exchange, a patient-donor pair forms an agent, and two biologically compatible pairs are connected by edges.
%Then the market aims to find a matching on the network so that they can swap kidneys.
%The difficulty of the dynamic setting is in that the planner’s decision today affects the future.

A dynamic matching market has been studied extensively in theory~\cite{karp1990optimal, mehta2013online, akbarpour2020thickness} and practice~\cite{emek2016online, henzinger2020dynamic}. 
Recently, Akbarpour et al.~\cite{akbarpour2020thickness} introduced a seminal matching market model with arrivals and departures.
In their model, agents arrive at and depart from the market according to the Poisson process.
The planner observes the network and chooses a matching, aiming to minimize the number of unmatched agents.
One of the key feature in their model is that the planner must decide not only which agents to match, but also \textit{when} to match them. 
Akbarpour et al.~\cite{akbarpour2020thickness} showed that the choice of when to match agents has large effects on performance.
Specifically, they introduced two simple mechanisms with different timing properties, \textit{Greedy} and \textit{Patient}. They provided theoretical guarantees for these mechanisms, that suggests waiting has substantial benefits on maximizing a matching over the network.

This paper focuses on a \emph{bipartite matching market} where the network is a bipartite graph.
Agents in the market are divided into two separated groups, and a matching is formed between the two groups. 
A bipartite matching market is one of the most popular matching markets in practice; 
a labor market matches a worker to a task, and 
a ride-hailing market matches a taxi to a passenger~\cite{doval2014theory, sun2020taxi}.

We propose a bipartite matching market model with arrivals and departures as a variant of Akbarpour et al.'s (non-bipartite) matching model.
We aim at designing \textit{local} algorithms in the sense that they look only at the neighbors of an agent which attempts to match, rather than at the global network structure.
Local algorithms can be viewed as a mechanism that each agent individually decides to find a partner.
In a bipartite matching market, 
agents in two separated groups have different roles, and agents on one side often have no right to make a decision.
For example, in a freight exchange market between shippers and carriers, 
some platforms such as Wtransnet\footnote{\url{https://www.wtransnet.com/} [Online; accessed 24-January-2021]} only allow carriers to choose shipments.
For another, in a competitive labor market, only workers submit job applications to companies, and companies make final decisions.
Thus, it is natural to consider the situation when agents on only one side have a right to make a decision.
Such a setting is called a \textit{1-sided market}~\cite{abdulkadiroglu2013matching}.
We also consider the situation when agents on both sides can make a decision, called a \textit{2-sided market}, that also appears in practice such as a marriage and dating market and a freight exchange market like Cargopedia\footnote{\url{https://www.cargopedia.net/} [Online; accessed 23-December-2019]}.

%%%%%%%%%%%%%%%%%%%%%%%%%%%%%%%%%%%%%%%%%%%%%%%%%%%%%%%%%%%%%%%%
\subsection{Our Contributions}
%%%%%%%%%%%%%%%%%%%%%%%%%%%%%%%%%%%%%%%%%%%%%%%%%%%%%%%%%%%%%%%%

In this work, we evaluate the performance of simple local mechanisms on a bipartite matching market to measure the impact of waiting time in the 1-sided/2-sided markets.
Our main contributions are summarized as follows:
\begin{itemize}
    \item 
    We introduce a formal framework of 1-sided/2-sided bipartite market model with arrivals and departures. 
%    \item We propose a benchmark market model framework (infinite-horizon, continuous-time) to simulate agent behaviors in bipartite matching markets. This framework can be applied on both 2-sided market and 1-sided market.
    We propose algorithms with different timing properties, Greedy and Patient algorithms, for the 1-sided and 2-sided markets, respectively.
    We present almost optimal bounds on the performance of these algorithms. 
    Our results show that waiting to thicken the market is highly valuable for the 2-sided market, while it is not true for the 1-sided market.
    \item We provide lower bounds on the performance of any matching algorithms.
    We show that, if the planner does not know the information when an agent departs, any algorithm suffers a loss exponentially larger than that of an omniscient algorithm where the information is available.
\end{itemize}

Let us describe our results in more detail.

\paragraph{Model}
In our model, agents in two classes arrive at Poisson rates $\lambda_a$ and $\lambda_b$, respectively, and a pair of two agents in different classes are compatible with probability $p$.
Each agent departs at a Poisson rate, normalized to 1.
The planner chooses a matching on the current network, and matched agents leave the market. 
The planner aims to minimize the proportion of the expected unmatched agents (called \textit{the loss}).
This setting is a variant of a matching market model by Akbarpour et al.~\cite{akbarpour2020thickness}, where each agent arrives at Poisson rate $\lambda$ and edges are formed between any pair of agents with probability $p$.

In this paper, we consider two simple mechanisms, \textit{Greedy} and \textit{Patient}, for a bipartite matching market.
The Greedy algorithm attempts to match an agent upon her arrival, while the Patient algorithm attempts to match only an urgent agent, that is, an agent at departure.
Note that both these algorithms are local, in the sense that an agent individually makes a decision when she arrives in the Greedy algorithm, and when she departs in the Patient algorithm.
In the 2-sided market, every agent attempts to match to some agent according to Greedy or Patient algorithms, while, in the 1-sided market, agents in one side do it and agents in the other side stay in the market without making a decision~(\textit{inactive}).
%they look only at the immediate neighbors of the agent they attempt to match, rather than at the global network structure.
Note that the Patient algorithm requires the planner to know which agents will
perish imminently if not matched.
The information is referred to as the \textit{departure information}.
%know agents' departure time (also referred to as \textit{departure information} in our work), which is expensive in some scenarios. 

\paragraph{Theoretical Guarantee}
Our main contributions are to derive theoretical bounds on the Greedy and Patient algorithms in the 2-sided and 1-sided markets, respectively.
The obtained guarantees are summarized as in Tables~\ref{tab:result1} and~\ref{tab:result2}.
We here denote $d_i = \lambda_i p$ for $i\in \{a, b\}$ and $\Delta = \frac{|d_a-d_b|}{d_a+d_b}$.
We remark that lower bounds for the 2-sided market model were also derived by Jiang~\cite{jiang2018bipartite}.

%%%%%%%%%%%%%%%%%%%%%%%%%%%%%%%%%%%%%%%%%%%%%%%%%%%%%%%%%%%%%%%%%%%%%%%%%%%%%%%%%%%%%%%%%%%%%%%%%%%%%%%%%%%%%%%

\begin{table}[t]
\caption{Summary of the loss when $T, \lambda, \lambda_a, \lambda_b \to \infty$. We denote $d = \lambda p$ and $d_i = \lambda_i p$ for $i\in\{a, b\}$, which are all constants.} 
\begin{tabular}{c}
    \begin{minipage}[c]{0.9\hsize}
      \begin{center}
%\label{table:result2}
\subfigure[Lower bounds of the loss, and upper bounds of the loss when $\lambda_a=\lambda_b$]{\begin{tabular}{cc|c|c} 
   \hline
    \multicolumn{2}{c|}{Setting} & \multicolumn{2}{c}{Loss} 
    \\ 
    & & Lower bound & Upper bound~($d_a=d_b$)  \\ \hline \hline 
    \multirow{2}{*}{non-bipartite} & Greedy~\cite{akbarpour2020thickness} & $\frac{1}{2d+1}$ & $\frac{\log (2)}{d}$ \\ \cline{2-4}
    &Patient~\cite{akbarpour2020thickness} & $\frac{e^{-d}}{d+1}$ & $\frac{e^{-d/2}}{2}$  \\ \hline \hline
    \multirow{2}{*}{$2$-sided} & Greedy & $\max\left\{\Delta, \frac{1}{2d_a+d_b+1}\right\}$~(cf.~\cite{jiang2018bipartite}) &  $\frac{2\log (d_a+3)}{d_a}$\\ \cline{2-4}
     & Patient &$\frac{1}{2}\left(\frac{e^{-d_a}}{d_a+1}+\frac{e^{-d_b}}{d_b+1}\right)$&$e^{-O(d_a)}$ \\ \hline
    \multirow{2}{*}{$1$-sided} & Greedy & $\max\left\{\Delta, \frac{1}{1+2d_a + d_b}\right\}$ & \multirow{2}{*}{\large $\frac{2\log (d_a+3)}{d_a}$}\\ \cline{2-3}
    & Patient &  $\max\left\{\Delta, \frac{\log d_b}{d_a+d_b}\right\}$ &  \\ \hline
    \end{tabular}
    \label{tab:result1}
    }
      \end{center}
      \end{minipage}\\
    \begin{minipage}[c]{0.9\hsize}
      \begin{center}
%\label{table:result2}
\subfigure[Upper bounds when  $\lambda_a\neq \lambda_b$. In the 2-sided market, we assume $\lambda_a\geq \lambda_b$, and in the 1-sided market, we assume that agents with rate $\lambda_a$ are inactive.]
  {\begin{tabular}{cc|c||c|c} 
   \hline
    \multicolumn{2}{c|}{Setting} & \multicolumn{3}{c}{Loss} \\ 
    &        & Total & $\lambda_a$-side: $\mathbf{L}_a$ & $\lambda_b$-side: $\mathbf{L}_b$ \\ \hline \hline 
\multirow{2}{*}{\shortstack{$2$-sided\\ ($d_a\geq d_b$)}} & Greedy & $\Delta +\frac{2\log (d_b+3)}{d_a+d_b}$ & \multirow{2}{*}{$\frac{d_a-d_b}{d_a}+\frac{\log (d_b+3)}{d_a}$}& $\frac{\log (d_b+3)}{d_b}$\\ \cline{2-3}\cline{5-5}
& Patient & $\Delta +\frac{\log (d_b+3)}{d_a+d_b}+e^{-\max\left\{d_a- d_b, \frac{d_a}{1+d_b}\right\}}$ & & $e^{-\max\left\{d_a- d_b, \frac{d_a}{1+d_b}\right\}}$\\ \hline
\multirow{2}{*}{\shortstack{$1$-sided\\ ($d_a\geq d_b$)}} & Greedy & \multirow{4}{*}{$\Delta +\frac{2\log (d_b+3)}{d_a+d_b}$} & \multirow{2}{*}{$\frac{|d_a-d_b|}{d_a} +\frac{\log (d_b+3)}{d_a}$} &   \multirow{2}{*}{$\frac{\log (d_b+3)}{d_b}$}\\ \cline{2-2}
& Patient &  & & \\ \cline{1-2}\cline{4-5}
\multirow{2}{*}{\shortstack{$1$-sided\\ ($d_a<d_b$)}} & Greedy      &  & \multirow{2}{*}{$\frac{\log (d_b+3)}{d_a}$} & \multirow{2}{*}{$\frac{|d_a-d_b|}{d_b} + \frac{\log (d_b+3)}{d_b}$}\\ \cline{2-2}
& Patient &  & & \\ \hline  
\end{tabular}
\label{tab:result2}
    }
      \end{center}
      \end{minipage}
\end{tabular}
\end{table}
%%%%%%%%%%%%%%%%%%%%%%%%%%%%%%%%%%%%%%%%%%%%%%%%%%%%%%

%In the 2-sided market, we prove that the loss of the Patient algorithm is exponentially better than that of the Greedy algorithm.
Let us first consider the balanced case, that is, when $\lambda_a=\lambda_b$, implying that $d_a=d_b$ and $\Delta = 0$ in Table~\ref{tab:result1}.
Table~\ref{tab:result1} shows that the loss of the 2-sided Greedy algorithm is $\Theta \left(\frac{1}{d_a}\right)$, ignoring a logarithmic factor in $d_a$, while the 2-sided Patient algorithm has the loss $e^{-\Theta(d_a)}$.
Thus waiting to match agents allows us to achieve exponentially small loss, which is a similar consequence to the non-bipartite matching market~\cite{akbarpour2020thickness}.
In contrast, the 1-sided market leads to a different conclusion.
In fact, both of the Greedy and Patient algorithms have the same loss, which is $\Theta \left(\frac{1}{d_a}\right)$, ignoring a logarithmic factor in $d_a$.
This means that waiting to match agents is not valuable in the 1-sided market, and other information such as the graph structure is necessary to achieve smaller loss.

The situation changes when $d_a\neq d_b$.
For better understanding, we evaluate the proportion of unmatched agents on both sides separately, which are the losses $\mathbf{L}_a$ and $\mathbf{L}_b$ of $\lambda_a$-side and $\lambda_b$-side, respectively, in Table~\ref{tab:result2}.
Note that the total loss is equal to $\frac{d_a}{d_a+d_b}\mathbf{L}_a + \frac{d_b}{d_a+d_b}\mathbf{L}_b$.
We see from Table~\ref{tab:result2} that the larger side, i.e., the side with $\max\{d_a, d_b\}$, has a constant loss of $\frac{|d_a-d_b|}{\max\{d_a, d_b\}}$ in every market.
This factor is unavoidable since a bipartite graph is unbalanced.
%Since agents in two classes arrive at Poisson rates $\lambda_a$ and $\lambda_b$, the expected numbers of agents on both sides in the time interval $[0, T]$ are $\lambda_a T$ and $\lambda_b T$.
%Then it cannot have a matching of size greater than $\min \{\lambda_a, \lambda_b\}T$, which implies that $|\lambda_a -\lambda_b| T$ agents can never be matched during the process.
Our results say that, except for the unavoidable loss, we suffer only the loss of $O\left(\frac{1}{\max\{d_a, d_b\}}\right)$ on the large side in every market.

In the $2$-sided market when $d_a\neq d_b$, the smaller side of the Patient algorithm has exponentially smaller loss than that of the Greedy algorithm.
This again indicates that waiting to thicken the market in the $2$-sided market is beneficial.
In contrast, both of 1-sided Greedy and Patient algorithms have the same loss as the 2-sided Greedy algorithms.

We remark that, in the 1-sided Greedy algorithm, agents on one side do not attempt to match.
%a greedy agent can find a partner only at her arrival, and ignores new agents on the other side after her arrival, since the new agents are inactive and do not attempt to match.
Hence, it has less opportunity to make a partner compared to the 2-sided Greedy algorithm, which implies that the 1-sided Greedy algorithm seems to have larger loss.
However, our results show that their losses have the same order.
On the other hand, in the 1-sided Patient algorithm, since an active agent delays her decision, she is allowed to have more neighbors. 
Hence the 1-sided Patient algorithm intuitively has smaller loss than the 1-sided Greedy algorithm.
However, our results show that their losses have the same order.
In fact, Table~\ref{tab:result1} shows that the loss of the 1-sided Patient algorithm is strictly worse than the 2-sided one when $d_a = d_b$.

Another contribution of this paper is to evaluate the loss of optimal algorithms.
We show that \textit{any} algorithm suffers a loss of at least $1/(2d_a+d_b+1)$ if it does not know the departure information.
In other words, no matter how long each agent waits, the loss must be at least $1/(2d_a+d_b+1)$.
Thus the Greedy algorithm is almost optimal, up to a logarithmic factor in $d_a$.
In contrast, if we know the departure information, we prove that the loss of \textit{any} algorithm is at least $\frac{1}{2}\left(\frac{e^{-d_a}}{d_a+1}+\frac{e^{-d_b}}{d_b+1}\right)$.
Thus, since the loss of the Patient algorithm is $e^{-O(d_a)}$ when $d_a = d_b$, waiting to match agents suffices to achieve optimal loss.

%%%%%%%%%%%%%%%%%%%%%%%%%%%%%%%%%%%%%%%%%%%%%%%%%%%%%%%%%%%%
\paragraph{Technical Highlights}

The key observation for bounding the loss is that the number of agents in the market determines the loss of matching algorithms.
This is observed in a non-bipartite market~\cite{akbarpour2020thickness} as well.
In our bipartite markets, in particular, the loss on one side is determined by the number of agents on the other side; an agent is likely to be matched if there are many agents on the other side, and the number of agents on the same side does not matter.

In the 2-sided market, since the Greedy algorithm attempts to match agents as soon as possible, the number of available agents on both sides is reduced rapidly when $d_a$ and $d_b$ grow~(the market is \textit{thin}).
Since the market has no edges under the Greedy algorithm, all urgent agents perish, which are counted as the loss.
On the other hand, the Patient algorithm attempts to match only urgent agents, which implies that the number of agents on both sides will remain large even when $d_a$ and $d_b$ increase~(the market is \textit{thick}).
This allows the planner to find a pair to an urgent agent, which reduces the loss.
We remark that, in the case when $d_a > d_b$, since the number of agents on the $\lambda_b$-side is small compared to the one on the $\lambda_a$-side, agents on the $\lambda_a$-side is hard to find a partner even if the market is thick, which worsens the loss of the larger side.
%Thus, the loss of the larger side becomes a constant when $d_a \gg d_b$.

A similar observation can be applied to the 1-sided market.
As observed in the 2-sided market, 
%the number of available agents does not necessarily affect the loss.
the market size of active agents~(i.e., agents who can make a decision) will be thin under the Greedy algorithm, while it will be thick under the Patient algorithm.
However, as we will see, the number of inactive agents decreases rapidly under both algorithms when $d_a$ and $d_b$ grow.
%Since a matching is formed between active and inactive agents, 
This causes large loss for both the algorithms. 
%The above observation makes the loss of both algorithms almost the same.

The above observation can be formalized with Markov chain. 
That is, the dynamics of our proposed algorithms can be modeled as continuous-time Markov chains determined by a pair of market sizes on both sides.
We first show that the loss of the proposed algorithms can be expressed as the pool sizes in the stationary distribution of the Markov chain.
Moreover, we prove that, for each of the proposed algorithms, the pool sizes in the stationary distribution highly concentrate around some values, which allows us to upper-bound the loss of the algorithms.

The most challenging part is to find the concentration of the pool sizes in the steady state.
The primary technical tool is the balance equations of Markov chains.
The balance equation describes the probability flux in and out of a given set of states.
For a non-bipartite matching model~\cite{akbarpour2020thickness}, a Markov chain is of a simple form on the set of non-negative integers, and hence we can naturally apply the balance equations.
On the other hand, our Markov chain is defined on 2-dimensional space, i.e.,  each state is a pair of market sizes. 
This requires us to choose a set of states for the balance equations more carefully.
In fact, we need to adopt different strategies for each of the proposed algorithms.
See Section~\ref{s66} for the details.

\subsection{Related Work}

There have been many studies on dynamic matching with different applications in the literature of economics, computer science, and operations research. 
In particular, the problem has received much attention recently in online advertising.
%\nknote{I asked: I am not so sure Karp et al gives an application, as the paper looks a theory paper}\dznote{Their model is designed for online advertising. This is not said by me, this is by many papers cited his paper. They said so. Check [28] for example.}
We refer readers to Mehta~\cite{mehta2013online} for a detailed survey on Ad display.
%\nknote{\cite{mehta2013online} needs bibtex information}
Another example is a paired kidney exchange market, which requires simultaneous transplantation between different patients~\cite{unver2010dynamic, ross1997ethics}.
%Crowd sourcing platforms allocate tasks to workers where both sides are online~\cite{karger2014budget}.
%For the online task allocation problem, heterogeneity of agents is well-considered under different settings~\cite{chen2019matching, anderson2017efficient, tong2018dynamic}.
Other applications include crowd sourcing platform~\cite{karger2014budget}, task allocation~\cite{chen2019matching, anderson2017efficient, tong2018dynamic}, house allocation~\cite{abdulkadirouglu1999house, sharam2018matching}, school choice~\cite{feigenbaum2020dynamic}, and real-time ride sharing~\cite{ozkan2020dynamic}.
%\nknote{Do you mean that they do not consider stochastic models?}\dznote{Most of these works are not stochastic models. I mean the arrivals are not stochastic. But mostly of them have a setting that the process of forming the edge is by a constant probability.}
%\nknote{To make the setting unclear, (since we have no time to classify..), I moved these sentences saying that they are applications of dynamic matching problem}

%In the online (bipartite) matching problem, a (bipartite) graph is not given in advance, and each vertex with its incident edges arrives one by one.
%Each time a vertex arrives, we decide irrevocably whether to make a pair to the vertex or not.

\paragraph{Stochastic Matching Market}
In the classical setting of the dynamic bipartite matching problem~(e.g.,~\cite{karp1990optimal, mehta2013online, dickerson2018allocation, gujar2015dynamic}), the vertex set on one side is given in advance.
Emek et al.~\cite{emek2016online} studied the setting where agents on both sides arrive randomly.
We note that their model assumes that each agent stays in the market until she gets matched.
They evaluated algorithms to find a maximum-weight matching incorporating waiting cost.
Ashlagi et al.~\cite{ashlagi2017min} generalized the results by Emek et al.~\cite{emek2016online} to apply ride-sharing platforms.
Mertikopoulos et al.~\cite{mertikopoulos2020quick} further extended to investigate the trade-off issue on waiting time and matching costs.
They showed that waiting to thicken the market is profitable, but no algorithm can find an optimal matching with both objectives.
%These papers indicate that waiting is profitable.
Loertscher et al.~\cite{loertscher2018optimal} proposed a batching algorithm (batch auctions) that matching occurs periodically in a bipartite network.
They indicate that a lower frequency of matching can increase the market thickness.
Manshadi and Ro~\cite{manshadi2020online} introduced a model where agent arrival distributions are endogenous.%\dznote{All these papers in this paragraph are based on a bipartite graph.}

As a generalization of a paired kidney exchange market,
the dynamic matching problem was extended to finding disjoint 3-way circles and chains~\cite{johari2016matching, akbarpour2020unpaired, anderson2015finding, ashlagi2012need, unver2010dynamic}.
%\nknote{I moved {\"U}nver et al.~\cite{unver2010dynamic} into the previous sentence.}
Anderson et al.~\cite{anderson2017efficient} and Ashlagi et al.~\cite{ashlagi2019matching} analyzed the expected waiting time in the market.
%of finding 3-way circle matching or chain matching, and showed a similar conclusion that the greedy algorithm is near-optimal.
%\nknote{Do they consider patient algorithms as well?}
%\dznote{The three works do not allow departures.}\dznote{All works in above paragraphs are based on a bipartite graph.}
%\nknote{If departures are not allowed, what do you mean by "Greedy is near optimal"? What is compared to greedy algorithm?}\dznote{They compare with the optimal algorithm without knowing the departure information. Greedy $\rightarrow$ locally match upon arrival; OPT $\rightarrow$ consider graph structure and match time. BTW, the goal in [5] and [9] is expected waiting time but not loss.}

\paragraph{Matching Market with Departures}

A matching market where each agent is allowed to leave has been studied in various settings.
Johari et al.~\cite{johari2016matching} and Ashlagi et al.~\cite{ashlagi2017min} studied a matching model where agents would depart after a constant time after arrival.
%\nknote{You said the references in the next sentence do not consider the departures of agents in your e-mail, but the sentence is different.}\dznote{No, the [20] in the e-mail is [26] now, and [26] does not consider departures.}Hu et al.~\cite{hu2018dynamic}, Aouad et al.~\cite{aouad2020dynamic} and Jiang~\cite{jiang2018bipartite} studied a model that the staying time of each agent is determined stochastically.

Akbarpour et al.~\cite{akbarpour2020thickness} introduced a dynamic model where each vertex arrives and departs stochastically on a general (i.e., non-bipartite) network.
Our model is a bipartite variant of their model.
Jiang~\cite{jiang2018bipartite} gave lower bounds on the loss of algorithms in the 2-sided market.
%introduced a bipartite matching model as a variant of Akbarpour et al.~\cite{akbarpour2020thickness}, which is the 2-sided model in this paper.
Beccara et al.~\cite{baccara2020optimal} studied a model where heterogeneous agents arrive randomly.

We remark that, in the previous work, waiting to thicken the market is shown to be more beneficial than myopic algorithms such as the Greedy algorithm.
Moreover, if the departure information is not available, the Greedy algorithm is nearly optimal.
Our results on the 2-sided market have the same conclusion that the Patient algorithm is exponentially better than the Greedy algorithm.
However, interestingly, in the 1-sided market, we obtain a different conclusion that waiting to thicken the market is not beneficial.

\subsection{Organization}
The paper is organized as follows.
We introduce our models and algorithms in Section~\ref{s2}.
Section~\ref{s3} summarizes our main theorems.
In Section~\ref{s66}, we present the proof outline to obtain upper bounds of the loss of the proposed algorithms.
The details follow in Sections~\ref{s92}--\ref{sec:LB}.
The conclusion is given in in Section~\ref{s7}.
%Due to the page limitation, we omit most of the proofs in this submission.
%The details may be found in the appendix.
%is described in  
%We show that the market's random process is Markovian and analyzes the performance of algorithms in each market in Section~\ref{s66}. Then we lower bound any algorithm in Section~\ref{s6}. Lastly, we conclude the work 

\section{Model Definition}\label{s2}

This section defines a continuous-time stochastic model for a bipartite matching market in a formal way.
The model runs in the time interval $[0, T]$.

%%%%%%%%%%%%%%%%%%%%%%%%%%%%%%%%%%%%%%%%%%%%%%%%%%%%%%%%%%
\paragraph{Bipartite Matching Market}
%%%%%%%%%%%%%%%%%%%%%%%%%%%%%%%%%%%%%%%%%%%%%%%%%%%%%%%%%%
In a market, there are two types of agents, say type $a$ and $b$.
An agent of both types arrives at the market according to the Poisson process with rates $\lambda_a$ and $\lambda_b$, respectively.
Thus, in any interval $[t, t+1]$, $\lambda_a$~(resp., $\lambda_b$) new agents of type $a$~(resp., type $b$) enter the market in expectation.
We denote by $U_t$ and $V_t$ the sets of agents of the two types at time $t$, respectively, called the \textit{pools} of the market.
We also refer to $A_t=|U_t|$ and $B_t=|V_t|$ as the \textit{pool sizes} at time $t$.
We assume that the market is empty at the beginning, i.e., $U_0=\emptyset$ and  $V_0=\emptyset$.

When an agent $v$ of some type arrives at the market, 
she forms edges to agents of the other type, that represents her compatible partners.
An edge is formed with a probability $p$. 
We denote by $E_t$ the set of edges between $U_t$ and $V_t$ at time $t$.
Thus the market forms a bipartite graph, denoted by $G_t=(U_t, V_t, E_t)$, at time $t$. 
We assume that edges persist over time until one of the end vertices leaves. 
For an agent $v \in U_t\cup V_t$, we use $N_t(v)$ to denote the set of neighbors of $v$ in $G_t$.
We also denote $U=\bigcup_{t=0}^T U_t$, $V=\bigcup_{t=0}^T V_t$, and $E=\bigcup_{t=0}^T E_t$.

Each agent can stay in the market for a while. Her staying time is modeled as an independent Poisson process. 
Without loss of generality, we can normalize time so that the Poisson process has rate one.
That is, if an agent $v$ enters the market at time $t_0$, she can stay until time $t_0+Y$ where $Y$ is an exponential-distribution random variable with mean 1.
We say that $v$ becomes \textit{critical} at time $t_0+X$.
An agent $v\in U_t\cup V_t$ leaves the market at time $t$ if $v$ becomes critical at time $t$ and/or $v$ is matched to another agent at time $t$. 
Thus $v$ leaves the market at some time $t_1$ where $t_0 \leq t_1 \leq t_0+X$.
We call $t_1-t_0$ the \textit{sojourn} of $v$, denoted $s(v)$.

Define $d_a=\lambda_a p$ and $d_b=\lambda_b p$ to be the \emph{density} of the market. 
Then a bipartite matching market is defined by the tuple $(\lambda_a, \lambda_b, p)$ or equivalently, $(d_a, d_b, p)$. 

%%%%%%%%%%%%%%%%%%%%%%%%%%%%%%%%%%%%%%%%%%%%%%%%%%%%%%%%%%
\paragraph{Matching Algorithms}
%%%%%%%%%%%%%%%%%%%%%%%%%%%%%%%%%%%%%%%%%%%%%%%%%%%%%%%%%%

A set of edges $M_t \subseteq E_t$ is a \emph{matching} if there are no edges that share the same vertices. 
A \emph{matching algorithm},  at any time $t$, selects a (possibly empty) matching $M_t$ in the current graph $G_t$, and the end vertices of the edges of $M_t$ leave the market immediately.
We assume that any matching algorithm, at any time $t_0$, only knows the current graph $G_t$ for $t\leq t_0$ and does not know the future information on $G_{t}$ for $t>t_0$.

We measure the performance of a matching algorithm by the \emph{loss} defined as follows.
Let $\mathrm{ALG}(T)$ be the set of agents matched by an algorithm $\mathsf{ALG}$ by time $T$.

%%added after EC
\begin{definition}[Loss of Matching Algorithms]
   The loss $\mathbf{L}_a$ of a matching algorithm $\mathsf{ALG}$ is defined by the ratio of the expected number of perished agents in $U$ over the expected number of agents in $U$, that is,   
   \begin{align}
       \mathbf{L}_a(\mathsf{ALG}) 
       = \frac{\mathbb{E}\left[|U-\mathrm{ALG}(T)\cap U - U_T|\right]}{\mathbb{E}[|U|]}
       = \frac{\mathbb{E}\left[|U-\mathrm{ALG}(T)\cap U - U_T|\right]}{\lambda_a T}.
   \end{align}
   The loss $\mathbf{L}_b(\mathsf{ALG})$ can be defined similarly.
   The total loss $\mathbf{L}(\mathsf{ALG})$ is then defined as 
   \begin{align*}
       \mathbf{L}(\mathsf{ALG}) 
       &=
       \frac{\mathbb{E}\left[|U\cup V-\mathrm{ALG}(T)- U_T\cup V_T|\right]}{(\lambda_a+\lambda_b)T}
       = \frac{\lambda_a\mathbf{L}_a(\mathsf{ALG}) + \lambda_b \mathbf{L}_b(\mathsf{ALG})}{\lambda_a+\lambda_b}.
%       = \frac{\lambda_a}{\lambda_a+\lambda_b}\mathbf{L}_a(\mathsf{ALG}) 
%        + \frac{\lambda_b}{\lambda_a+\lambda_b}\mathbf{L}_b(\mathsf{ALG}).
   \end{align*}   
\end{definition}

%\begin{definition}[Loss of Matching Algorithms]
%   The loss $\mathbf{L}$ of a matching algorithm $\mathsf{ALG}$ is defined by the ratio of the expected number of perished agents over the expected number of agents, that is,   
%   \begin{align}
%       \mathbf{L}(\mathsf{ALG}) 
%       = \frac{\mathbb{E}\left[|U\cup V-\mathrm{ALG}(T)- U_T\cup V_T|\right]}{\mathbb{E}[|U\cup V|]}
%       = \frac{\mathbb{E}\left[|U\cup V-\mathrm{ALG}(T)- U_T\cup V_T|\right]}{(\lambda_a+\lambda_b)T}.
%   \end{align}
%\end{definition}
By definition, the loss is in $[0, 1]$.
Minimizing the loss is equivalent to maximizing the size of matchings by the algorithm.
We note that we focus only on the cost of being unmatched and ignore the waiting cost, as we here assume that the latter is negligible compared to the former. 

%%%%%%%%%%%%%%%%%%%%%%%%%%%%%%%%%%%%%%%%%%%%%%%%%%%%%%%%%%
\paragraph{2-sided/1-sided Greedy and Patient Algorithms} 
%%%%%%%%%%%%%%%%%%%%%%%%%%%%%%%%%%%%%%%%%%%%%%%%%%%%%%%%%%

In this paper, we focus on \textit{local} algorithms, in the sense that they look only at the neighbors of agents which attempt to match, rather than at the global network structure.
Local algorithms can be designed by defining how each agent behaves to make a decision during her staying time.

We consider two simple agent behaviors with different timing properties, \textit{greedy} and \textit{patient}, defined as follows. 
We also introduce an agent who does nothing in the market.

%\begin{definition}
%\quad 
\begin{description}
   \item[Greedy agent:] An agent $v$ is \emph{greedy} if, immediately after $v$ enters the market at time $t$, the agent $v$ is matched to an arbitrary agent in $N_t(v)$ whenever $N_t(v)\neq \emptyset$. If $N_t(v)=\emptyset$, then $v$ does not make a decision.
    \item[Patient agent:] An agent $v$ is \emph{patient} if, at her critical time $t$, the agent $v$ is matched to an arbitrary agent in $N_t(v)$ whenever $N_t(v)\neq \emptyset$.
    \item[Inactive agent:] An agent $v$ is \emph{inactive} if $v$ does not make a decision in the market. That is, $v$ leaves the market if $v$ becomes critical and/or $v$ is asked by some neighbor to match during her staying time.
    \end{description}
%\end{definition}

In the \textit{2-sided} matching algorithms, 
we assume that all agents are homogeneous, greedy or patient.
In the \textit{1-sided} matching algorithms, we assume that all agents on one side are inactive and that either all agents on the other side are greedy or patient.
Our proposed mechanisms are summarized as follows.

\begin{description}
   \item[$\mathsf{Greedy}_2$:] Every agent is greedy.
   \item[$\mathsf{Patient}_2$:] Every agent is patient.
   \item[$\mathsf{Greedy}_1$:] Every agent in $V$ is greedy, while every agent in $U$ is inactive.
   \item[$\mathsf{Patient}_1$:] Every agent in $V$ is patient, while every agent in $U$ is inactive.
\end{description}

Note that $\mathsf{Patient}_i$ for $i=1,2$ can only be applied when the planner knows the departure information. 

%%%%%%%%%%%%%%%%%%%%%%%%%%%%%%%%%%%%%%%%%%%%%%%%%%%%%%%%%%
\paragraph{Optimal and Omniscient Algorithms}
%%%%%%%%%%%%%%%%%%%%%%%%%%%%%%%%%%%%%%%%%%%%%%%%%%%%%%%%%%

We compare the performance of a matching algorithm to that of an \textit{optimal} algorithm and an \textit{omniscient} algorithm.
Let $\mathsf{OPT}$ be the algorithm that minimizes the loss over a time period $T$, assuming that the algorithm does not know the future information~(or even the departure information).
The omniscient algorithm $\mathsf{OMN}$ is the algorithm that finds the maximum matching, provided  full information about the future, i.e., it knows the full realization of the graph $G_t$ for all $0\leq t\leq T$.
By definition, we easily observe that $\mathbf{L}(\mathsf{OPT}) \geq \mathbf{L}(\mathsf{OMN})$.
Moreover, we have, for any $i=1,2$, 
$\mathbf{L}(\mathsf{Greedy}_i) \geq \mathbf{L}(\mathsf{OPT})$
and $\mathbf{L}(\mathsf{Patient}_i) \geq \mathbf{L}(\mathsf{OMN})$.
%\[
%\mathbf{L}(\mathsf{Greedy}_i) \geq \mathbf{L}(\mathsf{OPT}) \quad \text{and}\quad 
%\mathbf{L}(\mathsf{Patient}_i) \geq \mathbf{L}(\mathsf{OMN}).
%\]

We remark that it does not make sense to define optimal and omniscient algorithms for the 1-sided markets.
In fact, if a 1-sided algorithm allows active agents to make decision at any time during her stay, then any 2-sided algorithm can be simulated as 1-sided algorithms. 
This is because the decision of her partner can be regarded as her decision. 
Thus optimal and omniscient algorithms for the 1-sided markets are identical with those in the 2-sided setting, respectively.
%Nevertheless, we can derive lower bounds on each of $\mathsf{Greedy}_1$ and $\mathsf{Patient}_1$.

\section{Main Results}\label{s3}

Our purpose is to evaluate the performance of the proposed algorithms to measure the impact of waiting time on the performance.
We estimate their losses for the case of large markets with sparse graphs in the steady state, that is, in the case when $\lambda_a, \lambda_b \to \infty$, $d_a$ and $d_b$ are kept constants, and $T\to\infty$.
In this section, we present our main theorems. 
Proof sketches for upper-bounding the loss will be described in the next section.
%See Appendices~\ref{sec:UB:Prelim}--\ref{sec:P1detail} for more detailed proofs.
%The proofs of lower bounds on the loss, i.e, Theorems~\ref{tt22} and~\ref{tt44}, may be found in Appendix~\ref{sec:LB}.

%%%%%%%%%%%%%%%%%%%%%%%%%%%%%%%%%%%%%%%%%%%%%%%%%%%%%%%%%%
\subsection{2-sided Matching Algorithms}
%%%%%%%%%%%%%%%%%%%%%%%%%%%%%%%%%%%%%%%%%%%%%%%%%%%%%%%%%%

For the 2-sided matching algorithms, we have the following upper bounds.

\begin{theorem}\label{t1g}
For a bipartite matching market $(d_a, d_b, p)$ with $d_a \geq d_b$, we have
\begin{align*}
    \mathbf{L}_a (\mathsf{Greedy}_2)  \leq \frac{d_a-d_b}{d_a} + \frac{\log{(d_b+3)}}{d_a}
    \quad\text{and}\quad
    \mathbf{L}_b (\mathsf{Greedy}_2)  \leq \frac{\log{(d_b+3)}}{d_b}
\end{align*}
when $\lambda_a, \lambda_b, T\to \infty$.
Therefore, when $\lambda_a, \lambda_b, T\to \infty$, it holds that
\begin{align}\label{eq:t1-1}
    \mathbf{L} (\mathsf{Greedy}_2) & \leq \frac{d_a-d_b}{d_a+d_b} + \frac{2\log{(d_b+3)}}{d_a+d_b}.
\end{align}
\end{theorem}

\begin{theorem}\label{t1p}
For a bipartite matching market $(d_a, d_b, p)$ with $d_a \geq d_b$, we have
\begin{align*}
    \mathbf{L}_a (\mathsf{Patient}_2)   \leq \frac{d_a-d_b}{d_a} + \frac{\log{(d_b+3)}}{d_a}
    \quad\text{and}\quad
    \mathbf{L}_b (\mathsf{Patient}_2)   \leq e^{-\max\left\{d_a - d_b, \frac{d_a}{1+d_b}\right\}}
\end{align*}
when $\lambda_a, \lambda_b, T\to \infty$.
Therefore, when $\lambda_a, \lambda_b, T\to \infty$, it holds that
\begin{align}
    \mathbf{L} (\mathsf{Patient}_2)  & \leq \frac{d_a-d_b}{d_a+d_b} + \frac{\log{(d_b+3)}}{d_a+d_b} + e^{-\max\left\{d_a - d_b, \frac{d_a}{1+d_b}\right\}}  \label{eq:t1-2}.
\end{align}
If $d_a = d_b$, then we have a better bound:
\begin{align*}
    \mathbf{L} (\mathsf{Patient}_2)  & \leq 2 e^{-Cd_a}  \text{\quad for some constant $C$}. %\label{eq:t1-2equal}
\end{align*}
\end{theorem}

We observe from Theorems~\ref{t1g} and~\ref{t1p} that, when $d_a = d_b$, 
$\mathbf{L} (\mathsf{Greedy}_2)  = O \left(\frac{\log{d_a}}{d_a}\right)$ and $\mathbf{L} (\mathsf{Patient}_2)  = e^{-O(d_a)}$.
%\[
%    \mathbf{L} (\mathsf{Greedy}_2)  = O \left(\frac{\log{d_a}}{d_a}\right)\quad \text{and}\quad 
%    \mathbf{L} (\mathsf{Patient}_2)  = e^{-O(d_a)}.
%\]
It indicates that the upper bound of $\mathbf{L}(\mathsf{Patient}_2)$ is exponentially smaller than that of $\mathbf{L}(\mathsf{Greedy}_2)$.
On the other hand, when $d_a > c d_b$ for some constant $c>1$, we see that $\mathbf{L}_a(\mathsf{Greedy}_2) \approx  \mathbf{L}_a(\mathsf{Patient}_2)$, while $\mathbf{L}_b(\mathsf{Patient}_2)$ is exponentially better than $\mathbf{L}_b(\mathsf{Greedy}_2)$.
%On the other hand, when $d_a \gg d_b$, the RHS of \eqref{eq:t1-1} becomes close to $1$, while the RHS of \eqref{eq:t1-2} remains decreasing exponentially as $d_a$ and $d_b$ grow. 

Furthermore, we provide lower bounds for optimal algorithms.

\begin{restatable}{theorem}{LBTwoSide}
%\begin{theorem}
\label{tt22}
For a bipartite matching market $(d_a, d_b, p)$ such that $d_a \geq d_b\geq 1$ and $p< 1/10$, it holds that
\begin{align}
    \mathbf{L} (\mathsf{OPT}) & \geq
    \max\left\{ \frac{1}{1+2d_a + d_b + 2d_a^2/\lambda_a + d_b^2/\lambda_b}, \frac{d_a - d_b}{d_a + d_b}\right\}, \label{e6}\\
    \mathbf{L} (\mathsf{OMN}) & \geq 
    \max\left\{ \frac{1}{2}\left(\frac{e^{-(d_a+d_ap)}}{1+d_a+d_a^2/\lambda_a}+\frac{e^{-(d_b+d_bp)}}{1+d_b+d_b^2/\lambda_b}\right), \frac{d_a - d_b}{d_a + d_b} \right\}. \label{e7}
\end{align}
%\end{theorem}
\end{restatable}
%We remark that Jiang~\cite{jiang2018bipartite} also gave lower bounds on $\mathbf{L} (\mathsf{OPT})$ and $\mathbf{L} (\mathsf{OMN})$.
%Compared to Jiang~\cite{jiang2018bipartite}, our bound of $\mathbf{L} (\mathsf{OPT})$ is of a simpler form with the same order, while that of $\mathbf{L} (\mathsf{OMN})$ is larger.
%See Section~\ref{s6} for more details.

It is not difficult to see that the loss $\frac{d_a - d_b}{d_a + d_b}$ is unavoidable in every market.
Since agents in two classes arrive at Poisson rates $\lambda_a$ and $\lambda_b$, the expected numbers of agents on both sides in the time interval $[0, T]$ are $\lambda_a T$ and $\lambda_b T$.
Then it cannot have a matching of size greater than $\lambda_b T$ when $\lambda_a \geq \lambda_b$, which implies that $(\lambda_a -\lambda_b) T$ agents in $U$ can never be matched during the process.
Thus the loss of the omniscient algorithm is at least $\frac{d_a - d_b}{d_a + d_b}$.
See Lemma~\ref{lem:DeltaLB} in Section~\ref{sec:LB}.

Since $\mathbf{L} (\mathsf{Greedy}_2)\geq \mathbf{L} (\mathsf{OPT})$ and $\mathbf{L} (\mathsf{Patient}_2)\geq \mathbf{L} (\mathsf{OMN})$, 
they give lower bounds for $\mathbf{L} (\mathsf{Greedy}_2)$ and $\mathbf{L} (\mathsf{Patient}_2)$, respectively.
When $\lambda_a$ and $\lambda_b$ are sufficiently large, \eqref{e6} in Theorem~\ref{tt22} implies that
$    \mathbf{L} (\mathsf{OPT})  \geq
    \frac{1}{2}\left(\frac{d_a-d_b}{d_a+d_b}+\frac{1}{1+2d_a + d_b}\right).
$
Thus, the upper bound of Theorem~\ref{t1g} is nearly optimal, up to logarithmic factors.
Moreover, by~\eqref{e7}, when $d_a=d_b$, we obtain
$
    \mathbf{L} (\mathsf{OMN})  \geq \frac{e^{-d_a}}{1+d_a}.
$
This implies that the upper bound when $d_a=d_b$ in Theorem~\ref{t1p} is nearly optimal.

%See Appendix~\ref{} for the proof of Theorem~\ref{tt22}.

\medskip
\noindent

\textbf{Remark.}
Jiang~\cite{jiang2018bipartite} gave lower bounds on $\mathbf{L} (\mathsf{OPT})$ and $\mathbf{L} (\mathsf{OMN})$.
Compared to his bound, our bound of $\mathbf{L} (\mathsf{OPT})$ is of a simpler form with the same order, while that of $\mathbf{L} (\mathsf{OMN})$ is better.
In fact, Jiang~\cite{jiang2018bipartite} showed that
\[
\mathbf{L} (\mathsf{OPT})\geq 
\frac{1- \log (1-p)\left(\lambda_a-\lambda_b\right)\frac{\lambda_a-\lambda_b}{\lambda_a+\lambda_b}}{1-\log(1-p)(\lambda_a+\lambda_b)}.
\]
%We will see that this bound is roughly the same as our bound~\eqref{e6}.
Using the fact that $-p-2p^2 \leq \log(1-p) \leq -p$ for $p\leq 1/2$, 
%Since $\log (1-p)=-p-\frac{1}{2(1-\alpha)^2}p^2$ for some $\alpha \in (0, p)$, we have $-p-2p^2 \leq \log(1-p) \leq -p$ for $p\leq 1/2$.
the RHS is at least $\frac{1}{1+(d_a+d_b)+2(d_a+d_b)p}$, 
%\[
%\frac{1 + p \left(\lambda_a-\lambda_b\right)\frac{\lambda_a-\lambda_b}{\lambda_a+\lambda_b}}{1+p(\lambda_a+\lambda_b) + 2p^2 (\lambda_a+\lambda_b)}
%\geq 
%\frac{1}{1+(d_a+d_b)+2(d_a+d_b)p}.
%\]
which is of the same order as~\eqref{e6}.
On the other hand, the lower bound of $\mathbf{L} (\mathsf{OMN})$ by Jiang~\cite{jiang2018bipartite} is 
\[
\mathbf{L}(\mathsf{OMN})
\geq 
\frac{1- \log (1-p)\left(\lambda_a-\lambda_b\right)\frac{\lambda_a-\lambda_b}{\lambda_a+\lambda_b} - (\log (1-p))^2\lambda_a\lambda_b
}{1-\log(1-p)(\lambda_a+\lambda_b)}.
\]
This is worse than our bound~\eqref{e7}, as, for example, 
when $\lambda_a = \lambda_b$, the numerator is equal to 
$1-\left(\frac{\log (1-p)}{p}d_a\right)^2$,
%\[
%1-(\log (1-p))^2\lambda^2_a = 1-\frac{(\log (1-p))^2}{p^2}d_a^2
%\]
which is negative when $d_a\geq 1$.

%%%%%%%%%%%%%%%%%%%%%%%%%%%%%%%%%%%%%%%%%%%%%%%%%%%%%%%%%%
\subsection{1-sided Matching Algorithms}
%%%%%%%%%%%%%%%%%%%%%%%%%%%%%%%%%%%%%%%%%%%%%%%%%%%%%%%%%%

For the 1-sided matching algorithms $\mathsf{Greedy}_1$ and $\mathsf{Patient}_1$, we obtain the same upper bounds of their losses as below.
Recall that we assume that agents with density $d_a$ are inactive.

\begin{theorem}\label{t3}
Let $\mathsf{ALG}_1$ be either $\mathsf{Greedy}_1$ or $\mathsf{Patient}_1$.
For a bipartite matching market $(d_a, d_b, p)$, it holds that, if $d_a\geq d_b$, then 
\begin{align*}
    \mathbf{L}_a (\mathsf{ALG}_1) \leq \frac{d_a - d_b}{d_a} + \frac{\log (d_b+3) }{d_a}
    \quad \text{and} \quad
    \mathbf{L}_b (\mathsf{ALG}_1) \leq \frac{\log (d_b+3) }{d_b}, 
%    \mathbf{L}_a (\mathsf{Greedy}_1) \leq \frac{d_a - d_b}{d_a} + \frac{\log (d_b+3) }{d_a}
%    \quad \text{and} \quad
%    \mathbf{L}_b (\mathsf{Greedy}_1) \leq \frac{\log (d_b+3) }{d_b}, \\
%    \mathbf{L}_a (\mathsf{Patient}_1) \leq \frac{d_a - d_b}{d_a} + \frac{\log (d_b+3) }{d_a}
%    \quad \text{and} \quad
%    \mathbf{L}_b (\mathsf{Patient}_1) \leq \frac{\log (d_b+3) }{d_b}, 
\end{align*}
and, if $d_a < d_b$, then
\begin{align*}
    \mathbf{L}_a (\mathsf{ALG}_1) \leq \frac{\log (d_b+3) }{d_a}
    \quad \text{and} \quad
    \mathbf{L}_b (\mathsf{ALG}_1) \leq \frac{d_b - d_a}{d_b} + \frac{\log (d_b+3) }{d_b}, 
%    \mathbf{L}_a (\mathsf{Greedy}_1) \leq \frac{\log (d_b+3) }{d_a}
%    \quad \text{and} \quad
%    \mathbf{L}_b (\mathsf{Greedy}_1) \leq \frac{d_b - d_a}{d_b} + \frac{\log (d_b+3) }{d_b}, \\
%    \mathbf{L}_a (\mathsf{Patient}_1) \leq \frac{\log (d_b+3) }{d_a}
%    \quad \text{and} \quad
%    \mathbf{L}_b (\mathsf{Patient}_1) \leq \frac{d_b - d_a}{d_b} + \frac{\log (d_b+3) }{d_b}, 
\end{align*}
when $\lambda_a, \lambda_b, T\to \infty$.
Therefore, when $\lambda_a, \lambda_b, T\to \infty$, we have
\begin{align*}
    \mathbf{L} (\mathsf{ALG}_1) &\leq \frac{|d_b - d_a|}{d_a+d_b} + \frac{2\log (d_b+3) }{d_a+ d_b}.
%    \mathbf{L} (\mathsf{Greedy}_1) &\leq \frac{|d_b - d_a|}{d_a+d_b} + \frac{2\log (d_b+3) }{d_a+ d_b},\\    
%    \mathbf{L} (\mathsf{Patient}_1) &\leq \frac{|d_b - d_a|}{d_a+d_b} + \frac{2\log (d_b+3) }{d_a+ d_b}.
\end{align*}
\end{theorem}

We can see that the bound of $\mathbf{L} (\mathsf{Greedy}_1) $ has the same order as that of $ \mathbf{L} (\mathsf{Greedy}_2)$, while the bound of $\mathbf{L} (\mathsf{Patient}_1) $ is worse than that of $\mathbf{L} (\mathsf{Patient}_2)$.
In fact, $\mathbf{L} (\mathsf{Patient}_1) $ is lower-bounded by the following theorem, which shows that $\mathbf{L} (\mathsf{Patient}_1) $ is strictly worse than $\mathbf{L} (\mathsf{Patient}_2)$ when $d_a=d_b$.
In addition, both of $\mathsf{Greedy}_1$ and $\mathsf{Patient}_1$ have the same performances.

\begin{restatable}{theorem}{LBOneSide}
%\begin{theorem}
\label{tt44}
For a bipartite matching market $(d_a, d_b, p)$ with $d_a, d_b\geq 1$ and $p < 1/10$, it holds that 
\begin{align*}
    \mathbf{L}(\mathsf{Greedy}_1) &\geq  
     \max\left\{\frac{1}{2(1+d_b+d^2_b/\lambda_b)}, \frac{|d_a-d_b|}{d_a+d_b} \right\},\\
%     \frac{1}{2} \left( \frac{1}{1+d_b+d^2_b/\lambda_b}+\frac{|d_a-d_b|}{d_a+d_b} \right),\\
    \mathbf{L}(\mathsf{Patient}_1) &\geq  
     \max\left\{\frac{\log (d_b+d^2_b/\lambda_b)}{d_a+d_b+d^2_a/\lambda_a+d^2_b/\lambda_b}, \frac{|d_a -d_b|}{d_a+d_b}\right\}.
\end{align*}
%\end{theorem}
\end{restatable}

\section{Performance Analysis: Proof Overview}\label{s66}

In this section, we describe the proof outline of Theorems~\ref{t1g}, \ref{t1p}, and~\ref{t3}.
See Section~\ref{sec:UB} for more detailed proofs, and the proofs of Theorems~\ref{tt22} and~\ref{tt44} may be found in Section~\ref{sec:LB}.
The proof outline is basically similar to that for a non-bipartite matching model by Akbarpour et al.~\cite{akbarpour2020thickness}, but it requires new techniques with more rigorous analysis.

As mentioned in Introduction, 
we first observe that dynamics of our proposed algorithms can be formulated as continuous-time Markov chains determined by the pool sizes $(A_t, B_t)$, that is, a pair of non-negative integers.
%The detailed definitions are omitted due to the page limit, but, as depicted in Figures~\ref{} and~\ref{}, a
%At any time $t$, the pool sizes $(A_t, B_t)$ are changed by at most $1$ for each, based on the current pool sizes.
%Examples of the transitions are depicted in Figure~\ref{p1main}. 
The chains are shown to have unique stationary distributions $\pi$.
It means that, in the long run, the distribution of the pool sizes is converged to $\pi$.
Moreover, we can see that the loss of the proposed algorithms can be expressed as the pool sizes, and therefore, it suffices to estimate the pool sizes in the steady state to bound the loss.

% removed in the full version
%\begin{figure}[t]
%\centering
%\subfigure[$\mathsf{Greedy}_1$/$\mathsf{Greedy}_2$]{\includegraphics[scale=0.25]{pictures/MCgreedy3.png}}
%\subfigure[$\mathsf{Patient}_1$/$\mathsf{Patient}_2$]{\raisebox{0.4mm}{\includegraphics[scale=0.25]{pictures/MCpatient3.png}}}
%\caption{Transitions of Markov chains for the proposed algorithms.}
%\label{p1main}
%\end{figure}

Our main contribution is to show that, for each of the proposed algorithms, the pool sizes in the steady state highly concentrate around some values.
The primary technical tool is the balance equations of Markov chains.
%The balance equation describes the probability flux in and out of a given set of states.
%Since our Markov chain is defined on 2-dimensional space, we need to apply balance equations in a careful way to find a concentration region.
%The region tells us the upper bounds on the loss of the algorithms.
%We characterize these values and derive upper bounds on the loss of the algorithms.
In the rest of this section, we will explain proof outlines for each algorithm in a bit more detail.

%The most challenging part is to find the concentration of the pool sizes in the steady state.
%The primary technical tool is the balance equations of Markov chains.
%The balance equation describes the probability flux in and out of a given set of states.
%For a non-bipartite matching model~\cite{akbarpour2020thickness}, a Markov chain is of a simple form on the set of non-negative integers, and hence we can naturally apply the balance equations.
%On the other hand, our Markov chain is defined on 2-dimensional space, i.e.,  each state is a pair of non-negative integers. 
%This requires us to choose a set of states for the balance equations more carefully.
%In fact, we need to adopt different strategies for each of the proposed algorithms.
%The balance equations yield recursive equations on the stationary distributions, and we can find  concentrated points by solving them.

%%%%%%%%%%%%%%%%%%%%%%%%%%%%%%%%%%%%%%%%%%%%%%%%%%%%%%%%%%%%%%%%%%%%%%%%
%\subsection{Analysis of the Steady State}

\subsection{2-sided Greedy Algorithm}\label{sec:G2_outline}

Let us first discuss the 2-sided Greedy algorithm $\mathsf{Greedy}_2$.
Recall that we may assume by symmetry that $\lambda_a\geq \lambda_b$.
When we run $\mathsf{Greedy}_2$, the bipartite graph $G_t=(U_t, V_t, E_t)$ almost always has no edges.
Since each agent's staying time follows the Poisson process with mean 1, the rate that some agent in $U_t$~(resp., $V_t$) becomes critical is $A_t$~(resp., $B_t$).
%Critical agents perish with probability one, which implies that the expected number of perished agents at time $t$ is $A_t + B_t$ in total.
Critical agents perish with probability one, which implies that the expected number of perished agents on each side at time $t$ is $\mathbb{E}[A_t]$ and $\mathbb{E}[B_t]$.
Since the distribution of $(A_t, B_t)$ converges to the stationary distribution $\pi$ of the corresponding Markov chain after the long run, $\mathbb{E}[A_t]$~(resp., $\mathbb{E}[B_t]$) is approximated by $\mathbb{E}_{(A, B)\sim\pi}[A]$~(resp., $\mathbb{E}_{(A, B)\sim\pi}[B]$) when $T\to \infty$.
%Similarly, $\mathbb{E}[B_t]$ is approximated by $\mathbb{E}_{(A, B)\sim\pi}[B]$ when $T\to \infty$.
Thus the losses $\mathbf{L}_a(\mathsf{Greedy}_2)$ and $\mathbf{L}_b(\mathsf{Greedy}_2)$ are roughly equal to $\frac{\mathbb{E}_{(A, B)\sim \pi}[A]}{\lambda_a}$ and $\frac{\mathbb{E}_{(A, B)\sim \pi}[B]}{\lambda_b}$, respectively, and the total loss $\mathbf{L}(\mathsf{Greedy}_2)$ is $\frac{1}{\lambda_a+\lambda_b}\mathbb{E}_{(A, B)\sim\pi}[A + B]$.
%which is the loss under the stationary distribution $\pi$.

In what follows, we obtain upper bounds of $\mathbb{E}_{(A, B)\sim\pi}[A]$ and $\mathbb{E}_{(A, B)\sim\pi}[B]$.
We show that, under $\mathsf{Greedy}_2$, the probability that the pool size $A_t$~(resp., $B_t$) in the steady state is larger than some value $k_2$~(resp., $\ell_2$) is small.
More specifically, we prove the following proposition.

\begin{restatable}{proposition}{GTwoConcentrate}
%\begin{proposition}
\label{prop:G2:concentration}
There exist $k_2$ and $\ell_2$, where $k_2 \leq  \lambda_a-\lambda_b+\frac{\lambda_a\log (d_b+3)}{d_a}$ and $\ell_2 \leq \frac{\lambda_b \log (d_b+3)}{d_b}$, such that, for any $\sigma \geq 1$, we have
\[
    \Pr_{(A, B)\sim \pi}[A \geq k_2 + \sigma +1] \leq O(\lambda_a) e^{-\frac{\sigma^2}{\sigma+\lambda_a}}
    \quad \text{and}\quad
    \Pr_{(A, B)\sim \pi}[B \geq \ell_2 + \sigma +1] \leq O(\lambda_b) e^{-\frac{\sigma^2}{\sigma+\lambda_b}}.
\]
%\end{proposition}
\end{restatable}

\begin{figure}[t]
    \centering
    \includegraphics[scale=0.2]{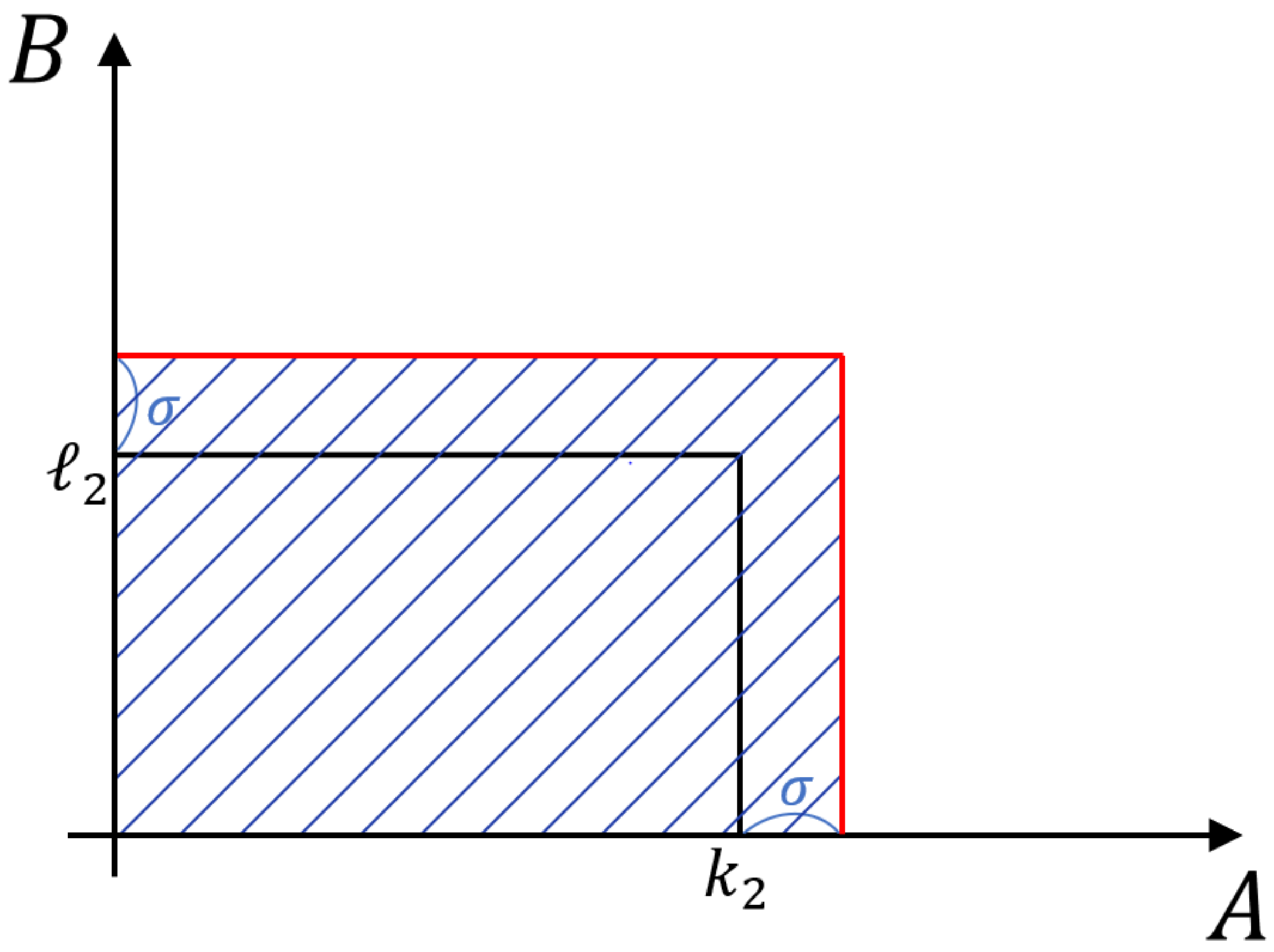}
    \caption{Concentration region for the pool sizes of $\mathsf{Greedy}_2$ in the steady state.}
    \label{fig:g2concentration}
\end{figure}

The first inequality of the above proposition says that the probability that $A$ is larger than $k_2$ drops exponentially.
For example, if we set $\sigma=\Theta (\sqrt{\lambda_a \log \lambda_a})$, we see that $\Pr_{(A, B)\sim \pi}[A \geq k_2 + \sigma +1]$ is constant.
Therefore, the pool sizes are concentrated with high probability in the region depicted as in Figure~\ref{fig:g2concentration}.

Proposition~\ref{prop:G2:concentration} implies that $\mathbb{E}_{(A, B)\sim\pi}[A] \leq k_2 + o(\lambda_a)$ by setting $\sigma=\Theta (\sqrt{\lambda_a \log \lambda_a})$.
Since $\mathbf{L}_a(\mathsf{Greedy}_2)\approx \frac{1}{\lambda_a}\mathbb{E}_{(A, B)\sim\pi}[A]$, this implies that 
\[
\mathbf{L}_a(\mathsf{Greedy}_2) \leq \frac{1}{\lambda_a}\left(\lambda_a-\lambda_b+\frac{\lambda_a \log (d_b+3)}{d_a}+o(\lambda_a)\right) = \frac{d_a-d_b}{d_a}+\frac{\log (d_b+3)}{d_a}+o(1).
\]
Similarly, by setting $\sigma=\Theta (\sqrt{\lambda_b \log \lambda_b})$, we have $\mathbb{E}_{(A, B)\sim\pi}[B] \leq \ell_2 + o(\lambda_b)$, and hence we obtain $\mathbf{L}_b(\mathsf{Greedy}_2)\approx \frac{1}{\lambda_b}\mathbb{E}_{(A, B)\sim\pi}[B] \leq \frac{\log (d_b+3)}{d_b}+o(1)$.
%\[
%\mathbf{L}_b(\mathsf{Greedy}_2)\approx \frac{1}{\lambda_b}\mathbb{E}_{(A, B)\sim\pi}[B] \leq \frac{\log (d_b+3)}{d_b}.
%\]
This shows Theorem~\ref{t1g}.
See Section~\ref{sec:UB:G2} for the details.

Here is an intuition behind the values $k_2$ and $\ell_2$.
In the unit-time interval, $\lambda_a$~(resp., $\lambda_b$) new agents in $U$~(resp. $V$) enter the market in expectation, and $\mathsf{Greedy}_2$ attempts to make a matching between them.
We observe that the size of a maximum matching in the time interval is at most $\lambda_b$, and hence at least $\lambda_a - \lambda_b$ agents in $U$ cannot get matched.
Thus, it is unavoidable that at least $\lambda_a - \lambda_b$ agents in $U$ enter the market in expectation.
This is an intuitive reason why $\mathbb{E}_{(A, B)\sim\pi} [A]$ is larger than $\mathbb{E}_{(A, B)\sim\pi} [B]$ by $\lambda_a - \lambda_b$.
The number of the other agents in the market is $O\left(\frac{\lambda_b \log d_b}{d_b}\right)$, which is small, compared to the case when all agents are inactive, in which the expected pool size is $(\lambda_a, \lambda_b)$.
This is because the Greedy algorithm matches agents as soon as possible, which reduces the number of agents in the pool.

\subsection{2-sided Patient Algorithm}

Under the 2-sided Patient algorithm $\mathsf{Patient}_2$, 
%the proof outline is similar to the case of $\mathsf{Greedy}_2$.
%Under $\mathsf{Patient}_2$, 
conditioned on $A_t$ and $B_t$, the graph $G_t$ is a random bipartite graph with vertex sets $A_t$ and $B_t$ where an edge is formed with probability $p$. The rate that some agent in $U_t$~(resp., $V_t$) becomes critical is $A_t$~(resp., $B_t$). 
Since $G_t$ is a random bipartite graph, a critical agent in $U_t$~(resp., $V_t$) perishes with probability $(1-p)^{B_t}$~(resp., $(1-p)^{A_t}$). 
Therefore, the expected numbers of perished agents on both sides at time $t$ are $\mathbb{E}[A_t(1-p)^{B_t}]$ and $\mathbb{E}[B_t(1-p)^{A_t}]$, respectively.
Since the distribution of $(A_t, B_t)$ converges to the stationary distribution $\pi$ in the long run, they are approximated by $\mathbb{E}_{(A, B)\sim\pi}[A(1-p)^{B}]$ and $\mathbb{E}_{(A, B)\sim\pi}[B(1-p)^{A}]$ when $T\to \infty$.
Thus $\mathbf{L}_a(\mathsf{Patient}_2)$ and $\mathbf{L}_b(\mathsf{Patient}_2)$ are roughly equal to $\frac{1}{\lambda_a}\mathbb{E}_{(A, B)\sim \pi}\left[A(1-p)^{B}\right]$ and $\frac{1}{\lambda_b}\mathbb{E}_{(A, B)\sim \pi}\left[ B(1-p)^{A}\right]$, respectively.

We show that, under $\mathsf{Patient}_2$, the pool sizes $(A, B)$ in the steady state are roughly between $k_2\leq A\leq \lambda_a$ and $\ell_2\leq B\leq \lambda_b$ with high probability, where we recall that $k_2$ and $\ell_2$ are the values found in Proposition~\ref{prop:G2:concentration}.
%highly concentrated around $(\lambda_a, \lambda_b)$.
More formally, we show the following, saying that the probability that the pool sizes are out of the region $S$ drops exponentially.
See Figure~\ref{fig:p2Left} for your help.

\begin{restatable}{proposition}{PTwoConcentrate}
%\begin{proposition}
\label{prop:2sideP_notS}
For any $\sigma_a, \sigma_b\geq 1$, there exist $\ukT$ and $\ulT$, where $k_2 - \sigma_b\leq \ukT \leq k_2$ and $\ell_2 - \sigma_a\leq \ulT \leq \ell_2$, such that
%, for any $\sigma'_a, \sigma'_b\geq 1$, we have
\[
  \Pr_{(A, B)\sim \pi} \left[ (A, B) \not\in S \right]
  \leq 
  O(\lambda_a\lambda^2_b) e^{-\frac{\sigma^2_b}{2(\sigma_b+\lambda_b)}}
  +O(\lambda^2_a\lambda_b) e^{-\frac{\sigma^2_a}{2(\sigma_a+\lambda_a)}},
  \]
where $S=\{(i,j)\mid \ukT - \sigma_a \leq i\leq \lambda_a+\sigma_a, \ulT - \sigma_b \leq j \leq \lambda_b+\sigma_b\}$.
\end{restatable}

By setting $\sigma_a = \Theta (\sqrt{\lambda_a \log \lambda_a})$ and $\sigma_a = \Theta (\sqrt{\lambda_b \log \lambda_b})$ in Proposition~\ref{prop:2sideP_notS}, we can upper-bound $\mathbb{E}_{(A, B)\sim\pi}[A(1-p)^{B}]$ and $\mathbb{E}_{(A, B)\sim\pi}[B(1-p)^{A}]$, which shows the first part of Theorem~\ref{t1p}.

%It follows from the proposition that, if we set $\sigma_a = \Theta (\sqrt{\lambda_a \log \lambda_a})$ and $\sigma_a = \Theta (\sqrt{\lambda_b \log \lambda_b})$, we have $\mathbb{E}_{(A, B)\sim\pi}[A(1-p)^{B}] \leq \lambda_a (1+o(1))(1-p)^{\ell_2}$, which shows Theorem~\ref{t1p} as $\mathbf{L}_a(\mathsf{Patient}_2) \approx \frac{1}{\lambda_a}\mathbb{E}_{(A, B)\sim\pi}[A(1-p)^{B}]$.
%% and $\mathbb{E}_{(A, B)\sim\pi}[B(1-p)^{A}] \leq \lambda_b (1+o(1))(1-p)^{k_2}$

We remark that the expected pool size in the steady state is larger than about $(k_2, \ell_2)$ with high probability, which means that the pool sizes are larger than those of the 2-sided Greedy algorithm.
In other words, waiting to match makes the market thicker after the long run.

When $d_a=d_b$, the bounds can be improved by further narrowing the concentration region.
We first show that the probability that the sum of the pool sizes $A_t+B_t$ in the steady state is at most about $(\lambda_a+\lambda_b)/2$ is small.
Under $\mathsf{Patient}_2$, $\lambda_a+\lambda_b$ agents enter the pool in unit time interval, and at most $A_t+B_t$ agents leave the pool.
Thus, if $A_t+B_t$ is much smaller than $\lambda_a+\lambda_b$, then the number of entering agents is larger than that of leaving agents, which increases the pool size.
Hence, in the steady state, the sum of the pool sizes cannot be small.

\begin{restatable}{proposition}{PTwoConcentrateSum}
%\begin{proposition}
\label{prop:P2balancedSum}
  For any $\sigma\geq 1$, it holds that
  \begin{align*}
  \Pr_{(A, B)\sim \pi} \left[ A + B \leq \left(\frac{\lambda_a+\lambda_b}{2} -2\right)- \sigma -1\right]&\leq O\left(\lambda_a+\lambda_b\right) e^{-\frac{\sigma^2}{4(\lambda_a+\lambda_b)}}.
  \end{align*}
%\end{proposition}
\end{restatable}

Moreover, we prove that the probability that the difference of the pool sizes $A_t-B_t$ in the steady state is at least $\lambda_a/2$ is small.
Formally, we prove the following.

\begin{restatable}{proposition}{PTwoConcentrateBalanced}
%\begin{proposition}
\label{prop:P2balanced}
  Suppose that $d_a=d_b\geq 3$ and $p<1/10$.
  Suppose that $\sigma_a$ satisfies that $1\leq \sigma_a\leq \lambda_a$. 
%  We assume that $\lambda_a$ is sufficiently large so that $\olambda_a \leq 2\lambda_a$. %\dznote{Is $\olambda_a$ just a value smaller than $2\lambda_a$?}
  For any $\sigma_d\geq 1$, it holds that
  \begin{align*}
  \Pr_{(A, B)\sim \pi} \left[ A - B \geq \frac{\lambda_a+\sigma_a}{2} + \sigma_d \right]
  &\leq \frac{e^{-c_d\sigma_d}}{1 - e^{-c_d}} + O\left(\lambda_a\right) e^{-\frac{\sigma_a^2}{\sigma_a +\lambda_a}}.
  \end{align*}
  for some constant $0<c_d<1$.
%\end{proposition}
\end{restatable}

%See Propositions~\ref{prop:P2balancedSum} and~\ref{prop:P2balanced} in Section~\ref{sec:UB:P2} for the precise statements.
Therefore, we can see that the probability that the pool sizes are out of some region $S'$ is small~(see Figure~\ref{fig:p2Right}), which allows us to upper-bound $\mathbb{E}_{(A, B)\sim\pi}[A(1-p)^{B}]$ and $\mathbb{E}_{(A, B)\sim\pi}[B(1-p)^{A}]$, which proves the second part of Theorem~\ref{t1p}.

\begin{figure}[t]
\centering
\subfigure[The unbalanced case]{\raisebox{0mm}{\includegraphics[width=51mm]{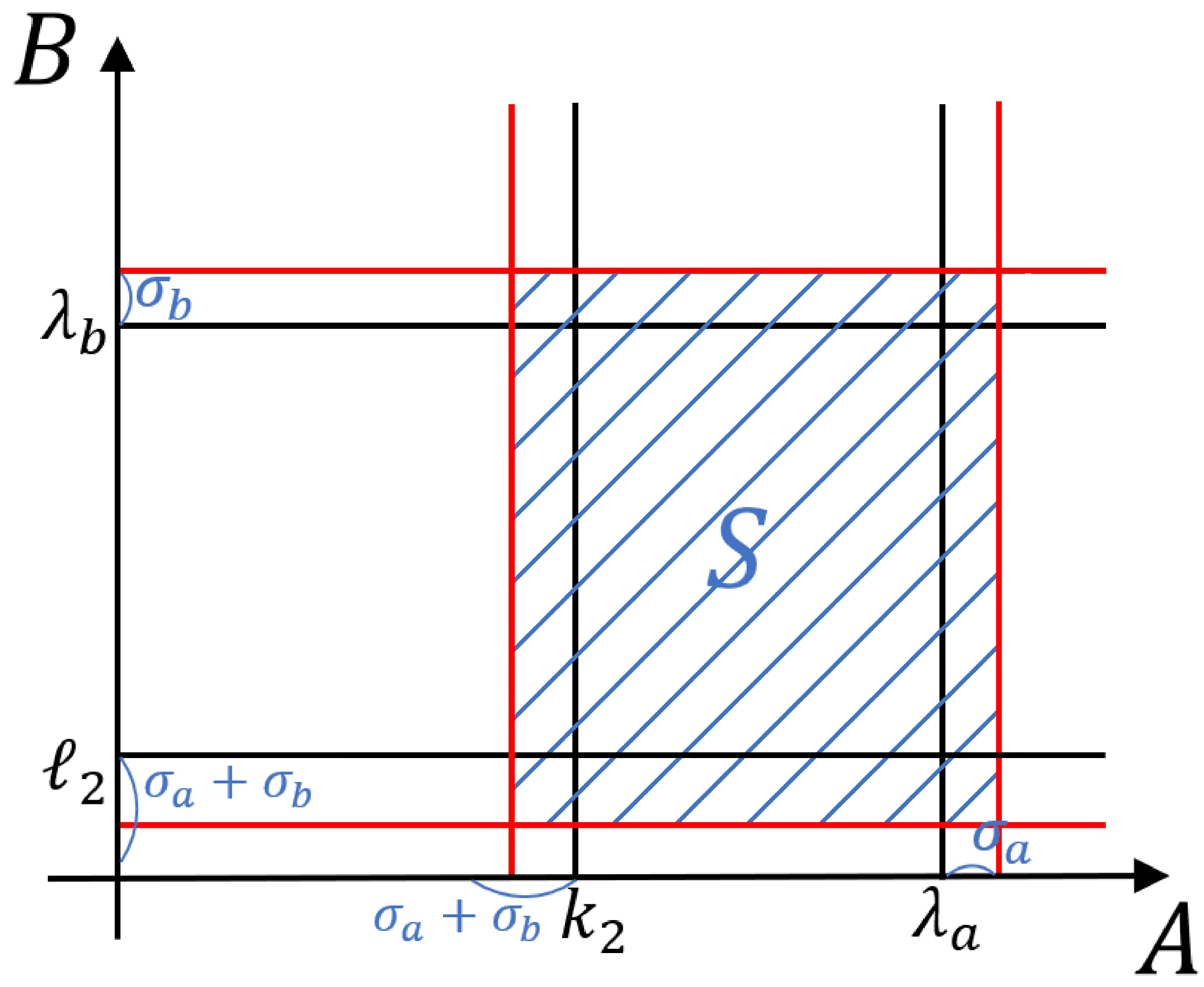}}\label{fig:p2Left}}
\hspace{3mm}
\subfigure[The balanced Case]{\includegraphics[width=44.5mm]{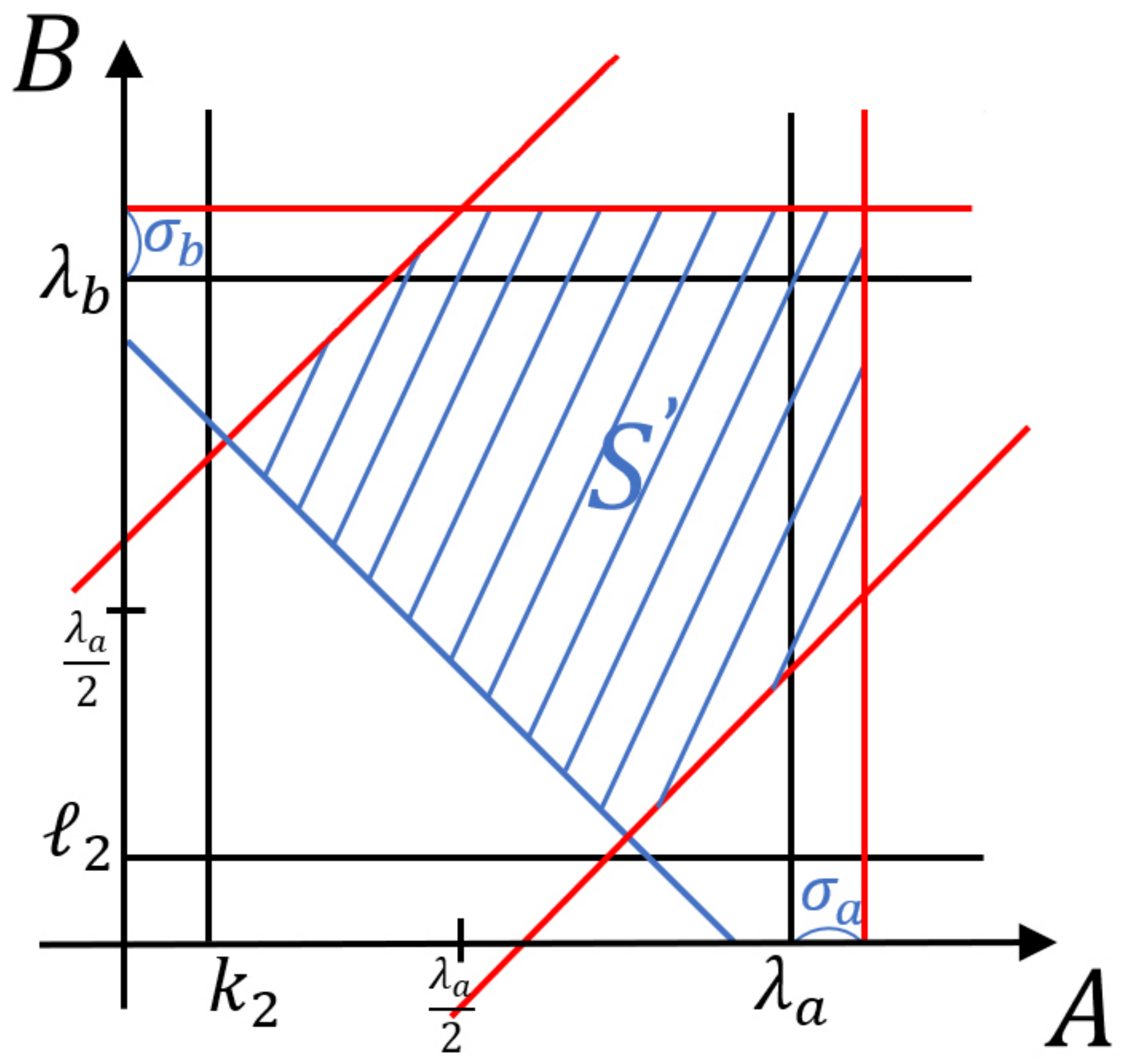}\label{fig:p2Right}}
\caption{Concentration region for $\mathsf{Patient}_2$.}
\label{fig:p2}
\end{figure}

\begin{comment}
\begin{figure}[t]
    \begin{tabular}{cc}
      \begin{minipage}[t]{0.45\hsize}
    \centering
    \includegraphics[width=55mm]{PicturesConcentra/Patient2Unbalan.PNG}
%    \caption{Concentration region for $\mathsf{Patient}_2$ in the unbalanced case}
%    \label{fig:p2unbalanced}
      \end{minipage} &
      \begin{minipage}[t]{0.45\hsize}
    \centering
    \includegraphics[width=57mm]{PicturesConcentra/Patient2Balan.PNG}
%    \caption{Concentration region for $\mathsf{Patient}_2$ in the balanced case}
%    \label{fig:p2balanced}
      \end{minipage}
    \end{tabular}
    \caption{Concentration region for $\mathsf{Patient}_2$. The left figure shows the unbalanced case, while the right one is the balanced case.}
    \label{fig:p2}
\end{figure}
\end{comment}

\subsection{1-sided Matching Algorithms}

The proof outline is similar to the 2-sided case, but 1-sided algorithms require a more involved analysis to show the concentration.
We remark that we need to solve recursive equations with additive terms.

In a similar way to the 2-sided case, 
we can observe that $\mathbf{L}(\mathsf{Greedy}_1)$ is approximated by $\frac{1}{\lambda_a+\lambda_b}\mathbb{E}_{(A, B)\sim\pi}[A+ B]$, and
$\mathbf{L}(\mathsf{Patient}_1)$ is by $\frac{1}{\lambda_a+\lambda_b}\mathbb{E}_{(A, B)\sim\pi}[A+ B(1-p)^{A}]$, where $\pi$ is the stationary distribution of the corresponding Markov chains.

For $\mathsf{Greedy}_1$, we show that the pool size $A_t$ of inactive agents in the steady state is highly concentrated around some value $k_1$, where $k_1 \leq \max\{\lambda_a - \lambda_b, 0\}+\frac{\lambda_b \log (d_b+3)}{d_b}$.
On the other hand, letting $\ell_1 = (1-p)^{-\sigma_a}(\lambda_b - \lambda_a + k_1)$ for $\sigma_a\geq 1$, the probability that the pool size $B_t$ of greedy agents in the steady state is larger than $\ell_1$ is small~(See Figure~\ref{fig:g1p1Left}).
The formal statement may be found as below, which proves Theorem~\ref{t3} for $\mathsf{Greedy}_1$.
%in Proposition~\ref{prop:G1} in Section~\ref{sec:G1detail}.
%This fact proves Theorem~\ref{t3} for $\mathsf{Greedy}_1$.
%More formally, we prove the following proposition.

\begin{restatable}{proposition}{GOneConcentrate}
%\begin{proposition}
\label{prop:G1}
There exists a number $k_1$, where $k_1 \leq \max\{\lambda_a - \lambda_b, 0\}+\frac{\lambda_b \log (d_b+3)}{d_b}$, such that, for any $\sigma_a\geq 1$, 
\begin{align*}
\Pr_{(A, B)\sim \pi} \left[ A \geq k_1 + \sigma_a+1 \right]\leq O\left(\lambda_a\right) e^{-\frac{\sigma_a^2}{\sigma_a+\lambda_a}},\quad
\Pr_{(A, B)\sim \pi} \left[ A \leq k_1 - \sigma_a-1 \right]\leq O(\lambda_a) e^{-\frac{\sigma_a^2}{\lambda_a}},
\end{align*}
and, for any $\sigma_a, \sigma_b\geq 1$, we have
\begin{align*}
\Pr_{(A, B)\sim \pi} \left[ B \geq \ell_1 + \sigma_b +1 \right]
&\leq O(\lambda_b) e^{-\frac{\sigma_b^2}{2(\sigma_b+\lambda_b)}} + O(\lambda_a\lambda^2_b) e^{-\frac{\sigma_a^2}{\lambda_a}},
\end{align*}
where $\ell_1 = (1-p)^{-\sigma_a}(\lambda_b -\lambda_a+k_1)$.
%\end{proposition}
\end{restatable}

%By setting $\sigma_a = \Theta (\sqrt{\lambda_a \log \lambda_a})$ and $\sigma_b = \Theta (\sqrt{\lambda_b \log \lambda_b})$, 
%Proposition~\ref{prop:G1} shows Theorem~\ref{t3} for $\mathsf{Greedy}_1$. 
%See Appendix~\ref{sec:G1detail} for the details.
%%If we set , Proposition~\ref{prop:G1} implies that
%%\[
%%\mathbf{L}(\mathsf{Greedy}_1) \approx \frac{1}{\lambda_a+\lambda_b}\mathbb{E}_{(A, B)\sim\pi}[A + B] 
%%\approx \frac{1}{\lambda_a+\lambda_b}\left(k_1 + \ell_1 + \sigma_a +\sigma_b\right).
%%\]
%%Since $\sigma_a = \Theta (\sqrt{\lambda_a \log \lambda_a})$, it follows that $(1-p)^{-\sigma_a}$ is a constant $C$.
%%Hence $\ell_1=Ck_1$ if  $\lambda_a \geq \lambda_b$ and $\ell_1 = C\left((\lambda_b - \lambda_a) + O\left(\frac{\lambda_b \log d_b}{d_b}\right)\right)$ otherwise.

\begin{figure}[t]
\centering
\subfigure[$\mathsf{Greedy}_1$]{\raisebox{-0.3mm}{\includegraphics[width=53mm]{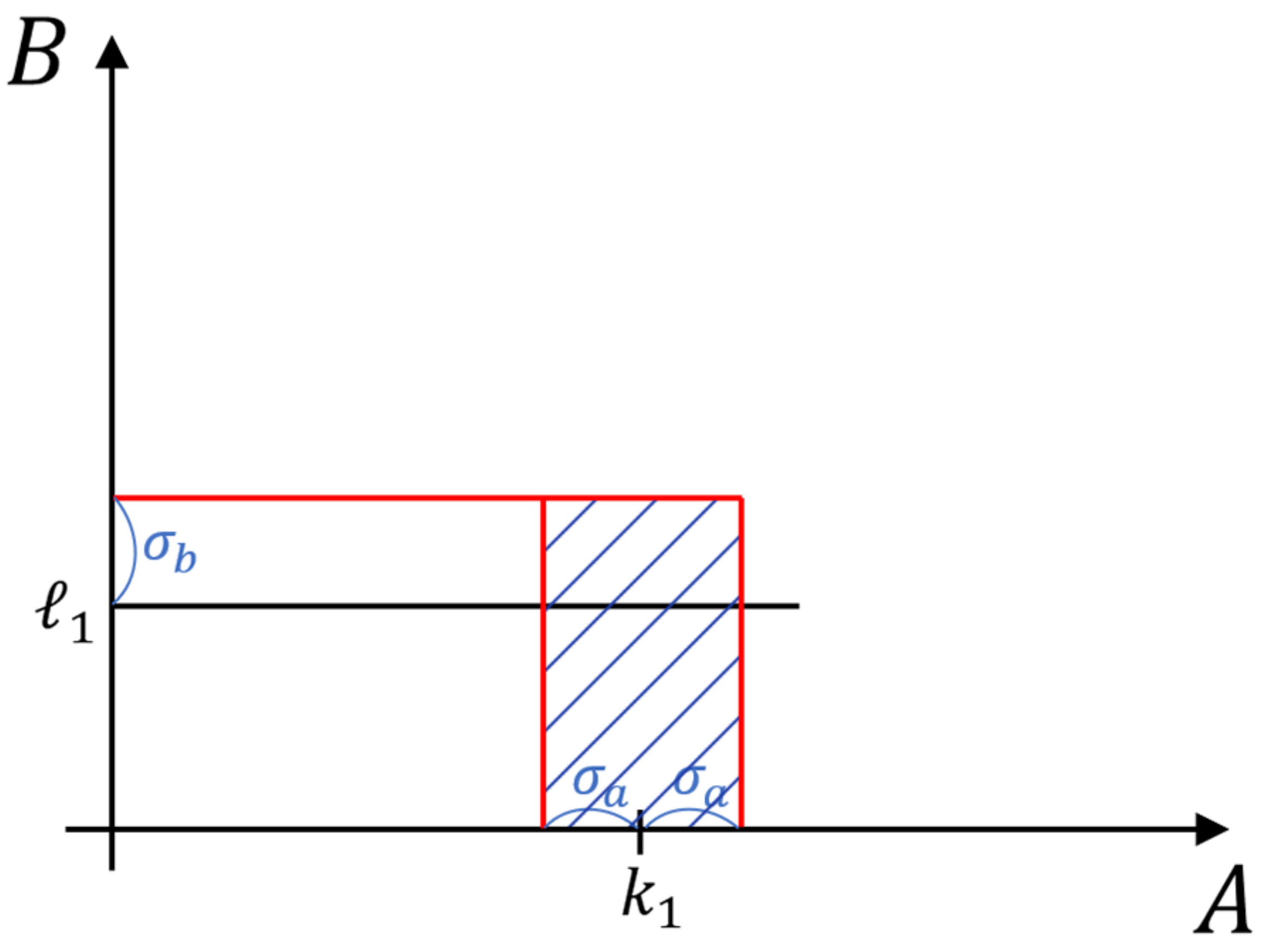}\label{fig:g1p1Left}}}
\subfigure[$\mathsf{Patient}_1$]{\includegraphics[width=53mm]{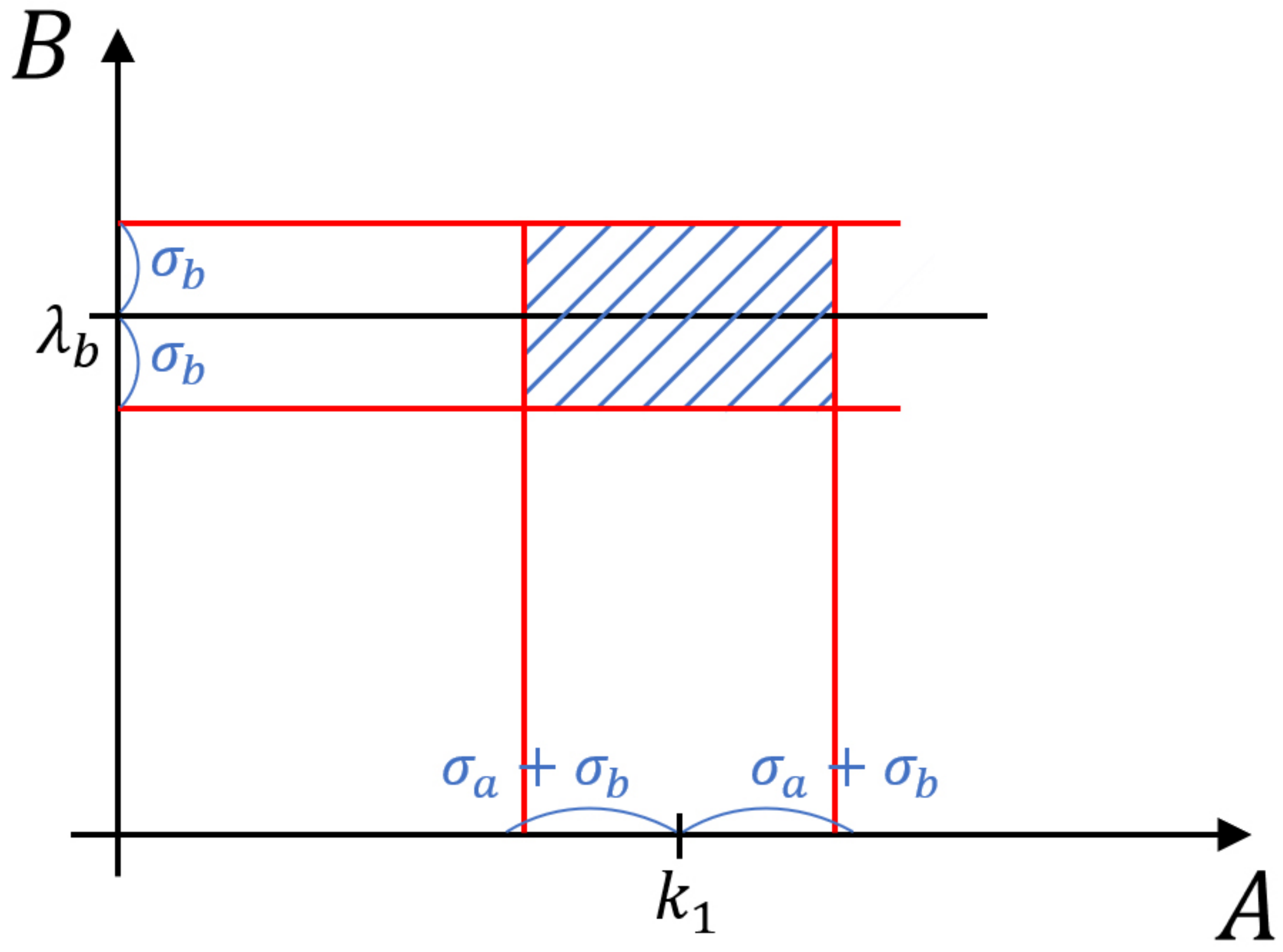}\label{fig:g1p1Right}}
\caption{Concentration regions for $\mathsf{Greedy}_1$ and $\mathsf{Patient}_1$.}
\label{fig:g1p1}
\end{figure}

\begin{comment}
\begin{figure}[t]
    \begin{tabular}{cc}
      \begin{minipage}[t]{0.45\hsize}
    \centering
    \includegraphics[width=55mm]{PicturesConcentra/Greedy1.PNG}
      \end{minipage} &
      \begin{minipage}[t]{0.45\hsize}
    \centering
    \includegraphics[width=55mm]{PicturesConcentra/Patient1.PNG}
      \end{minipage}
    \end{tabular}
    \caption{The left figure shows the concentration region for $\mathsf{Greedy}_1$, while the right one is for $\mathsf{Patient}_1$.}
    \label{fig:g1p1}
\end{figure}
\end{comment}

%\medskip

For $\mathsf{Patient}_1$, we prove the following, asserting that the pool size $(A_t, B_t)$ in the steady state is highly concentrated around $(k_1, \lambda_b)$, where $k_1$ is defined as above for $\mathsf{Greedy}_1$.
See Figure~\ref{fig:g1p1Right}.
%The formal statement may be found in Proposition~\ref{prop:1sideP_notS2} in Section~\ref{sec:P1detail}.
This fact proves Theorem~\ref{t3} for $\mathsf{Patient}_1$.

\begin{restatable}{proposition}{POneConcentrate}
%\begin{proposition}
\label{prop:1sideP_notS2}
For any $\sigma_a, \sigma_b\geq 1$, 
there exist $\uk$ and $\ok$ such that  $k_1-\sigma_b < \uk < k_1 < \ok < \min\{k_1+\sigma_b, \lambda_a\}$ and it holds that
\begin{align}\label{eq::1sideP_notS2}
\Pr_{(A, B)\sim \pi} \left[ (A, B) \not\in S \right]
&\leq 
O(\lambda_a) e^{-\frac{\sigma_a^2}{2(\sigma_a+\lambda_a)}}  + O(\lambda^2_a\lambda_b+\lambda^3_b) e^{-\frac{\sigma_b^2}{\sigma_b+\lambda_b}},
\end{align}
where $S=\{(k,j)\mid \uk-\sigma_a \leq k \leq \ok+\sigma_a, \lambda_b - \sigma_b\leq j \leq \lambda_b + \sigma_b\}$.
%\end{proposition}
\end{restatable}

Interestingly, for both of $\mathsf{Greedy}_1$ and $\mathsf{Patient}_1$, the size of inactive agents has the same concentrated value $k_1$.

%By Proposition~\ref{prop:1sideP_notS2}, Theorem~\ref{t3} for $\mathsf{Patient}_1$ holds.
%See Appendix~\ref{sec:p1_loss}.

%If we set $\sigma_a = \Theta (\sqrt{\lambda_a \log \lambda_a})$ and $\sigma_b = \Theta (\sqrt{\lambda_b \log \lambda_b})$, the above proposition implies that
%\begin{align*}
%\mathbf{L}(\mathsf{Patient}_1) &\approx \frac{1}{\lambda_a+\lambda_b}\mathbb{E}_{(A, B)\sim\pi}[A + B(1-p)^A] \\
%&\leq \frac{1}{\lambda_a+\lambda_b}\left(k_1 +o(\lambda_b)+ \lambda_b (1-p)^{k_1-o(\lambda_b)}\right),
%\end{align*}
%which shows the second part of Theorem~\ref{t3}.

%\appendix
\section{Stationary Distribution of Markov Chains}\label{s92}

In this section, we formulate our stochastic models as continuous-time Markov chains determined by the pool sizes, and show that the chains have unique stationary distributions.

%%%%%%%%%%%%%%%%%%%%%%%%%%%%%%%%%%%%%%%%%%%%%%%%%%%%%%%%%%%%%%%%%%%%%%%%
\subsection{Formulation by Markov Chains}\label{sec:markov}
%%%%%%%%%%%%%%%%%%%%%%%%%%%%%%%%%%%%%%%%%%%%%%%%%%%%%%%%%%%%%%%%%%%%%%%%

We first review continuous-time Markov chains.
We refer to Norris et al.~\cite{norris1998markov} for the details.
Let $Z_t$ be a continuous-time Markov chain on the ground set $\Omega$.
For any two states $i, j\in \Omega$, we denote by $r_{i\to j}$ the transition rate from $i$ to $j$.
Then $r_{i\to j}\geq 0$. 
The rate matrix $Q \in \mathbb{R}^{\Omega\times \Omega}$ is defined as
\[
Q_{ij} =
\begin{cases}
r_{i\to j} & \text{if } i\neq j\\
\sum_{k\neq i} - r_{i\to k} & \text{otherwise}.
\end{cases}
\]
Then the transition probability in $t$ units of time is $e^{tQ}=\sum_{i=0}^\infty \frac{t^iQ^i}{i!}$.
Let $P_t = e^{tQ}$ be the transition probability matrix of the Markov chain in $t$ time units.
We say that a distribution $\pi: \Omega \to \mathbb{R}_+$ is \textit{stationary distribution} of the Markov chain if $P_t=\pi P_t$ for any $t\geq 0$.

%%%%%%%%%%%%%%%%%%%%%%%%%%%%%%%%%%%%%%%%%%%%%%%%%%
%\subsection{Markov Chains for Our Algorithms}
%%%%%%%%%%%%%%%%%%%%%%%%%%%%%%%%%%%%%%%%%%%%%%%%%%

For each of our algorithms, we define a Markov chain where each state corresponds to the pool sizes on both sides.
That is, each state is a pair of non-negative integers.
Specific transitions for each algorithm can be found in Sections~\ref{sec:UB:G2}--\ref{sec:P1detail}.
We here observe that, for each of our algorithms, the pool sizes $(A_t, B_t)$ are Markovian.

Let us first consider the 2-sided Greedy algorithm $\mathsf{Greedy}_2$. 
Since a new arriving agent matches to someone immediately~(if possible),
the graph $G_t$ is a graph with no edges for any $t \geq 0$.
This means that the current pool sizes $(A_t, B_t)$ characterize the market at time $t$, and, conditioned on $A_t$ and $B_t$, $G_t$ is conditionally independent of $A_{t'}$ and $B_{t'}$ for $t'<t$.
Hence the pool sizes $(A_t, B_t)$ are Markovian.
%the distribution of $G_t$ is uniquely defined. Thus, given $A_t$ and $B_t$, $G_t$ is conditionally independent of $A_{t'}$ and $B_{t'}$ for $t'<t$.

For the 1-sided Greedy algorithm $\mathsf{Greedy}_1$, conditioned on $A_t$ and $B_t$, $G_t$ \textit{does} depend on $A_{t'}$ and $B_{t'}$ for $t'<t$.
In fact, the existence of edges depends on when an agent arrives.
However, the pool size $(A_t, B_t)$ itself is Markovian.
This is because a new agent chooses a partner only from her neighbors, which does not depend on the existence of other edges.
Moreover, we will see later that maintaining the pool sizes is enough to estimate the loss of $\mathsf{Greedy}_1$.

We next discuss the 2-sided and 1-sided Patient algorithms $\mathsf{Patient}_i$ for $i=1,2$.
We observe that, conditioned on $A_t$ and $B_t$ for $t\geq 0$, the graph $G_t$ is a random bipartite graph with vertex sets $U_t$ and $V_t$ such that edges between them are formed with probability $p$.
Indeed, a critical agent chooses a partner only from her neighbors, and hence, after the critical agent leaves, the rest of the graph still has the same distribution.
Therefore, conditioned on $A_t$ and $B_t$, the distribution of $G_t$ is uniquely determined, and hence the pool sizes $(A_t, B_t)$ are Markovian.

The above discussion is formally stated as follows.
Moreover, we show that the corresponding Markov chains have unique stationary distributions.

\begin{restatable}{theorem}{markov}
%\begin{theorem}
\label{t5}
For $\mathsf{Greedy}_i$ and $\mathsf{Patient}_i$ for $i=1,2$ and any $0 \leq t_0 \leq t_1$, it holds that
\begin{align}\label{eq:Markov}
    \Pr[(A_{t_1}, B_{t_1}) \mid (A_t, B_t)\ for\ 0 \leq t < t_1]=\Pr[ (A_{t_1}, B_{t_1}) \mid (A_t, B_t)\ for\ t_0 \leq t < t_1 ].
\end{align}
Moreover, the corresponding Markov chains have unique stationary distributions.
%\end{theorem}
\end{restatable}

\noindent
The proof is given in Section~\ref{sec:proof_markov}.
%In the proof, we use the ergodic theorem~(see, e.g., \cite{norris1998markov}), which asserts that an irreducible continuous-time Markov chain has a unique stationary distribution if and only if it has a positive recurrent state.

This section concludes with definitions on a Markov chain, which will be used in the subsequent sections.
For a Markov chain on the ground set $\Omega$ with stationary distribution $\pi$, the \textit{balance equations} say that, for any $X\subseteq \Omega$,
\begin{equation}\label{eq:balance}
\sum_{i\in X, j\not\in X} \pi (i) r_{i\to j} = \sum_{i\in X, j\not\in X} \pi (j) r_{j\to i}.
\end{equation}

We denote $z_t (i) = \Pr [Z_t = i]$ for any $i\in \Omega$ and any time $t\geq 0$.
The \textit{total variation distance} between $z_t$ and $\pi$ is defined by 
$||z_t - \pi||_{\rm TV} = \sum_{i\in \Omega} |z_t (i) - \pi (i)|$.
For any $\epsilon > 0$, we define the \textit{mixing time} as
\[
\tau_{\rm mix} (\epsilon) = \inf \{ t \mid ||z_t - \pi||_{\rm TV} \leq \epsilon\}.
\]

%%%%%%%%%%%%%%%%%%%%%%%%%%%%%%%%%%%%%%%%%%%%%%%%%%%%%%%
\subsection{Proof of Theorem~\ref{t5}}\label{sec:proof_markov}
%%%%%%%%%%%%%%%%%%%%%%%%%%%%%%%%%%%%%%%%%%%%%%%%%%%%%%%

In this section, we prove Theorem~\ref{t5}.
%
%\markov*
%
To prove the theorem, we use the ergodic theorem~(see, e.g.,~\cite{norris1998markov}).
For a Markov chain $Z_t$ on the ground set $\Omega$, it is \textit{irreducible} if for any pair of states $i, j\in \Omega$, $j$ is reachable from $i$ with positive probability.
For any $i\in \Omega$ with $Z_{t_0} = i$, we say that $i$ is \textit{positive recurrent} if 
\[
\mathbb{E}[\inf \{t\geq T_1 \mid Z_t = i \}\mid Z_{t_0} = i] < \infty
\]
where $T_1$ is the first time it jumps out of $i$.
The \textit{ergodic theorem} asserts that an irreducible Markov chain has a unique stationary distribution if and only if it has a positive recurrent state.

\begin{proof}[Proof of Theorem~\ref{t5}]
It is not difficult to see that~\eqref{eq:Markov} holds for each of the algorithms as explained in Section~\ref{sec:markov}.

We here prove that the corresponding Markov chains have unique stationary distributions. 
By the ergodic theorem, it suffices to show that the chains are irreducible, and has a positive recurrent state.

First, we show that the chain $(A_t, B_t)$ is irreducible under any algorithm.
Indeed, every state $(i,j)\in \mathbb{Z}_+\times\mathbb{Z}_+$ is reachable from the initial state $(0,0)$ with positive probability, because it happens when $i$ and $j$ agents arrives at the market, respectively, with no acceptable transactions. 
The state $(0,0)$ is reachable from any state $(i,j)$  with positive probability, because it happens when all the agents leave the market without having new agents.
Thus the chain $(A_t, B_t)$ is irreducible under any algorithm.
%because, at any time $t$, any agent can enter or leave the market. All states can reach each other directly or indirectly. 

In the rest of the proof, we show that the state $(0,0)$ is positive recurrent. 
We consider an algorithm that agents on both sides are inactive. 
A pair of pool sizes under this algorithm also forms a Markov chain, denoted by $(\tilde{A}_t, \tilde{B}_t)$.
Then $\tilde{A}_t \geq A_t$ and $\tilde{B}_t \geq B_t$ hold at any time $t$.
%to denote the Markov Chain of this inactive algorithm. At any time $t \geq 0$, following the same path (same agents arrive with the same maximum staying time, and form same edges as other algorithms), it always has $\tilde{A}_t \geq A_t$ and $\tilde{B}_t \geq B_t$. 
In particular, $A_t=0$ if $\tilde{A}_t=0$ and $B_t=0$ if $\tilde{B}_t=0$. 
Therefore, if the state $(0,0)$ in the chain $(\tilde{A}_t, \tilde{B}_t)$ is positive recurrent, the state $(0,0)$ is also positive recurrent for the chain $(A_t, B_t)$ of each algorithm.

Since $\tilde{A}_t$ and $\tilde{B}_t$ are independent, 
we can decompose  the chain $(\tilde{A}_t, \tilde{B}_t)$ into two independent Markov chains $\{\tilde{A}_t\}$ and $\{\tilde{B}_t\}$. 
Akbarpour et al.~\cite{akbarpour2020thickness} showed that the state $0$ is positive recurrent for $\{\tilde{A}_t\}$, i.e., 
%Respectively, by , it shows for any $i=A,B$
\begin{align*}
    \mathbb{E}\left[ \inf \{ t \geq t_0: \tilde{A}_t =0\} \mid \tilde{A}_{t_0} =0 \right] < \infty.
\end{align*}
Thus the proof is complete.
%\qed
\end{proof}

\section{Upper Bounds of Loss}\label{sec:UB}

In the subsequent sections, we prove Theorems~\ref{t1g},~\ref{t1p}, and~\ref{t3} on the upper bounds of the loss, following the proof outline in Section~\ref{s66}.
%proving Propositions~\ref{prop:G2:concentration}--\ref{prop:1sideP_notS2} in Section~\ref{s66}.

%%%%%%%%%%%%%%%%%%%%%%%%%%%%%%%%%%%%%%%%%%%%%%%%%%%%%%%%
\subsection{Preliminaries for Upper-Bounding Loss}\label{sec:UB:Prelim}
%%%%%%%%%%%%%%%%%%%%%%%%%%%%%%%%%%%%%%%%%%%%%%%%%%%%%%%%

In the proofs of Propositions~\ref{prop:G2:concentration}--\ref{prop:1sideP_notS2} in Section~\ref{s66}, we develop recursive equations with the balance equations on the stationary distribution $\pi$.
In this section, we summarize technical lemmas for solving recursive equations, where the proofs of these lemmas are deferred to Appendix~\ref{sec:recursion}.

\begin{restatable}{lemma}{recursionOne}
%\begin{lemma}
\label{l7}
  Let $k^\ast$ be an integer.
  Suppose that, for any $k\geq k^\ast$, 
  \[
  f(k+1) \leq \exp\left(-\frac{k-k^\ast}{k+\eta}\right) f(k)
  \]
  for some constant $\eta$, and $f(k^\ast)\leq 1$.
  \begin{itemize}
  \item[(i)] For any $\sigma\geq 1$, it holds that
  \begin{align}\label{eq:recursion-denom}
  \sum_{k=k^\ast + \sigma+1}^{\infty} f(k) &= O\left(\frac{k^\ast+\eta+\sigma}{\sigma}\right) \exp\left(-\frac{\sigma^2}{\sigma +k^\ast+\eta} \right).
  \end{align}
  \item[(ii)] For any $\sigma\geq 1$ with $\sigma = O(k^\ast+\eta)$, it holds that
  \begin{align}\label{eq:recursion-denom-k}
  \sum_{k=k^\ast + \sigma+1}^{\infty} k f(k) = O\left((k^\ast+\eta)^3\right) \exp\left(-\frac{\sigma^2}{\sigma +k^\ast+\eta} \right).
  \end{align}
  \end{itemize}
%\end{lemma}
\end{restatable}

\begin{restatable}{lemma}{recursionTwo}
%\begin{lemma}
\label{lem:recursion-denom-dec}
  Let $k^\ast$ be an integer.
  Suppose that, for any $k\leq k^\ast$, 
  \[
  f(k-1) \leq \exp\left(-\frac{k^\ast-k}{\eta}\right) f(k)
  \]
  for some constant $\eta$, and $f(k^\ast)\leq 1$.
  Then, for any $\sigma\geq 1$, 
  \[
  \sum_{k=0}^{k^\ast - \sigma} f(k) = O(\eta)e^{-\frac{\sigma^2}{\eta}}.
  \]
%\end{lemma}
\end{restatable}

We also solve recursive equations with small additive terms.

\begin{restatable}{lemma}{recursionExtra}
%\begin{lemma}
\label{lem:recursion-extra}
  Let $k^\ast$ be an integer.
  Suppose that, for any $k\geq k^\ast$, 
  \[
  g(k+1) \leq \alpha_k g(k) + \beta_k,
%  \]
  \text{\quad where\quad }
%  \[
  \alpha_k \leq \exp\left(-\frac{k-k^\ast}{k+\eta}\right),
  \]
  for some constant $\eta$, and $g(k^\ast)\leq 1$.
%  \begin{itemize}
%  \item[(i)] 
Then, for any $\sigma\geq 1$, we have
  \[
  \sum_{k=k^\ast + \sigma+1}^{\infty} g(k) =  O(k^\ast+\eta) e^{-\frac{\sigma^2}{2(\sigma+k^\ast + \eta)}} + O(k^\ast+\eta) \sum_{k=k^\ast}^\infty \beta_i.
  \]
%  \item[(i)]
%  For $\sigma = O(k^\ast + \eta)$ with $\sigma\geq 1$, 
%  \begin{align*}
%  \sum_{k=k^\ast+\sigma+1}^\infty k g(k) 
%  & \leq O\left((k^\ast + \eta)^3 \right) e^{-\frac{\sigma^2}{2(\sigma+k^\ast + \eta)}} + O((k^\ast + \eta)^3 ) \sum_{k=k^\ast}^\infty \beta_i.
%  \end{align*}
%  \end{itemize}
%\end{lemma}
\end{restatable}

%%%%%%%%%%%%%%%%%%%%%%%%%%%%%%%%%%%%%%%%%%%%%%%%%%%%%%%%
\subsection{$\mathsf{Greedy}_2$}\label{sec:UB:G2}
%%%%%%%%%%%%%%%%%%%%%%%%%%%%%%%%%%%%%%%%%%%%%%%%%%%%%%%%

%Recall that we may assume by symmetry that $\lambda_a\geq \lambda_b$.

%%%%%%%%%%%%%%%%%%%%%%%%%%%%%%%%%%%%%%%%%%%%%%%
%\subsection{Outline}
%%%%%%%%%%%%%%%%%%%%%%%%%%%%%%%%%%%%%%%%%%%%%%%

When we run the 2-sided greedy algorithm, the bipartite graph $G_t=(U_t, V_t, E_t)$ almost always has no edges.
Since each agent's staying time follows the Poisson process with mean 1, the rate that some agent in $U_t$~(resp., $V_t$) becomes critical is $A_t$~(resp., $B_t$).
Critical agents perish with probability one, which implies that the expected number of perished agents in $U_t$ at time $t$ is $A_t$~(resp., $B_t$).
Hence we have
\[
\mathbf{L}_a (\mathsf{Greedy}_2) =\frac{1}{\lambda_a T}\mathbb{E}\left[ \int_{t=0}^T A_t dt\right]\text{\ and \ }
\mathbf{L}_b (\mathsf{Greedy}_2) =\frac{1}{\lambda_b T}\mathbb{E}\left[ \int_{t=0}^T B_t dt\right].
\]
%\[
%\mathbf{L} (\mathsf{Greedy}_2) =\frac{1}{(\lambda_a+\lambda_b)T}\mathbb{E}\left[ \int_{t=0}^T \left(A_t + B_t\right) dt\right].
%\]
They are roughly equal to $\frac{\mathbb{E}_{(A, B)\sim \pi}[A]}{\lambda_a}$ and $\frac{\mathbb{E}_{(A, B)\sim \pi}[B]}{\lambda_b}$, respectively, in the stationary distribution $\pi$, which is formally stated as follows.

\begin{lemma}\label{lem:G2_Loss2ExpectedSizes}
For a bipartite matching market $(d_a, d_b, p)$ with $d_a \geq d_b$, 
let $\pi$ be the stationary distribution of the Markov chain of $\mathsf{Greedy}_2$.
For any $\epsilon > 0$ and $T >0$, it holds that
%\begin{align*}
%    \mathbf{L}(\mathsf{Greedy}_2)\leq \frac{\mathbb{E}_{(A,B)\sim \pi}[A+B]}{m}+\frac{\tau_{\textup{mix}}(\epsilon)}{T} + 6 \epsilon + \frac{1}{m}2^{-6\lambda_a}+ \frac{1}{m}2^{-6\lambda_b},
%\end{align*}
%where $m=\lambda_a + \lambda_b$ and $\tau_{\textup{mix}}(\epsilon)$ is the mixing time of the Markov chain.
\begin{align*}
    \mathbf{L}_a(\mathsf{Greedy}_2)&\leq \frac{\mathbb{E}_{(A,B)\sim \pi}[A]}{\lambda_a}+\frac{\tau_{\textup{mix}}(\epsilon)}{T} + 6 \epsilon + \frac{1}{\lambda_a}2^{-6\lambda_a},\\
    \mathbf{L}_b(\mathsf{Greedy}_2)&\leq \frac{\mathbb{E}_{(A,B)\sim \pi}[B]}{\lambda_b}+\frac{\tau_{\textup{mix}}(\epsilon)}{T} + 6 \epsilon + \frac{1}{\lambda_b}2^{-6\lambda_b},
\end{align*}
where $\tau_{\textup{mix}}(\epsilon)$ is the mixing time of the Markov chain.
\end{lemma}

\begin{figure}[t]
    \centering
     \includegraphics[scale=0.3]{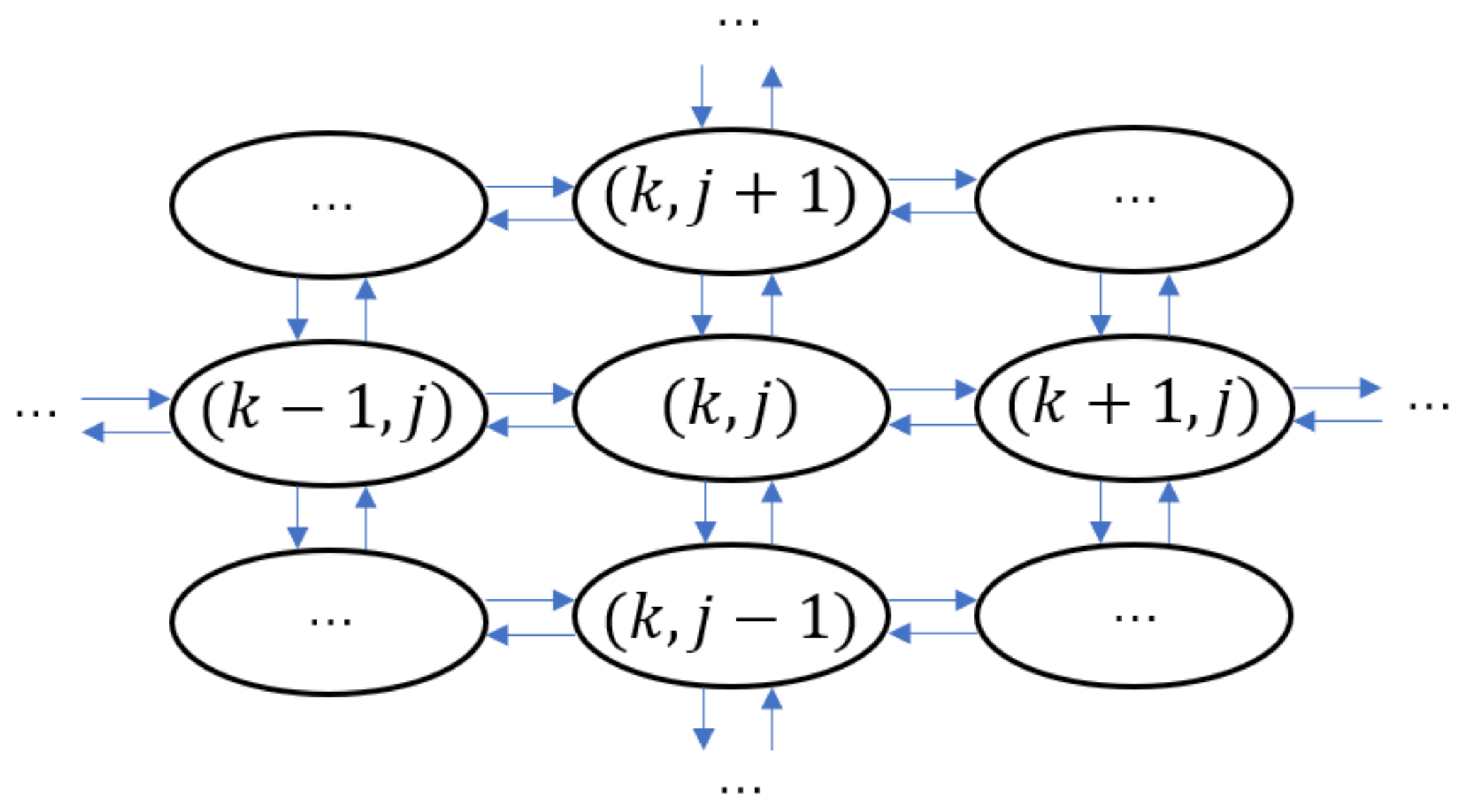}
    \caption{The Markov chain of the 2-sided/1-sided Greedy algorithm.}
    \label{fig:greedy2markov}
\end{figure}

The proof of Lemma~\ref{lem:G2_Loss2ExpectedSizes} is given in Section~\ref{sec:Greduction}. Note that as $\lambda_a$, $\lambda_b$, and $T$ grow, the last three terms become negligible.

%%%%%%%%%%%%%%%%%%%%%%%%%%%%%%%%%%%%%%%%%%%%%%%
\paragraph{Markov Chain}
%%%%%%%%%%%%%%%%%%%%%%%%%%%%%%%%%%%%%%%%%%%%%%%

We formally define a Markov chain for the 2-sided Greedy algorithms.
Let $(A_t, B_t)$ be a continuous-time Markov Chain on $\mathbb{Z}_+\times \mathbb{Z}_+$, where $\mathbb{Z}_+$ denotes the set of non-negative integers.
For any $(k, j)\in \mathbb{Z}_+\times \mathbb{Z}_+$ and $(k', j')\in \mathbb{Z}_+\times \mathbb{Z}_+$, we denote by $r_{(k, j)\to (k', j')}$ the transition rate from $(k, j)$ to $(k', j')$.

For any pair of pool sizes $(k, j)$, the Markov chain transits only to the states $(k+1, j)$, $(k, j+1)$, $(k-1, j)$, and $(k, j-1)$.
See Figure~\ref{fig:greedy2markov}.
It transits to $(k+1, j)$ or $(k, j+1)$ when a new agent arrives and she does not get matched, and  to $(k-1, j)$ or $(k, j-1)$ when either a new agent arrives and she gets matched with some agent in the pool, or some agent in the pool perishes.
Therefore, the transition rates are defined by
\begin{align*}
    r_{(k,j) \rightarrow (k+1,j)}&=\lambda_a (1-p)^j, \\
    r_{(k,j) \rightarrow (k,j+1)}&=\lambda_b (1-p)^k, \\
    r_{(k,j) \rightarrow (k-1,j)}&=k+\lambda_b(1-(1-p)^k), \\
    r_{(k,j) \rightarrow (k,j-1)}&=j+\lambda_a(1-(1-p)^j).
\end{align*}

%%%%%%%%%%%%%%%%%%%%%%%%%%%%%%%%%%%%%%%%%%%%%%%
\paragraph{Concentration of Pool Sizes}
%%%%%%%%%%%%%%%%%%%%%%%%%%%%%%%%%%%%%%%%%%%%%%%
By Lemma~\ref{lem:G2_Loss2ExpectedSizes}, we can upper-bound the loss of $\mathsf{Greedy}_2$ by bounding the expected pool sizes $\mathbb{E}_{(A, B)\sim \pi}[A]$ and $\mathbb{E}_{(A, B)\sim \pi}[B]$  in the stationary distribution $\pi$.
%In the rest of this section, we aim at bounding $\mathbb{E}_{(A, B)\sim \pi}[A]$ and $\mathbb{E}_{(A, B)\sim \pi}[B]$.
%To this end, we analyze the corresponding continuous-time Markov chain to show that both $\mathbb{E}_{(A, B)\sim \pi}[A]$ and $\mathbb{E}_{(A, B)\sim \pi}[B]$ do not become so large.
%In what follows, we obtain an upper bound of $\mathbb{E}_{(A, B)\sim\pi}[A + B]$.
We show that, under $\mathsf{Greedy}_2$, the probability that the pool size $A_t$~(resp., $B_t$) in the steady state is larger than some value $k_2$~(resp., $\ell_2$) is exponentially small, where $k_2 \leq \lambda_a-\lambda_b+\frac{\lambda_b \log (d_b+3)}{d_b}$ and $\ell_2 \leq \frac{\lambda_b \log (d_b+3)}{d_b}$.

%\textbf{Need to lower-bound $k_2$?????}

\GTwoConcentrate*

The proof of Proposition~\ref{prop:G2:concentration} is based on the balance equations.
For any integer $k\geq 0$, let $X=\{(i,j)\mid 0\leq i \leq k, j \geq 0 \}$.
We apply the balance equation~\eqref{eq:balance} with the set $X$.
Then we obtain
\begin{align*}
    \sum_{j=0}^{\infty} \lambda_a (1-p)^j \pi (k, j) = \left(k+1+\lambda_b(1-(1-p)^{k+1})\right)\sum_{j=0}^{\infty} \pi (k+1, j). 
\end{align*}
Since the LHS is at most $\lambda_a \sum_{j=0}^{\infty} \pi(k,j)$, it holds that
\[
\sum_{j=0}^{\infty} \pi (k+1, j)\leq \frac{\lambda_a}{k+1+\lambda_b(1-(1-p)^{k+1})}\sum_{j=0}^{\infty} \pi (k, j).
\]
Thus we have a recursive relationship for $\sum_{j=0}^{\infty} \pi (k, j)$ for $k=0,1,2,\dots$.
By solving this recursive equation, we have the first part of Proposition~\ref{prop:G2:concentration}.
The second one is analogous.
See Section~\ref{sec:proof_g2concentration} for the complete proof.

%By setting $\sigma=\Theta (\sqrt{\lambda_a \log \lambda_a}+\sqrt{\lambda_b \log \lambda_b})$, 
Proposition~\ref{prop:G2:concentration} implies that $\mathbb{E}_{(A, B)\sim\pi}[A] \leq k_2 + o(\lambda_a)$ by setting $\sigma=\Theta (\sqrt{\lambda_a \log \lambda_a})$.
Since $\mathbf{L}_a(\mathsf{Greedy}_2)\approx \frac{1}{\lambda_a}\mathbb{E}_{(A, B)\sim\pi}[A]$ by Lemma~\ref{lem:G2_Loss2ExpectedSizes}, this implies that 
\[
\mathbf{L}_a(\mathsf{Greedy}_2) \leq \frac{1}{\lambda_a}\left(\lambda_a-\lambda_b+\frac{\lambda_a \log (d_b+3)}{d_a}+o(\lambda_a)\right) = \frac{d_a-d_b}{d_a}+\frac{\log (d_b+3)}{d_a}+o(1).
\]
Similarly, by setting $\sigma=\Theta (\sqrt{\lambda_b \log \lambda_b})$, we have $\mathbb{E}_{(A, B)\sim\pi}[B] \leq \ell_2 + o(\lambda_b)$, and hence
\[
\mathbf{L}_b(\mathsf{Greedy}_2) \leq \frac{1}{\lambda_b}\left(\frac{\lambda_b \log (d_b+3)}{d_b}+o(\lambda_b)\right) = \frac{\log (d_b+3)}{d_b}+o(1).
\]
This shows Theorem~\ref{t1g}.
See Section~\ref{sec:loss_g2} for the details.

%%%%%%%%%%%%%%%%%%%%%%%%%%%%%%%%%%%%%%%%%%%%%%%
%\paragraph{Concentration of Pool Sizes}
\subsubsection{Concentration of Pool Sizes: Proof of Proposition~\ref{prop:G2:concentration}}\label{sec:proof_g2concentration}
%%%%%%%%%%%%%%%%%%%%%%%%%%%%%%%%%%%%%%%%%%%%%%%

We define $k_2$ as the value that satisfies $\lambda_a=k_2+\lambda_b (1-(1-p)^{k_2})$, and  $\ell_2$ as the value that satisfies $\lambda_b=\ell_2 + \lambda_a(1-(1-p)^{\ell_2})$. 
Since we assume that $\lambda_a \geq \lambda_b$, it holds that $k_2\geq \ell_2$.
Moreover, these values can be approximated with $\lambda_a$ and $\lambda_b$.

\begin{lemma}\label{G2:kstarlstar}
Suppose that $\lambda_a \geq \lambda_b$ and $d_a\geq d_b>0$.
It holds that
\begin{align}
\max\left\{\lambda_a-\lambda_b, \frac{\lambda_a}{1+d_b}\right\}&\leq k_2\leq \lambda_a-  \lambda_b +\frac{1}{p} \log (d_b+3),\notag \\
%\max\left\{\lambda_a-\lambda_b, \frac{\lambda_a}{1+d_b}\right\}&\leq k_2\leq \lambda_a-  \lambda_b +\frac{1}{p} + \frac{1}{p} \log d_b,\\
\frac{\lambda_b}{1+d_a}&\leq \ell_2\leq \frac{\lambda_b}{d_b}\log (d_b+3).\notag
\end{align}
\end{lemma}

\begin{proof}
We first estimate $k_2$.
We denote $k_{\max} = \lambda_a-  \lambda_b +\frac{1}{p} \log (d_b+3)$ and $k_{\min} = \max\left\{\lambda_a-\lambda_b, \frac{\lambda_a}{1+d_b}\right\}$.
We define a function $f$ by $f(k)= k + \lambda_b (1-(1-p)^{k})-\lambda_a$.
Then the function $f$ is a non-decreasing function, and $f(k_2) = 0$.
Since $f(k)\geq k + \lambda_b (1-e^{-pk})- \lambda_a$, it holds that
\begin{align*}
f(k_{\max}) 
\geq \lambda_a-  \lambda_b +\frac{\lambda_b}{d_b} \log (d_b+3) + \lambda_b\left(1-e^{-\log (d_b+3)}\right) - \lambda_a 
=  \frac{\lambda_b}{d_b} \log (d_b+3) - \frac{\lambda_b}{d_b+3} \geq 0 
\end{align*}
since $\log (d_b+3)\geq 1$.
Moreover, since $f(k)\leq k + \lambda_b (1-(1-pk))- \lambda_a=k+\lambda_b pk -\lambda_a$, it holds that
\[
f\left(\frac{\lambda_a}{1+d_b}\right) \leq \frac{\lambda_a}{1+d_b} + \lambda_b p \frac{\lambda_a}{1+d_b} -\lambda_a= 0.
\]
We also have
\[
f(\lambda_a-\lambda_b) = (\lambda_a-\lambda_b) + \lambda_b (1-e^{-p(\lambda_a-\lambda_b)})- \lambda_a=-\lambda_b e^{-p(\lambda_a-\lambda_b)}< 0.
\]
Thus the inequality for $k_2$ holds.

The inequality for $\ell_2$ can be proved similarly.
In fact, by letting $f'(k)= k + \lambda_a (1-(1-p)^{k})-\lambda_b$, we obtain
\begin{align*}
f'\left(\frac{\lambda_b}{d_b}\log (d_b+3)\right)
\geq (\lambda_a-\lambda_b)+\frac{\lambda_b}{d_b}\log (d_b+3) - \frac{\lambda_a}{d_b+3}
\geq (\lambda_a-\lambda_b)+\frac{\lambda_b}{d_b+3} - \frac{\lambda_a}{d_b+3}>0,
\end{align*}
since $\log (d_b+3)\geq 1$.
Thus the inequality holds.
%\qed
\end{proof}

We now prove Proposition~\ref{prop:G2:concentration} using $k_2$ and $\ell_2$ defined as above.

\begin{proof}[Proof of Proposition~\ref{prop:G2:concentration}]
Let $k_2$, $\ell_2$ be the values defined as above, that is, they satisfy  $\lambda_a=k_2+\lambda_b (1-(1-p)^{k_2})$ and $\lambda_b=\ell_2 + \lambda_a(1-(1-p)^{\ell_2})$. 
By the balance equation~\eqref{eq:balance} with $X=\{(i,j)\mid 0\leq i \leq k, j \geq 0 \}$ for any $k \geq 0$, it holds that
%and $\{(i,j) \mid i \geq k+1, j \geq 0 \}$
\begin{align}
    \sum_{j=0}^{\infty} \lambda_a (1-p)^j \pi (k, j) = \left(k+1+\lambda_b(1-(1-p)^{k+1})\right)\sum_{j=0}^{\infty} \pi (k+1, j). \label{e14}
\end{align}
Since the LHS of \eqref{e14} is at most $\lambda_a \sum_{j=0}^{\infty} \pi(k,j)$, we obtain
\[
\frac{\sum_{j=0}^{\infty} \pi (k+1, j)}{\sum_{j=0}^{\infty} \pi (k, j)}\leq \frac{\lambda_a}{k+1+\lambda_b(1-(1-p)^{k+1})}.
\]
By the definition of $k_2$, we have, for $k\geq k_2$, 
\begin{align*}
\frac{\lambda_a}{k+1+\lambda_b(1-(1-p)^{k+1})}
 = \frac{\lambda_a}{k+1-k_2+\lambda_b(1-(1-p)^{k+1})-\lambda_b (1-(1-p)^{k_2})+\lambda_a}
\leq \frac{\lambda_a}{k-k_2+\lambda_a}.
\end{align*}
Hence, for any $k\geq k_2$, 
\begin{align*}
\frac{\sum_{j=0}^{\infty} \pi (k+1, j)}{\sum_{j=0}^{\infty} \pi (k, j)}\leq 
\frac{\lambda_a}{k-k_2+\lambda_a} = 1 - \frac{k-k_2}{k-k_2+\lambda_a}\leq \exp\left(-\frac{k-k_2}{k-k_2+\lambda_a}\right).
%    \frac{\sum_{\infty}^{j=0}\pi(k+1, j)}{\sum_{\infty}^{j=0}\pi(k, j)} &\leq \frac{\lambda_a}{k+1+\lambda_b(1-(1-p)^{k+1})} \leq \frac{\lambda_a}{k-k_2+\lambda_a} \\
%    &\leq 1- \frac{k-k_2}{k-k_2+\lambda_a} \leq \exp\left( -\frac{k -k_2}{k-k_2 +\lambda_a} \right).
\end{align*}

We apply Lemma~\ref{l7}~\eqref{eq:recursion-denom} with $k^\ast=k_2$ and $\eta = \lambda_a - k_2$, which implies that, for any $\sigma \geq 1$,
\begin{align*}
\sum_{k=k_2+ \sigma+1}^\infty \sum_{j=0}^{\infty} \pi (k, j)
& \leq O\left(\frac{\sigma +\lambda_a}{\sigma}\right) \exp\left(-\frac{\sigma^2}{\sigma +\lambda_a} \right)
 = O\left(\lambda_a\right) \exp\left(-\frac{\sigma^2}{\sigma +\lambda_a} \right) %\label{e1}
\end{align*}
where we note $O\left(\frac{\sigma +\lambda_a}{\sigma}\right) = O\left(\lambda_a\right)$ since $\sigma \geq 1$.
This proves the first inequality in the proposition.
The second inequality is analogous. 
%\qed
\end{proof}

We remark that, in the proof of Proposition~\ref{prop:G2:concentration}, using Lemma~\ref{l7}~(ii) instead of~(i), 
we obtain
\begin{align}\label{eq:greedy-hugeZ1-sum}
\sum_{k\geq k_2 + \sigma_a+1} k \sum_{j\geq 0} \pi(k, j) = O\left(\lambda_a^3\right) e^{-\frac{\sigma_a^2}{\sigma_a +\lambda_a}}
\end{align}
when $\sigma_a =O(\lambda_a)$.

%%%%%%%%%%%%%%%%%%%%%%%%%%%%%%%%%%%%%%%%%%%%%%%
%\paragraph{Bounding the Loss}
\subsubsection{Bounding the Loss of $\mathsf{Greedy}_2$}\label{sec:loss_g2}
%%%%%%%%%%%%%%%%%%%%%%%%%%%%%%%%%%%%%%%%%%%%%%%

We are ready to bound $\mathbb{E}_{(A, B)\sim \pi}[A]$ and $\mathbb{E}_{(A, B)\sim \pi}[B]$, respectively.
This shows Theorem~\ref{t1g} combined with Lemma~\ref{lem:G2_Loss2ExpectedSizes}.

\begin{theorem}\label{thm:G2:upbound}
For a bipartite matching market $(d_a, d_b, p)$ with $d_a \geq d_b$, 
let $\pi$ be the stationary distribution of the Markov chain of $\mathsf{Greedy}_2$.
%Suppose that $\lambda_a\geq \lambda_b$. 
It holds that
\begin{align*}
\frac{\mathbb{E}_{(A, B)\sim \pi}[A]}{\lambda_a} &\leq \frac{d_a-  d_b}{d_a}+ \frac{\log (d_b+3)}{d_a} + o(1), \\
\frac{\mathbb{E}_{(A, B)\sim \pi}[B]}{\lambda_b} &\leq \frac{\log (d_b+3)}{d_b} + o(1).
\end{align*}
%\[
%\frac{\mathbb{E}_{(A, B)\sim \pi}[A + B]}{\lambda_a + \lambda_b} \leq \frac{d_a-  d_b}{d_a+d_b}+ \frac{1+2\log d_b}{d_a+d_b} + o(1).
%\]
\end{theorem}
\begin{proof}
Define $\sigma_a = \Theta(\sqrt{\lambda_a \log \lambda_a})$.
We observe that
\begin{align*}
\mathbb{E}_{(A, B)\sim \pi}[A] & = \sum_{k\geq 0} k \sum_{j\geq 0} \pi(k, j)
 \leq k_2+\sigma_a  + \sum_{k\geq k_2 +\sigma_a+1} k \sum_{j\geq 0} \pi(k, j).
\end{align*}
Since $\sigma_a =O(\lambda_a)$, the last term is bounded by a constant from~\eqref{eq:greedy-hugeZ1-sum}.
Hence we obtain
\[
\mathbb{E}_{(A, B)\sim \pi}[A] \leq k_2+ O\left(\sqrt{\lambda_a \log \lambda_a}\right).
%\sum_{k\geq 0} k \sum_{j\geq 0} \pi(k, j) \leq k_2+ O\left(\sqrt{\lambda_a \log \lambda_a}\right).
\]
Therefore, by Lemma~\ref{G2:kstarlstar}, 
\begin{align*}
\frac{\mathbb{E}_{(A, B)\sim \pi}[A]}{\lambda_a} 
\leq \frac{k_2+o(\lambda_a)}{\lambda_a}
\leq \frac{1}{\lambda_a}\left(\lambda_a -  \lambda_b + \frac{1}{p} \log (d_b+3) \right) + o(1)
= \frac{d_a- d_b}{d_a}+ \frac{\log (d_b+3)}{d_a} + o(1).
\end{align*}

Similarly, we see that
\[
\sum_{j\geq 0} j \sum_{k\geq 0} \pi(k, j) \leq \ell_2+ O\left(\sqrt{\lambda_b \log \lambda_b}\right).
\]
and, by Lemma~\ref{G2:kstarlstar}, it holds that
\begin{align*}
\frac{\mathbb{E}_{(A, B)\sim \pi}[B]}{\lambda_b} 
\leq \frac{\ell_2+o(\lambda_b)}{\lambda_b}
\leq \frac{1}{\lambda_b}\left(\frac{1}{p} \log (d_b+3) \right) + o(1)
= \frac{\log (d_b+3)}{d_b} + o(1).
\end{align*}
This completes the proof.
%\qed
\end{proof}

\subsection{$\mathsf{Patient}_2$}\label{sec:UB:P2}
%%%%%%%%%%%%%%%%%%%%%%%%%%%%%%%%%%%%%%%%%%%%%%%%%%%%%%%%

%%%%%%%%%%%%%%%%%%%%%%%%%%%%%%%%%%%%%%%%%%%%%%%%%%%%%%%%%%%%%%%%%%%%%%%%%
%\subsection{Outline}
%%%%%%%%%%%%%%%%%%%%%%%%%%%%%%%%%%%%%%%%%%%%%%%%%%%%%%%%%%%%%%%%%%%%%%%%%%%

Similarly to the case of $\mathsf{Greedy}_2$, we first show that the loss can be formulated using the expected pool sizes in the steady state.

Under $\mathsf{Patient}_2$, conditional on $A_t$ and $B_t$, the graph $G_t$ is a random bipartite graph with vertex sets $A_t$ and $B_t$ where an edge is formed with probability $p$. The rate that some agent in $U_t$~(resp., $V_t$) becomes critical is $A_t$~(resp., $B_t$). 
Since $G_t$ is a random bipartite graph, a critical agent in $U_t$~(resp., $V_t$) perishes with probability $(1-p)^{B_t}$~(resp., $(1-p)^{A_t}$). 
Therefore, we have
\begin{align*}
\mathbf{L}_a (\mathsf{Patient}_2)  =\frac{1}{\lambda_aT}\mathbb{E}\left[ \int_{t=0}^T A_t(1-p)^{B_t} dt\right]
\text{\quad and\quad}
\mathbf{L}_b (\mathsf{Patient}_2)  =\frac{1}{\lambda_bT}\mathbb{E}\left[ \int_{t=0}^T B_t(1-p)^{A_t} dt\right].
\end{align*}
%\[
%\mathbf{L} (\mathsf{Patient}_2) =\frac{1}{(\lambda_a+\lambda_b)T}\mathbb{E}\left[ \int_{t=0}^T \left(A_t(1-p)^{B_t} + B_t(1-p)^{A_t}\right) dt\right].
%\]
They are roughly equal to $\frac{1}{\lambda_a}\mathbb{E}_{(A, B)\sim \pi}\left[A(1-p)^{B}\right]$ and $\frac{1}{\lambda_b}\mathbb{E}_{(A, B)\sim \pi}\left[ B(1-p)^{A}\right]$, respectively, where $\pi$ is the stationary distribution of the corresponding Markov chain, as stated below.
The proof is deferred to Section~\ref{sec:P2reduction}.

\begin{lemma}\label{lem:P2_Loss2ExpectedSizes}
For a bipartite matching market $(d_a, d_b, p)$ with $d_a \geq d_b$, 
let $\pi$ be the stationary distribution of the Markov chain of $\mathsf{Patient}_2$.
For any $\epsilon > 0$ and $T >0$,
    \begin{align*}
    \mathbf{L}_a(\mathsf{Patient}_2)
    &\leq \frac{\mathrm{E}_{(A,B)\sim \pi}[A(1-p)^{B}]}{\lambda_a} + \frac{\tau_{\textup{mix}}(\epsilon)}{T} + 6\frac{\epsilon}{p}+\frac{2^{-6\lambda_a}}{\lambda_a},\\
    \mathbf{L}_b(\mathsf{Patient}_2)
    &\leq \frac{\mathrm{E}_{(A,B)\sim \pi}[B(1-p)^{A}]}{\lambda_b} + \frac{\tau_{\textup{mix}}(\epsilon)}{T} + 6\frac{\epsilon}{p}+\frac{2^{-6\lambda_b}}{\lambda_b}, 
    \end{align*}
%    \[
%    \mathbf{L}(\mathsf{Patient}_2)
%    \leq \frac{\mathrm{E}_{(A,B)\sim \pi}[A(1-p)^{B}+B(1-p)^{A}]}{\lambda_a+\lambda_b} + \frac{\tau_{\textup{mix}}(\epsilon)}{T} + 6\frac{\epsilon}{p}+\frac{2^{-6\lambda_a}+2^{-6\lambda_b}}{\lambda_a+\lambda_b}, 
%\]
where $\tau_{\textup{mix}}(\epsilon)$ is the mixing time of the Markov chain.
\end{lemma}

Note that, by setting $\varepsilon$ small enough so that $\frac{\epsilon}{p} = \frac{\epsilon \lambda_a}{d_a}\to 0$, e.g., $\epsilon = \frac{1}{\lambda_a^2}$, the last three terms are negligible.
%%%%%%%%%%%%%%%%%%%%%%%%%%%%%%%%%%%%%%%%%%%%%%%%%%%%%%%%
\paragraph{Markov Chain}
%%%%%%%%%%%%%%%%%%%%%%%%%%%%%%%%%%%%%%%%%%%%%%%%%%%%%%%%

We here define a Markov chain $(A_t, B_t)$ on $\mathbb{Z}_+\times \mathbb{Z}_+$ for the 2-sided Patient algorithm.
For any pair of pool sizes $(k, j)$, the Markov chain transits only to the states $(k+1, j)$, $(k, j+1)$, $(k-1, j)$, $(k, j-1)$, and $(k-1, j-1)$.
See Figure~\ref{fig:patient2markov}.
It transits to $(k+1, j)$ or $(k, j+1)$ when a new agent arrives, and to $(k-1, j)$ or $(k, j-1)$ when some agent in the pool leaves the market without getting matched to another agent, and to $(k-1, j-1)$ when some agent in the pool leaves the market with getting matched to another agent.
\begin{figure}[t]
\centering
 \includegraphics[scale=0.3]{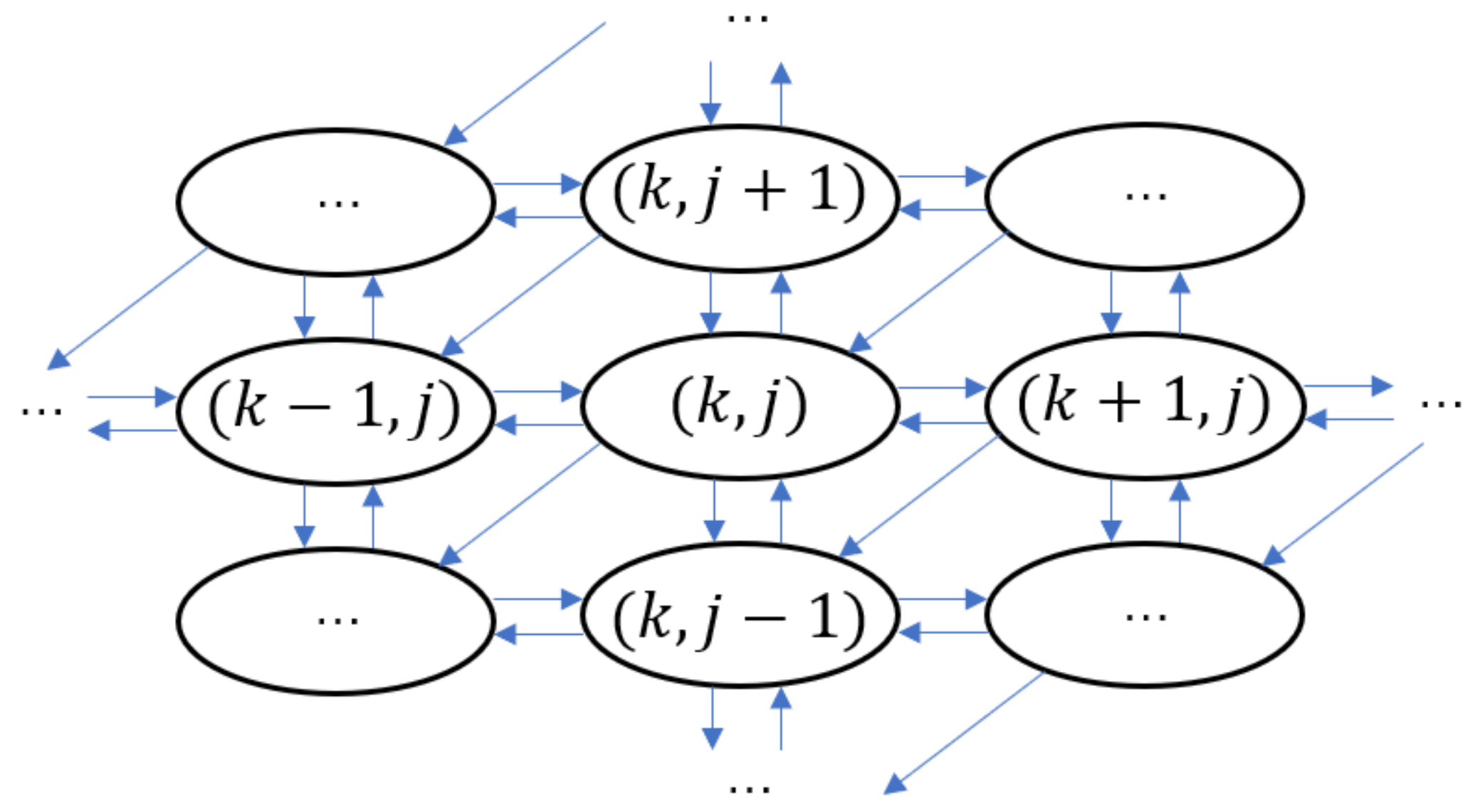}
\caption{The Markov chain of the 2-sided/1-sided Patient algorithm.}
\label{fig:patient2markov}
\end{figure}
Thus we have
\begin{align*}
    r_{(k,j) \rightarrow (k+1,j)}&=\lambda_a, \\
    r_{(k,j) \rightarrow (k,j+1)}&=\lambda_b, \\
    r_{(k,j) \rightarrow (k-1,j)}&=k(1-p)^j, \\
    r_{(k,j) \rightarrow (k,j-1)}&=j(1-p)^k, \\
    r_{(k,j) \rightarrow (k-1,j-1)}&=k(1-(1-p)^j)+j(1-(1-p)^k).
\end{align*}

%%%%%%%%%%%%%%%%%%%%%%%%%%%%%%%%%%%%%%%%%%%%%%%
\paragraph{Concentration of Pool Sizese}
%%%%%%%%%%%%%%%%%%%%%%%%%%%%%%%%%%%%%%%%%%%%%%%

We first consider the case when $d_a \geq d_b$.
%Better bounds for the balanced case, i.e., when $d_a = d_b$, will be discussed later.
Recall that $k_2$ is the value that satisfies $\lambda_a=k_2+\lambda_b (1-(1-p)^{k_2})$, and  $\ell_2$ is the value that satisfies $\lambda_b=\ell_2 + \lambda_a(1-(1-p)^{\ell_2})$~(See Section~\ref{sec:UB:G2}). 
By Lemma~\ref{G2:kstarlstar},
we know $k_2\geq \max\left\{\lambda_a-\lambda_b, \frac{\lambda_a}{1+d_b}\right\}$ and $\ell_2\geq \frac{\lambda_b}{1+d_a}$.

\PTwoConcentrate*

The above proposition shows that the probability $\Pr_{(A, B)\sim \pi} [(A, B) = (k, \ell)]$ decreases exponentially outside of the region $S$, which means that the expected pool sizes $\mathbb{E}_{(A, B)\sim \pi}[A]$ and $\mathbb{E}_{(A, B)\sim \pi}[B]$ are roughly between $k_2$ and $\lambda_a$ and between $\ell_2$ and $\lambda_b$, respectively.
%We remark that the expected pool sizes are larger than those of the 2-sided Greedy algorithm.
%In other words, waiting to match makes the market thicker after the long run.
%, making the loss exponentially smaller.
%To prove Proposition~\ref{prop:2sideP_notS},
%we first observe that the probability that the pool size $A_t$~(resp., $B_t$) in the steady state is larger than $\lambda_a$~(resp., $\lambda_b$) is small.
%This can be shown by applying the balance equation with the same set of states as for $\mathsf{Greedy}_2$.
%Moreover, we show that the probability that the pool size $A_t$~(resp., $B_t$) in the steady state is smaller than some value $\ukT$~(resp., $\ulT$) is small.
Therefore, we can upper-bound $\mathbb{E}_{(A, B)\sim \pi}\left[A(1-p)^{B}\right]$ and $\mathbb{E}_{(A, B)\sim \pi}\left[ B(1-p)^{A}\right]$~(see Theorem~\ref{thm:P2unbalanced}), which implies the first part of Theorem~\ref{t1p} by Lemma~\ref{lem:P2_Loss2ExpectedSizes}. 
See Section~\ref{sec:2sideP_notS} for the details with the proof of Proposition~\ref{prop:2sideP_notS}.

%%%%%%%%%%%%%%%%%%%%%%%%%%%%%%%%%%%%%%%%%%%%%%%
%\paragraph{Concentration of Pool Sizes in the Balanced Case}
%%%%%%%%%%%%%%%%%%%%%%%%%%%%%%%%%%%%%%%%%%%%%%%

When $d_a=d_b$, the bounds can be improved by further narrowing the concentration region.
We first show Proposition~\ref{prop:P2balancedSum} in Section~\ref{s66}, that says that the probability that the sum of the pool sizes $A_t+B_t$ in the steady state is at most about $(\lambda_a+\lambda_b)/2$ is small.
Note that this holds even when $\lambda_a \neq \lambda_b$.
%\PTwoConcentrateSum*
Moreover, we prove Proposition~\ref{prop:P2balanced}: when $d_a=d_b$, the probability that the difference of the pool sizes $A_t-B_t$ in the steady state is large is small.
%\PTwoConcentrateBalanced*
By the two propositions, together with Lemma~\ref{lem:P2_Loss2ExpectedSizes}, we can upper-bound the loss~(Theorem~\ref{thm:P2balanced}), which implies the second part of Theorem~\ref{t1p}. 
See Section~\ref{sec:P2balanced} for the details.

%%%%%%%%%%%%%%%%%%%%%%%%%%%%%%%%%%%%%%%%%%%%%%%%%%%%%%%%%%%%%%%%%%%%%%%%%
\subsubsection{Loss of the Unbalanced Case: Proof of Proposition~\ref{prop:2sideP_notS}}\label{sec:2sideP_notS}
%%%%%%%%%%%%%%%%%%%%%%%%%%%%%%%%%%%%%%%%%%%%%%%%%%%%%%%%%%%%%%%%%%%%%%%%%%%

%Moreover, we show that the probability that the pool size $A_t$~(resp., $B_t$) in the steady state is smaller than some value $\ukT$~(resp., $\ulT$) is small.

To prove Proposition~\ref{prop:2sideP_notS}, we show the following two lemmas.

We first observe that the probability that the pool size $A_t$~(resp., $B_t$) in the steady state is larger than $\lambda_a$~(resp., $\lambda_b$) is small.
This can be shown by applying the balance equation with the same set of states as in the proof of Proposition~\ref{prop:G2:concentration}.

%\lambda_a and \lambda_b
%\nknote{The lemma below is the same as before}
\begin{lemma}\label{lem:p2_A_B}
  For any  $\sigma\geq 1$, it holds that
  \begin{align*}
  \Pr_{(A, B)\sim \pi} \left[ A \geq \lambda_a+ \sigma +1\right]\leq O\left(\lambda_a\right) e^{-\frac{\sigma^2}{\sigma +\lambda_a}}
  \text{\quad and\quad}
  \Pr_{(A, B)\sim \pi} \left[ B \geq \lambda_b+ \sigma +1\right]\leq O\left(\lambda_b\right) e^{-\frac{\sigma^2}{\sigma +\lambda_b}}.
  \end{align*}
\end{lemma}
\begin{proof}
  By the balance equation~\eqref{eq:balance} with $X=\{(i, j)\mid 0\leq i\leq k, j\geq 0\}$ for any $k\geq 0$, it holds that
  \begin{equation}\label{eq:p2new1}
  \lambda_a \sum_{j=0}^{\infty} \pi (k, j) = \sum_{j=0}^{\infty} \left(k+1 + j (1-(1-p)^{k+1})\right) \pi (k+1, j) \geq (k+1) \sum_{j=0}^{\infty} \pi (k+1, j).
  \end{equation}
  Therefore, for $k\geq \lambda_a -1$, we have 
  \begin{align*}
      \frac{\sum_{j=0}^{\infty}\pi(k+1,j)}{\sum_{j=0}^{\infty}\pi(k,j)} \leq \frac{\lambda_a}{k+1} = 1-\frac{k-\lambda_a+1}{k+1} \leq\exp\left( -\frac{k-\lambda_a+1}{k+1}\right).
  \end{align*}
  Hence, by applying Lemma~\ref{l7}~\eqref{eq:recursion-denom} with $k^*=\lambda_a-1$ and $\eta=1$, we have, for any $\sigma\geq 1$,
    \begin{align*}
    \sum_{k=\lambda_a +\sigma+1}^\infty \sum_{j=0}^{\infty} \pi (k, j) &= O\left(\frac{\lambda_a+\sigma}{\sigma}\right) \exp\left(-\frac{\sigma^2}{\sigma +\lambda_a} \right).
    \end{align*}
  This proves the first inequality of the lemma as $\sigma\geq 1$.
  The second one is analogous.
  %\qed
  \end{proof}
  
  We remark that, using Lemma~\ref{l7}~\eqref{eq:recursion-denom-k} in the proof of Lemma~\ref{lem:p2_A_B}, we obtain
  \begin{align}
    \sum_{k\geq \lambda_a+\sigma_a+1} k \sum_{j\geq 0} \pi(k, j) &= O\left(\lambda_a^3\right) e^{-\frac{\sigma_a^2}{\sigma_a +\lambda_a}}, \label{eq:p2-sum1}\\ 
   \sum_{j\geq \lambda_b+\sigma_b+1} j\sum_{k\geq 0} \pi(k, j) &= O\left(\lambda_b^3\right) e^{-\frac{\sigma_b^2}{\sigma_b +\lambda_b}},\label{eq:p2-sum2}
  \end{align}
  when $\sigma_a = O(\lambda_a)$ and $\sigma_b = O(\lambda_b)$.

We next show that the pool size $A_t$~(resp., $B_t$) is not so small in the steady state.
For any $\sigma_a, \sigma_b\geq 1$, define $\ukT$ to be the value that satisfies $\lambda_a = k+(\lambda_b+ \sigma_b) (1-(1-p)^{k})$, and $\ulT$ to be the value that satisfies $\lambda_b = \ell +(\lambda_a+ \sigma_a) (1-(1-p)^{\ell})$.

We first observe the following relationship between $k_2, \ell_2$ and $\ukT$, $\ulT$.
\begin{lemma}\label{lem:2P_definition_k}
  For any $\sigma_a, \sigma_b\geq 1$, define $\ukT$ and $\ulT$ as above.
  Then $k_2-\sigma_b < \ukT < k_2$ and  $\ell_2-\sigma_a < \ulT < \ell_2$.
\end{lemma}
  \begin{proof}
  Define a function $f(k) = k+(\lambda_b +\sigma_b)(1-(1-p)^k)-\lambda_a$, which is a non-decreasing function. 
  Note that $f(\ukT)=0$. 
  Since $f(k_2-\sigma_b)=-\sigma_b(1-p)^{k_2-\sigma_b}<0$ from the definition of $k_2$, we obtain 
  $k_2-\sigma_b < \ukT$. 
  Since $f(k_2)=\sigma_b(1-(1-p)^{k_2}))>0$, we also have $\ukT < k_2$.
  The inequality for $\ulT$ is analogous.
  %\qed
\end{proof}

Using $\ukT$ and $\ulT$ defined as above, 
we can lower-bound the expected pool sizes.
We remark that we need to solve recursive equations with additive terms.

\begin{lemma}\label{lem:p2_LB}
  For any $\sigma_a, \sigma_b\geq 1$ such that $\sigma_a =O(\lambda_a)$ and $\sigma_b =O(\lambda_b)$, define $\ukT$ and $\ulT$ as above.
  Then it holds that, for any $\sigma \geq 1$, 
  \begin{align*}
    \Pr_{(A, B)\sim \pi} \left[ A \leq \ukT - \sigma -1\right]&\leq O(\lambda_a) e^{-\frac{\sigma^2}{2(\sigma+\lambda_a)}} + O(\lambda_a\lambda^2_b) e^{-\frac{\sigma^2_b}{\sigma_b+\lambda_b}}\\
    \Pr_{(A, B)\sim \pi} \left[ B \leq \ulT - \sigma -1\right]&\leq O(\lambda_b) e^{-\frac{\sigma^2}{2(\sigma+\lambda_b)}} + O(\lambda^2_a\lambda_b) e^{-\frac{\sigma^2_a}{\sigma_a+\lambda_a}}.
  \end{align*}
\end{lemma}
\begin{proof}
  By the balance equation~\eqref{eq:p2new1} with $X=\{(i, j)\mid 0\leq i\leq k, j\geq 0\}$ for any $k\geq 0$, it holds that
  \begin{align}\label{eq:p2_concentration2}
    \lambda_a \sum_{j=0}^{\infty} \pi (k, j) & = \sum_{j=0}^{\infty} \left(k+1 + j (1-(1-p)^{k+1}\right) \pi (k+1, j) \notag\\
    &\leq \left(k+1+(\lambda_b+ \sigma_b) (1-(1-p)^{k+1}\right) \sum_{j=0}^{\lambda_b + \sigma_b} \pi (k+1, j) 
     + \sum_{j=\lambda_b+\sigma_b+1}^{\infty} (k+1 + j ) \pi (k+1, j).
  \end{align}
  To simplify the notation, for any $k\geq 0$, define
  \begin{align*}
    g(k) =\sum_{j=0}^{\lambda_b + \sigma_b} \pi (k, j), \quad
    \alpha_k = \frac{k+(\lambda_b+ \sigma_b) (1-(1-p)^{k})}{\lambda_a}, \text{\quad and}\quad
    \beta_k = \frac{1}{\lambda_a}\sum_{j=\lambda_b+\sigma_b+1}^{\infty} (k + j ) \pi (k, j).
  \end{align*}
  Then, since the LHS of~\eqref{eq:p2_concentration2} is at least $\lambda_a g(k)$, \eqref{eq:p2_concentration2} implies the following recursive relationship for any $k\geq 0$:
  \[
    g(k)\leq \alpha_{k+1} g(k+1) + \beta_{k+1}.
  \]
  We observe that, for $0\leq k\leq \ukT$, 
  \begin{align*}
    \alpha_k &= \frac{k+(\lambda_b+ \sigma_b) (1-(1-p)^{k})-\ukT- (\lambda_b+ \sigma_b) (1-(1-p)^{\ukT})+\lambda_a}{\lambda_a}\\
    & \leq 1 - \frac{\ukT - k}{\lambda_a} \leq \exp \left(- \frac{\ukT - k}{\lambda_a}\right)  \leq \exp \left(- \frac{\ukT - k}{\lambda_a+k}\right).
  \end{align*}
  Moreover, by Lemma~\ref{lem:p2_A_B} with~\eqref{eq:p2-sum2}, when $\sigma_b =O(\lambda_b)$, 
  \begin{align*}
    \sum_{k=0}^{\ukT} \beta_k 
    & = \frac{1}{\lambda_a}\sum_{k=0}^{\ukT} k \sum_{j=\lambda_b+\sigma_b+1}^{\infty} \pi (k, j)+ \frac{1}{\lambda_a}\sum_{k=0}^{\ukT} \sum_{j=\lambda_b+\sigma_b+1}^{\infty} j \pi (k, j)\\
    & \leq \frac{\ukT}{\lambda_a}\sum_{k=0}^{\infty} \sum_{j=\lambda_b+\sigma_b+1}^{\infty} \pi (k, j)+ \frac{1}{\lambda_a}\sum_{k=0}^{\infty} \sum_{j=\lambda_b+\sigma_b+1}^{\infty} j \pi (k, j)\\
    & \leq O(\lambda_b) \exp \left(-\frac{\sigma^2_b}{\sigma_b+\lambda_b}\right) + \frac{1}{\lambda_a} O(\lambda_b^3) \exp \left(-\frac{\sigma^2_b}{\sigma_b+\lambda_b}\right)\\
    & \leq O(\lambda^2_b) \exp \left(-\frac{\sigma^2_b}{\sigma_b+\lambda_b}\right)
\end{align*}
  where the second inequality holds since $\ukT\leq \lambda_a$ and the last inequality follows from that $\lambda_a \geq \lambda_b$.
  Therefore, by Lemma~\ref{lem:recursion-extra} with $k^\ast=\ukT$ and $\eta = \lambda_a$, we have that
\begin{align*}
\sum_{k=0}^{\ukT - \sigma - 1} g(k)
 =  O(\lambda_a) e^{-\frac{\sigma^2}{2(\sigma+\lambda_a)}} + O(\lambda_a\lambda^2_b) e^{-\frac{\sigma^2_b}{\sigma_b+\lambda_b}}.
\end{align*}
This proves the first inequality of the lemma.
The second one is analogous.
%\qed
\end{proof}

Proposition~\ref{prop:2sideP_notS} easily follows from Lemmas~\ref{lem:p2_A_B} and \ref{lem:p2_LB} by defining $\sigma$ to be $\sigma_a$ or $\sigma_b$.

We are ready to bound the loss in the unbalanced case.
The following theorem, together with Lemma~\ref{lem:P2_Loss2ExpectedSizes}, provides upper bounds on $\mathbf{L}_a(\mathsf{Patient}_2)$ and $\mathbf{L}_b(\mathsf{Patient}_2)$, which shows the first part of Theorem~\ref{t1p}.

\begin{theorem}\label{thm:P2unbalanced}
  For a bipartite matching market $(d_a, d_b, p)$ with $d_a \geq d_b$, 
  let $\pi$ be the stationary distribution of the Markov chain of $\mathsf{Patient}_2$.
  Then it holds that
  %In a 2-sided matching market with patient agents, 
  \begin{align*}
    \frac{\mathbb{E}_{(A, B)\sim \pi}[A(1-p)^{B}]}{\lambda_a}
    &   \leq (1+o(1))\left( \frac{d_a - d_b}{d_a} + \frac{\log (d_b+3)}{d_a}\right) + o(1),\\
    \frac{\mathbb{E}_{(A, B)\sim \pi}[B(1-p)^{A}]}{\lambda_b}
    &   \leq (1+o(1)) e^{-\max\left\{d_a - d_b, \frac{d_a}{1+d_b}\right\}} + o(1).
%    \frac{\mathbb{E}_{(A, B)\sim \pi}[A(1-p)^{B} + B(1-p)^{A}]}{\lambda_a + \lambda_b}
%    &   \leq C_a\frac{d_a - d_b}{d_a+d_b} + C_a \frac{\log (d_b+3)}{d_a+d_b}+ C_b e^{-\max\left\{d_a - d_b, \frac{d_a}{1+d_b}\right\}} + o(1).
  \end{align*}
  \end{theorem}
  \begin{proof}
   Define $\sigma_a = \Theta (\sqrt{\lambda_a \log (\lambda_a\lambda_b)})$ and $\sigma_b = \Theta (\sqrt{\lambda_b \log (\lambda_a\lambda_b)})$.
   Let $S=\{(i,j)\mid \ukT - \sigma_a \leq i\leq \lambda_a+\sigma_a, \ulT - \sigma_b \leq j \leq \lambda_b+\sigma_b\}$.
    It follows that
    \begin{align}\label{eq:P2_loss11}
  \mathbb{E}_{(A, B)\sim \pi}[A(1-p)^{B}] 
  &\leq 
  \sum_{(i,j)\in S} i(1-p)^j \pi(i, j)+ \sum_{(i,j)\not\in S} i \pi(i, j).
  \end{align}
  The second term is bounded by
  \begin{align*}
    \sum_{(i,j)\not\in S} i \pi(i, j)
    &\leq (\lambda_a+\sigma_a)\sum_{(i,j)\not\in S} \pi(i, j)
    + \sum_{i\geq \lambda_a+\sigma_a+1} i\sum_{j\geq 0} \pi(i, j).
  \end{align*}
  By the definitions of $\sigma_a$ and $\sigma_b$, this is upper-bounded by a constant by Proposition~\ref{prop:2sideP_notS} with~\eqref{eq:p2-sum1}. %\dznote{on proposition 4.3 I have a question.}
  Moreover, the first term of~\eqref{eq:P2_loss11} is bounded by the definition of $S$ as follows:
  \begin{align*}
      \sum_{(i,j)\in S} i(1-p)^j \pi(i, j) \leq \max_{(i,j)\in S}\left(i(1-p)^j\right) \leq (\lambda_a + \sigma_a)(1-p)^{\ulT - \sigma_b}.
  \end{align*}
  Since $\ulT > \ell_2 - \sigma_a$ by Lemma~\ref{lem:2P_definition_k}, we have 
  \[
    (1-p)^{\ulT - \sigma_b} \leq (1-p)^{\ell_2 - \sigma_a- \sigma_b}\leq \frac{(1-p)^{\ell_2}}{1-\sigma_a p-\sigma_b p} \leq (1+o(1)) (1-p)^{\ell_2}.
  \]
  Since $\lambda_a (1-p)^{\ell_2} = \lambda_a - \lambda_b + \ell_2$ by the definition of $\ell_2$, we have
  \begin{align*}
      \mathbb{E}_{(A, B)\sim \pi}[A(1-p)^{B}]  &\leq (1+o(1))(\lambda_a + \sigma_a) (1-p)^{\ell_2}+ O(1)\\
      &\leq (1+o(1))(\lambda_a - \lambda_b + \ell_2 + o(\lambda_a))  + O(1).
  \end{align*}  
  We note that, since $\lambda_a = \frac{d_a}{d_b}\lambda_b$ and $d_a, d_b$ are constants, we may assume that $\lambda_a \leq C \lambda_b$ for some constant $C$, and hence $\sigma_a = o(\lambda_a)$.
  Therefore, since $\ell_2 \leq \log(d_b+3)/p$ by Lemma~\ref{G2:kstarlstar}, it holds that 
  \[
    \frac{\mathbb{E}_{(A, B)\sim \pi}[A(1-p)^{B}]}{\lambda_a}
    \leq (1+o(1))\left( \frac{d_a - d_b}{d_a} + \frac{\log (d_b+3)}{d_a}\right) + o(1).
  \]
  Similarly, it follows that
    \begin{align*}
  \mathbb{E}_{(A, B)\sim \pi}[B(1-p)^{A}] 
  &\leq \sum_{(i,j)\in S} j(1-p)^i \pi(i, j)+ \sum_{(i,j)\not\in S} j \pi(i, j)
  \leq \max_{(i,j)\in S}\left(j(1-p)^i\right) + O(1)\\
  &\leq (\lambda_b + \sigma_b)(1-p)^{\ukT-\sigma_a} + O(1),
  \end{align*}
  where the second inequality follows from Proposition~\ref{prop:2sideP_notS} with~\eqref{eq:p2-sum2},
  It follows from Lemma~\ref{lem:2P_definition_k} that $(1-p)^{\ukT-\sigma_a} = (1+o(1))(1-p)^{k_2}$. 
  Moreover, we have $(1-p)^{k_2} \leq e^{-pk_2}\leq e^{-\max\left\{d_a - d_b, \frac{d_a}{1+d_b}\right\}}$ by Lemma~\ref{G2:kstarlstar}.
  Hence, we obtain
  \begin{align*}
      \mathbb{E}_{(A, B)\sim \pi}[B(1-p)^{A}]  &\leq (1+o(1)) (\lambda_b + \sigma_b) (1-p)^{k_2}+ O(1)\\
      &\leq  (1+o(1)) (\lambda_b + o(\lambda_b)) e^{-\max\left\{d_a - d_b, \frac{d_a}{1+d_b}\right\}} + O(1).
  \end{align*}  
  Therefore, it holds that 
  \[
    \frac{\mathbb{E}_{(A, B)\sim \pi}[B(1-p)^{A}]}{\lambda_b}
    \leq (1+o(1))e^{-\max\left\{d_a - d_b, \frac{d_a}{1+d_b}\right\}} + o(1).
  \]
  This completes the proof.
%\qed
  \end{proof}

%%%%%%%%%%%%%%%%%%%%%%%%%%%%%%%%%%%%%%%%%%%%%%%%%%%%%%%%%%%%%%%%%%%%%%%%%%%
\subsubsection{Loss of the Balanced Case: Proofs of Propositions~\ref{prop:P2balancedSum} and~\ref{prop:P2balanced}}\label{sec:P2balanced}
%%%%%%%%%%%%%%%%%%%%%%%%%%%%%%%%%%%%%%%%%%%%%%%%%%%%%%%%%%%%%%%%%%%%%%%%%%%

We here bound the loss of $\mathsf{Patient}_2$ when $\lambda_a = \lambda_b$.
To this end, we first present further concentration bounds on the sum of the pool sizes, which follows even when $\lambda_a > \lambda_b$.
The proof can be done by taking another balance equation with $X=\{(i, j)\mid i+j \leq h, i\geq 0, j\geq 0\}$ for any $h\geq 0$.

\PTwoConcentrateSum*

\begin{proof}
  For any $h\geq 0$, the balance equation~\eqref{eq:balance} with $X=\{(i, j)\mid i+j \leq h, i\geq 0, j \geq 0\}$ implies that
  %and $\{(i, j)\mid i+j \geq h+1\}$, implying that 
%  \begin{multline*}
  \begin{align*}
      (\lambda_a+\lambda_b)\sum\limits_{i,j:i+j=h}\pi(i,j)
      =(h+1)\sum_{i,j:i+j=h+1}\pi(i,j)+\sum\limits_{i,j: i+j=h+2}\left(h+2-i(1-p)^j-j (1-p)^i\right)\pi(i,j).
  \end{align*}
%  \end{multline*}
  Hence we have
  \begin{align*}
      (\lambda_a+\lambda_b)\sum\limits_{i,j:i+j=h}\pi(i,j) \leq (2h+3)\max\left\{\sum\limits_{i,j:i+j=h+1}\pi(i,j), \sum\limits_{i,j:i+j=h+2}\pi(i,j)\right\}.
  \end{align*}
  We denote $g(h)=\sum\limits_{i,j:i+j=h}\pi(i,j)$ for $h\geq 0$.
  Then the above inequality can be written as 
  \begin{align*}
      (\lambda_a+\lambda_b)g(h)\leq (2h+3) \max \left\{ g(h+1), g(h+2) \right\}.
  \end{align*}
  We denote $m' = \frac{\lambda_a+\lambda_b}{2}-2$.
  For $h \leq m'$, it holds that
  \begin{align*}
      \frac{g(h)}{\max\left\{ g(h+1), g(h+2)\right\}} \leq \frac{2h+3}{\lambda_a+\lambda_b} = 1-\frac{\lambda_a+\lambda_b-(2h+3)}{\lambda_a+\lambda_b}
       \leq \exp \left( - \frac{m'-h}{(\lambda_a+\lambda_b)/2} \right).
  \end{align*}
  
  Let $n_0, n_1, \dots$ be a sequence of numbers defined as follows: $n_0=h$, $n_{i+1}= \argmax \left\{ g(n_{i}+1), g(n_{i}+2) \right\}$ for $i\geq 0$. 
  Since $h \leq m'$, the largest number $n_\ell$ in the sequence equals $m'$ or $m' -1$.
  Then we have
  \begin{align*}
      g(h)&=\frac{g(n_0)}{g(n_1)}\frac{g(n_1)}{g(n_2)}...\frac{g(n_{\ell-1})}{g(n_\ell)}g(n_\ell) \leq \exp \left( - \sum\limits_{j=0}^{\ell-1} \frac{m'-n_j}{\lambda_a+\lambda_b} \right) \\
      &\leq \exp \left( -\sum\limits_{j=0}^{\frac{m'-h}{2}} \frac{2j}{\lambda_a+\lambda_b} \right) \leq \exp \left( -\frac{1}{\lambda_a+\lambda_b} \frac{(m'-h)^2}{4} \right).
  \end{align*}
  Therefore, for any $\sigma \geq 1$, it holds that
  \begin{align*}
      \sum\limits_{h=0}^{m'-\sigma}g(h) &\leq \sum\limits_{h=0}^{m'-\sigma}\exp \left( -\frac{1}{\lambda_a+\lambda_b} \frac{(m'-h)^2}{4} \right)  
      \leq      \sum\limits_{h=\sigma}^{\infty}\exp \left( -\frac{1}{\lambda_a+\lambda_b} \frac{h^2}{4} \right) \\
      &\leq \frac{\exp \left( -\frac{\sigma^2}{4(\lambda_a+\lambda_b)} \right)}{\min \left( \frac{\sigma}{4(\lambda_a+\lambda_b)}, \frac{1}{2} \right)} =O(\lambda_a+\lambda_b)\exp \left( -\frac{\sigma^2}{4(\lambda_a+\lambda_b)} \right),
  \end{align*}
  where the third inequality follows from Lemma~\ref{l9} in Section~\ref{sec:recursion}.
  %\qed
\end{proof}

In what follows, we assume that $\lambda_a = \lambda_b$ and $\sigma_a = \sigma_b = O(\sqrt{\lambda_a \log \lambda_a})$.
Then $(\lambda_a+\lambda_b)/2 = \lambda_a$.

\PTwoConcentrateBalanced*

\begin{proof}
  We denote $\olambda_a = \lambda_a + \sigma_a$ and $\olambda_b = \lambda_b + \sigma_b$.
  Note that $\olambda_a=\olambda_b$ by the assumption.
  For any $z\geq \olambda_a/2$, 
  consider the balance equation with $X=\{(i, j)\mid i=j+z, i\geq 0, j\geq 0\}$, implying that 
  \begin{equation}\label{eq:p2_diagonal3balance}
    \sum_{i, j: i=j+z} \left(\lambda_a + j (1-p)^i\right) \pi (i, j) = \sum_{i, j: i = j + z+ 1} \left(\lambda_b + i (1-p)^j\right) \pi (i, j).  
  \end{equation}
  For $z\geq \olambda_a/2$, define %\dznote{Is $z$ a constant?}
  \begin{align*}
    M_z&=\left\{(i, j)\in \mathbb{Z}_+\times \mathbb{Z}_+\mid i=j+z, 0 \leq j \leq \olambda_b, z\leq i\leq \olambda_a \right\},\\
    \overline{M_z} & = \left\{(i, j)\in \mathbb{Z}_+\times \mathbb{Z}_+\mid i=j+z, (i, j)\not\in M_z \right\}.
  \end{align*}
  We also define
  \[
    g(z)  = \sum_{(i, j)\in M_z} \pi (i, j) \quad\text{and}\quad
    \beta_z  = (\lambda_a + \lambda_b)\sum_{(i, j)\in \overline{M_z}} \pi (i, j). 
  \]
  We first observe that the function $f(j)= j (1-p)^{j}$ is maximized when $j = - \frac{1}{\log(1-p)}$, whose value is at most $1/(pe)$.
  This implies that, for $(i, j)\in M_z$, $\lambda_a + j (1-p)^i\leq \lambda_a + \frac{1}{pe} (1-p)^z\leq \lambda_a + \lambda_b$.
  Hence the LHS of~\eqref{eq:p2_diagonal3balance} is at most 
  \begin{align*}
    \left(\lambda_a + \frac{1}{pe} (1-p)^{z}\right) g(z)+\beta_z
    \leq 
    \left(\lambda_a + \frac{1}{pe} (1-p)^{\olambda_a/2}\right) g(z)+\beta_z
  \end{align*}
%  Since $\ulT \geq - \frac{1}{\log(1-p)}$ if $d_a\geq 3$, $f$ is non-increasing when $j\geq \ell_2$.
  since $z\geq \olambda_a/2$.
  Moreover, since the function $i (1-p)^{i-z}$ is minimized when $i=\olambda_a$ for $z\leq i\leq \olambda_a$, the RHS of~\eqref{eq:p2_diagonal3balance} is at least
  \[
    \left(\lambda_b + \olambda_a (1-p)^{\olambda_a - z}\right) g(z+1)
    \geq 
    \left(\lambda_b + \olambda_a (1-p)^{\olambda_a/2}\right) g(z+1)
  \]
  since $z\geq \olambda_a/2$.
  Therefore, defining
  \[
    \alpha_z = \frac{\lambda_a + \frac{1}{pe} (1-p)^{\olambda_a/2}}{\lambda_b + \olambda_a (1-p)^{\olambda_a/2}}, 
  \]
  we obtain
  \[
    \alpha_z g(z) + \frac{1}{\lambda_b}\beta_z \geq g(z+1)  
  \]
  since $\lambda_b + \olambda_a (1-p)^{\olambda_a - z}\geq \lambda_b$.

  \begin{claim}\label{clm1}
    There exists some constant $c_d>0$ such that $\alpha_z < 1-c_d$ for any $z\geq \olambda_a/2$.
  \end{claim}
  \begin{proof}[Proof of Claim~\ref{clm1}]
    It holds that
    \begin{align*}
      \alpha_z &
      %= 1 - \frac{ \olambda_a (1-p)^{\olambda_a/2}- \frac{1}{pe} (1-p)^{\olambda_a/2}}{\lambda_b + \olambda_a (1-p)^{\olambda_a/2}}
      = 1 - \frac{ \olambda_a - \frac{1}{pe} }{\lambda_b + \olambda_a (1-p)^{\olambda_a/2}}(1-p)^{\olambda_a/2}.
    \end{align*}
    We now observe that
    \begin{align*}
      \olambda_a - \frac{1}{pe} & \geq \lambda_a \left(1-\frac{1}{d_a e}\right),\\
      \lambda_b + \olambda_a (1-p)^{\olambda_a/2} & \leq \lambda_b + \olambda_a e^{-p \lambda_a/2} \leq \lambda_b + \frac{\olambda_a}{2}\leq \lambda_a +\lambda_b,\\
      (1-p)^{\olambda_a /2} &\geq e^{-\lambda_a(p+p^2)} = e^{-(d_a + d_a p)} \geq e^{-2d_a},
    \end{align*}
    where the second one follows since $e^{-d_a/2}\leq 1/2$ if $d_a\geq 3$ and $\olambda_a \leq 2\lambda_a$ and the last one follows since $p<1/10$.
    Hence, we obtain
    \begin{align*}
      \alpha_z &
      \leq 1 - \frac{ 1-\frac{1}{d_a e} }{2} e^{-2d_a}.
    \end{align*}
    Therefore, there exists some constant $c_d>0$ such that $\alpha_z < 1-c_d$ for any $z\geq \olambda_a/2$.
    %\qed
  \end{proof} 

  Applying the inequality repeatedly, we have
  \begin{align*}
  g(z+1) 
    \leq \left(\prod_{i=\olambda_a/2}^{z} \alpha_i\right) g(k^\ast) + \frac{1}{\lambda_b}\sum_{i=\olambda_a/2}^{z} \left(\prod_{j=i+1}^{z} \alpha_j\right) \beta_i
   \leq (1-c_d)^{z-\olambda_a/2}  + \frac{1}{\lambda_b}\sum_{i=\olambda_a/2}^{z} \beta_i.
  \end{align*}
  Since 
  \begin{align*}
    \sum_{i=\olambda_a/2}^{z} \beta_i
    \leq (\lambda_a+\lambda_b) \sum_{i=\olambda_a/2}^{z} \sum_{(i, j)\in \overline{M_z}} \pi(i, j) 
    \leq (\lambda_a+\lambda_b) \sum_{i\geq \olambda_a} \sum_{j=0}^\infty \pi (i, j),
  \end{align*}
  and $\lambda_a=\lambda_b$, 
  we have
  \begin{align*}
  g(z+1) \leq (1-c_d)^{z-\olambda_a/2}  + 2 M.
  \end{align*}
  where $M = \sum_{i\geq \olambda_a} \sum_{j=0}^\infty \pi (i, j)$.
  
  We now estimate $\Pr_{(A, B)\sim \pi} \left[ A - B \geq \frac{\olambda_a}{2} + \sigma_d \right]$ for $\sigma_d\geq 1$, which is denoted by $P$ for simplicity.
  It holds that
  \begin{align*}
  P & \leq \sum_{z=\olambda_a/2 +\sigma_d}^{\olambda_a} g(z) + \sum_{i\geq \olambda_a} \sum_{j=0}^\infty \pi (i, j) 
  = \sum_{z=\olambda_a/2 +\sigma_d}^{\olambda_a} g(z) + M. 
  \end{align*}
  From the above discussion, we have 
  \[
  \sum_{z=\olambda_a/2 +\sigma_d}^{\olambda_a} g(z) \leq \sum_{z=\olambda_a/2 +\sigma_d}^{\olambda_a}(1-c_d)^{z-\olambda_a/2} + \olambda_a M. 
  \]
  We observe that
  \begin{align*}
    \sum_{z=\olambda_a/2 +\sigma_d}^{\olambda_a}(1-c_d)^{z-\olambda_a/2} \leq \sum_{z=\olambda_a/2 + \sigma_d}^{\olambda_a} e^{-c_d (z-\olambda_a/2)}
    \leq \sum_{z=\sigma_d}^{\infty} e^{-c_d z}  \leq \frac{e^{-c_d\sigma_d}}{1 - e^{-c_d}}.
  \end{align*}
  Therefore, we have
  \begin{align*}
   P \leq \frac{e^{-c_d\sigma_d}}{1 - e^{-c_d}} + (\olambda_a + 1)M.
  \end{align*}
  By Lemma~\ref{lem:p2_A_B}, it holds that $M =O\left(\lambda_a\right) e^{-\frac{\sigma_a^2}{\sigma_a +\lambda_a}}$.
  Thus the lemma holds.
% 
%  for some constant $C$ when $\sigma_a = \sigma_b = \Theta (\sqrt{\lambda_a \log \lambda_a})$.
%  Therefore, we have
%  \begin{align*}
%    \sum_{z=\olambda_a/2 + \Delta}^{\olambda_a -\ulT} g(z) 
%    \leq \frac{e^{-c_d\Delta}}{1 - e^{-c_d}} + \olambda_a (\lambda_a+\lambda_b) \sum_{(i, j)\not\in S} \pi(i, j).
%%    \leq \frac{e^{-c_d\Delta}}{1 - e^{-c_d}} + \frac{C}{2\lambda_a}
%  \end{align*}
%%  for some constant $C$.
  %\qed
\end{proof}

The above lemma implies that, if we set $\sigma_d = \frac{1}{c_d}\log \lambda_a$, then $\frac{e^{-c_d\sigma_d}}{1 - e^{-c_d}}=O(1/\lambda_a)$.

Recall Proposition~\ref{prop:2sideP_notS} that, for $\sigma_a, \sigma_b\geq 1$, define 
$S=\{(i,j)\mid \ukT - \sigma_a \leq i\leq \lambda_a+\sigma_a, \ulT - \sigma_b \leq j \leq \lambda_b+\sigma_b\}$.

\begin{theorem}\label{thm:P2balanced}
  For a bipartite matching market $(d_a, d_b, p)$ with $d_a = d_b$, 
  let $\pi$ be the stationary distribution of the Markov chain of $\mathsf{Patient}_2$.
  Then it holds that
  %In a 2-sided matching market with patient agents, 
  \begin{align*}
    \frac{\mathbb{E}_{(A, B)\sim \pi}[A(1-p)^{B}]}{\lambda_a}
    = \frac{\mathbb{E}_{(A, B)\sim \pi}[B(1-p)^{A}]}{\lambda_b}
    &   \leq (1+o(1)) e^{-Cd_a} + o(1)
%    \frac{\mathbb{E}_{(A, B)\sim \pi}[A(1-p)^{B} + B(1-p)^{A}]}{\lambda_a + \lambda_b}
%    &   \leq C_a\frac{d_a - d_b}{d_a+d_b} + C_a \frac{\log (d_b+3)}{p}+ C_b e^{-\max\left\{d_a - d_b, \frac{d_a}{1+d_b}\right\}} + o(1).
  \end{align*}
  for some constant $C$.
  \end{theorem}
  \begin{proof}
   Define $\sigma_a =\sigma_b= \Theta (\sqrt{\lambda_a \log \lambda_a})$, $\sigma_h = \Theta(\sqrt{\lambda_a \log \lambda_a})$, and $\sigma_d = \frac{1}{c_d}\log \lambda_a$ where $c_d$ is a constant in Proposition~\ref{prop:P2balanced}.
   Let 
   \[
     S'=\left\{(i,j)\in S \mid \lambda_a - 2 - \sigma_h \leq i+j, i-\frac{\olambda_a}{2}-\sigma_d\leq j \leq i+\frac{\olambda_a}{2}+\sigma_d\right\}.
    \]
  It follows that
    \begin{align}\label{eq:P2_loss1}
  \mathbb{E}_{(A, B)\sim \pi}[A(1-p)^{B}] 
  &\leq 
  \sum_{(k,j)\in S'} (k(1-p)^j) \pi(k, j)+ \sum_{(k,j)\not\in S'} k \pi(k, j).
%  \mathbb{E}_{(A, B)\sim \pi}[A(1-p)^{B} + B(1-p)^{A}] 
%  &\leq 
%  \sum_{(k,j)\in S} (k(1-p)^j+j(1-p)^k) \pi(k, j)+ \sum_{(k,j)\not\in S} (k+j) \pi(k, j).
  \end{align}
  The second term is bounded by
  \begin{align*}
    \sum_{(k,j)\not\in S} k \pi(k, j)
    &\leq (\lambda_a+\sigma_a)\sum_{(k,j)\not\in S'} \pi(k, j)
    + \sum_{k\geq \lambda_a+\sigma_a+1} k\sum_{j\geq 0} \pi(k, j).
  \end{align*}
  By the definitions of $\sigma_a$, $\sigma_b$ and $\sigma_h$, this is upper-bounded by a constant by Propositions~\ref{prop:2sideP_notS},~\ref{prop:P2balancedSum}, and~\ref{prop:P2balanced} with~\eqref{eq:p2-sum1} and~\eqref{eq:p2-sum2}.

  By the definition of $S'$, for any $(i, j)\in S$, we have
  \[
    \frac{1}{2}\left(\lambda_a-2-\sigma_h-\sigma_d-\frac{\olambda_a}{2}\right)\leq i \leq \olambda_a.
  \]
  Let $z = \frac{1}{2}(\lambda_a-2-\sigma_h-\sigma_d-\olambda_a/2)$.
  Since $\sigma_h = \Theta(\sqrt{\lambda_a \log \lambda_a})$ and $\sigma_d = O(\log \lambda_a)$, it holds that $z \geq C \lambda_a$ for some constant $C$.
  Hence the first term of~\eqref{eq:P2_loss1} is bounded as follows.
  \begin{align*}
      \sum_{(k,j)\in S} k(1-p)^j \pi(k, j) \leq \max_{(k,j)\in S}(k(1-p)^j) 
      \leq  \olambda_a(1-p)^{z}
      \leq  \lambda_a e^{-Cd}+o(\lambda_a).
%      \sum_{(k,j)\in S} (k(1-p)^j+j(1-p)^k) \pi(k, j) &\leq \max_{(k,j)\in S}(k(1-p)^j+j(1-p)^k) \\
%      &\leq 2 \olambda_a(1-p)^{z}
%      \leq 2 \lambda_a e^{-Cd}+o(\lambda_a).
  \end{align*}
  Therefore, 
  \[
    \frac{\mathbb{E}_{(A, B)\sim \pi}[A(1-p)^{B}]}{\lambda_a}
    \leq e^{-Cd}+o(1).
  \]
  This completes the proof.
  %\qed
\end{proof}

%%%%%%%%%%%%%%%%%%%%%%%%%%%%%%%%%%%%%%%%%%%%%%%%%%%%%%%%
\subsection{$\mathsf{Greedy}_1$}\label{sec:G1detail}
%%%%%%%%%%%%%%%%%%%%%%%%%%%%%%%%%%%%%%%%%%%%%%%%%%%%%%%%

%%%%%%%%%%%%%%%%%%%%%%%%%%%%%%%%%%%%%%%%%%%%%%%%%%%%%%%%
%\subsection{Proof Outline}
%%%%%%%%%%%%%%%%%%%%%%%%%%%%%%%%%%%%%%%%%%%%%%%%%%%%%%%%

The proof outline is similar to the 2-sided Greedy algorithm.
However, 1-sided algorithms require a more involved analysis to show the concentration.

Recall that, in $\mathsf{Greedy}_1$, we denote by $A_t$ the pool size of inactive agents at time $t$, and by $B_t$ the pool size of greedy agents at time $t$.
Similarly to the 2-sided Greedy algorithms, it suffices to bound the expected pool sizes $\mathbb{E}_{(A, B)\sim\pi}[A]$ and $\mathbb{E}_{(A, B)\sim\pi}[B]$ as below.
The proof of Lemma~\ref{lem:G1_Loss2ExpectedSizes} is presented in Section~\ref{sec:Greduction}.

\begin{lemma}\label{lem:G1_Loss2ExpectedSizes}
For a bipartite matching market $(d_a, d_b, p)$, 
let $\pi$ be the stationary distribution of the Markov chain of $\mathsf{Greedy}_1$.
For any $\epsilon > 0$ and $T >0$, it holds that
\begin{align*}
    \mathbf{L}_a(\mathsf{Greedy}_1)&\leq \frac{\mathbb{E}_{(A,B)\sim \pi}[A]}{\lambda_a}+
    \frac{\tau_{\textup{mix}}(\epsilon)}{T} + 6 \epsilon + \frac{1}{\lambda_a}2^{-6m},\\
    \mathbf{L}_b(\mathsf{Greedy}_1)&\leq \frac{\mathbb{E}_{(A,B)\sim \pi}[B]}{\lambda_b}+
    \frac{\tau_{\textup{mix}}(\epsilon)}{T} + 6 \epsilon + \frac{1}{\lambda_b}2^{-6m},
%    \mathbf{L}_2(\mathsf{Greedy}_1)\leq \frac{\mathbb{E}_{(A,B)\sim \pi}[A+B]}{m}+
%    \frac{\tau_{\textup{mix}}(\epsilon)}{T} + 6 \epsilon + \frac{1}{m}2^{-6m},
\end{align*}
where $\tau_{\textup{mix}}(\epsilon)$ is the mixing time of the Markov chain.
\end{lemma}

%%%%%%%%%%%%%%%%%%%%%%%%%%%%%%%%%%%%%%%%%%%%%%%%%%%
\paragraph{Markov Chain}
%%%%%%%%%%%%%%%%%%%%%%%%%%%%%%%%%%%%%%%%%%%%%%%%%%%

We here define a Markov chain $(A_t, B_t)$ on $\mathbb{Z}_+\times \mathbb{Z}_+$ for the 1-sided Greedy algorithms.
In contrast to $\mathsf{Greedy}_2$, agents in $U_t$ are inactive, and hence new agents in $U_t$ enter the market without making a pair. 
Specifically, for any pair of pool sizes $(k, j)$, the Markov chain transits only to the states $(k+1, j)$, $(k, j+1)$, $(k-1, j)$, and $(k, j-1)$.
See Figure~\ref{fig:greedy2markov}.
It transits to $(k+1, j)$ when a new agent arrives in $U_t$, to $(k, j+1)$ when a new agent arrives in $V_t$ and she does not get matched.
Moreover, it transits to $(k-1, j)$ when either a new agent arrives in $V_t$ and she gets matched with some agent in $U_t$, or some agent in $U_t$ perishes, and to $(k, j-1)$ when  some agent in $V_t$ perishes.
Therefore, we have
\begin{align*}
    r_{(k,j) \rightarrow (k+1,j)}&=\lambda_a, \\
    r_{(k,j) \rightarrow (k,j+1)}&=\lambda_b(1-p)^k,\\
    r_{(k,j) \rightarrow (k-1,j)}&=k+\lambda_b(1-(1-p)^k), \\
    r_{(k,j) \rightarrow (k,j-1)}&=j .
\end{align*}

%%%%%%%%%%%%%%%%%%%%%%%%%%%%%%%%%%%%%%%%%%%%%%%
\paragraph{Concentration of Pool Sizes}
%%%%%%%%%%%%%%%%%%%%%%%%%%%%%%%%%%%%%%%%%%%%%%%

We show that the pool size $A_t$ of inactive agents in the steady state is highly concentrated around some value $k_1$.
%\leq \max\{\lambda_a - \lambda_b, 0\}+\frac{\lambda_b \log (d_b+3)}{d_b}$.
On the other hand, letting $\ell_1 = (1-p)^{-\sigma_a}(\lambda_b - \lambda_a + k_1)$ for $\sigma_a\geq 1$, the probability that the pool size $B_t$ of greedy agents in the steady state is larger than $\ell_1$ is small.
More formally, we prove the following proposition.

\GOneConcentrate*

The first inequality of Proposition~\ref{prop:G1} can be shown in a similar way to Proposition~\ref{prop:G2:concentration}.
Under $\mathsf{Greedy}_1$, we further show that $\Pr_{(A, B)\sim \pi} \left[ A \leq k_1 - \sigma_a-1 \right]$ also decreases exponentially.
Thus the pool size $A_t$ of inactive agents in the steady state is highly concentrated around $k_1$.
Since $A_t$ is not so small in the steady state, a greedy agent in $V_t$ is likely to match an agent in $U_t$, and hence the pool size $B_t$ of greedy agents becomes small in the steady state. 
This intuition can be confirmed by taking the balance equations carefully.
We remark that we need to solve recursive equations with additive terms.
See Section~\ref{sec:proof_g1concentration} for the details.

If we set $\sigma_a = \Theta (\sqrt{\lambda_a \log (\lambda_a\lambda_b)})$ and $\sigma_b = \Theta (\sqrt{\lambda_b \log \lambda_b})$, Proposition~\ref{prop:G1}, together with Lemma~\ref{lem:G1_Loss2ExpectedSizes}, implies that
\begin{align*}
\mathbf{L}_a(\mathsf{Greedy}_1) &\approx \frac{1}{\lambda_a}\mathbb{E}_{(A, B)\sim\pi}[A] 
\leq \frac{1}{\lambda_a}\left(k_1 + \sigma_a \right)+o(1),\\
\mathbf{L}_b(\mathsf{Greedy}_1) &\approx \frac{1}{\lambda_b}\mathbb{E}_{(A, B)\sim\pi}[B] 
\leq \frac{1}{\lambda_b}\left(\ell_1  +\sigma_b\right)+o(1).
%\mathbf{L}(\mathsf{Greedy}_1) \approx \frac{1}{\lambda_a+\lambda_b}\mathbb{E}_{(A, B)\sim\pi}[A + B] 
%\approx \frac{1}{\lambda_a+\lambda_b}\left(k_1 + \ell_1 + \sigma_a +\sigma_b\right).
\end{align*}
%Since $\sigma_a = \Theta (\sqrt{\lambda_a \log \lambda_a})$, it follows that $(1-p)^{-\sigma_a}$ is $1+o(1)$.
%Hence $\ell_1=(1+o(1))k_1$ if  $\lambda_a \geq \lambda_b$ and $\ell_1 = (1+o(1))\left((\lambda_b - \lambda_a) + O\left(\frac{\lambda_b \log d_b}{d_b}\right)\right)$ otherwise.
This shows Theorem~\ref{t3} for $\mathsf{Greedy}_1$.

%%%%%%%%%%%%%%%%%%%%%%%%%%%%%%%%%%%%%%%%%%%%%%%%%%%%%%%%
\subsubsection{Concentration of Pool Sizes: Proof of Proposition~\ref{prop:G1}}\label{sec:proof_g1concentration}
%%%%%%%%%%%%%%%%%%%%%%%%%%%%%%%%%%%%%%%%%%%%%%%%%%%%%%%%

The goal of this subsection is to show Proposition~\ref{prop:G1}.

The number $k_1$ is defined to be the value that satisfies $\lambda_a=k_1+\lambda_b (1-(1-p)^{k_1})$.
Note that, if $\lambda_a \geq \lambda_b$, then $k_1$ is the same as $k_2$, and otherwise, $k_1$ is identical with $\ell_2$, where $k_2$ and $\ell_2$ are defined in Section~\ref{sec:UB:G2}.
We also note that $\ell_1 = \lambda_b(1-p)^{k_1- \sigma_a}=(1-p)^{-\sigma_a}(\lambda_b-\lambda_a+k_1)$ for $\sigma_a\geq 1$. 
The value $\ell_1$ depends on $\sigma_a$.

%To prove Proposition~\ref{prop:G1}, 
We first show the first part of Proposition~\ref{prop:G1} as below.
%the following, where the first inequality is identical to 
%The lemma shows that the pool size $A_t$ in the steady state is highly concentrated around $k_1$.

\begin{lemma}\label{lem:tB9}
For any $\sigma_a\geq 1$, it holds that
\begin{align*}
\Pr_{(A, B)\sim \pi} \left[ A \geq k_1 + \sigma_a+1 \right]\leq O\left(\lambda_a\right) e^{-\frac{\sigma_a^2}{\sigma_a+\lambda_a}}
\text{\quad and\quad}
\Pr_{(A, B)\sim \pi} \left[ A \leq k_1 - \sigma_a-1 \right]\leq O(\lambda_a) e^{-\frac{\sigma_a^2}{\lambda_a}}.
\end{align*}
\end{lemma}

\begin{proof}
By the balance equation~\eqref{eq:balance} with $X=\{(i, j)\mid 0\leq i\leq k, j\geq 0\}$ for any $k\geq 0$, it holds that
%and $\{(i, j)\mid i\geq k+1, j\geq 0\}$ for any $k\geq 0$, it holds that
\begin{align}\label{eq:1sideG_balanceZ1}
\lambda_a \sum_{j=0}^{\infty}  \pi (k, j) = \left(k+1+\lambda_b(1-(1-p)^{k+1})\right)\sum_{j=0}^{\infty} \pi (k+1, j).
\end{align}
Hence we obtain
\[
\frac{\sum_{j=0}^{\infty} \pi (k+1, j)}{\sum_{j=0}^{\infty} \pi (k, j)}\leq \frac{\lambda_a}{k+1+\lambda_b(1-(1-p)^{k+1})}.
\]
Recall that $k_1$ satisfies $\lambda_a= k_1 + \lambda_b (1-(1-p)^{k_1})$.
We can follow the proof of Proposition~\ref{prop:G2:concentration}, using Lemma~\ref{l7}~\eqref{eq:recursion-denom}, which implies that, for any $k\geq k_1$ and $\sigma_a\geq 1$, we have
\begin{align*}
\sum_{k=k_1+ \sigma_a+1}^\infty \sum_{j=0}^{\infty} \pi (k, j)
& = O\left(\frac{\sigma_a +\lambda_a}{\sigma_a}\right)\exp\left(-\frac{\sigma_a^2}{\sigma_a +\lambda_a} \right).
\end{align*}
This proves the first part of the lemma as $\sigma_a\geq 1$.

Similarly, for $k\leq k_1$, we have by~\eqref{eq:1sideG_balanceZ1}, 
\begin{align*}
\frac{\sum_{j=0}^{\infty} \pi (k-1, j)}{\sum_{j=0}^{\infty} \pi (k, j)} &= 
\frac{k+\lambda_b(1-(1-p)^k)}{\lambda_a}  
= \frac{k-k_1+\lambda_b(1-(1-p)^k)-\lambda_b (1-(1-p)^{k_1})+\lambda_a}{\lambda_a}\\
& \leq \frac{k-k_1+\lambda_a}{\lambda_a} \leq 1 - \frac{k_1-k}{\lambda_a}\leq e^{-\frac{k_1-k}{\lambda_a}}.
\end{align*}
Therefore, it follows from Lemma~\ref{lem:recursion-denom-dec} with $k^\ast = k_1$ and $\eta =\lambda_a$ that, for any $\sigma_a\geq 1$, we have
\begin{align*}
\sum_{k=0}^{k_1 - \sigma_a-1} \sum_{j=0}^{\infty} \pi (k, j)
& =O(\lambda_a) e^{-\frac{\sigma_a^2}{\lambda_a}}.
\end{align*}
This proves the second part of the lemma.
%\qed
\end{proof}

We remark that, by using Lemma~\ref{l7}~(ii) with $k^\ast = k_1$ and $\eta = \lambda_a -k_1$ in the above proof, it follows that,
%the balance equation~\eqref{eq:1sideG_balanceZ1} implies that, 
if $\sigma_a =O(\lambda_a)$, 
\begin{align}\label{eq:1sideG_Z2sum}
\sum_{k\geq k_1 + \sigma_a+1} k \sum_{j\geq 0} \pi(k, j) = O(\lambda_a^3) e^{-\frac{\sigma_a^2}{\sigma_a +\lambda_a}}.
\end{align}

We next prove that, if the pool size $A_t$ in the steady state is larger than about $k_1$, the pool size $B_t$ cannot be larger than about $\ell_1$.

\begin{lemma}\label{lem:tB10}
For $\sigma_a\geq 1$, we define $\ell_1 = \lambda_b(1-p)^{k_1- \sigma_a}$.
Then, for any $\sigma_b\geq 1$, we have
\begin{align*}
\Pr_{(A, B)\sim \pi} \left[ B \geq \ell_1 + \sigma_b +1, \ A \geq k_1- \sigma_a \right]
&\leq O(\lambda_b) e^{-\frac{\sigma_b^2}{2(\sigma_b+\lambda_b)}} + O(\lambda_a\lambda^2_b) e^{-\frac{\sigma_a^2}{\lambda_a}}.
\end{align*}
\end{lemma}

\begin{proof}
For any $j\geq 0$, the balance equation~\eqref{eq:balance} with $X=\{(i, \ell)\mid i\geq 0, 0\leq \ell \leq j\}$ implies that
\begin{align}\label{e15}
\sum_{k=0}^{\infty} \lambda_b(1-p)^k \pi (k, j) = (j+1)\sum_{k=0}^{\infty} \pi (k, j+1).
\end{align}
Define $K=k_1 - \sigma_a$.
The LHS can be transformed as follows.
\begin{align*}
\sum_{k=0}^{\infty} \lambda_b(1-p)^k \pi (k, j) 
&= \sum_{k=K}^\infty \lambda_b(1-p)^k \pi (k, j) + \sum_{k=0}^{K-1} \lambda_b(1-p)^k \pi (k, j)\\
&\leq \lambda_b(1-p)^{K} \sum_{k=K}^\infty  \pi (k, j) + \lambda_b \sum_{k=0}^{K-1} \pi (k, j).
\end{align*} 
To simplify the notation, define
\[
g(j)=\sum_{k=K}^\infty  \pi (k, j), \quad \alpha_j = \frac{\lambda_b(1-p)^{K}}{j+1} \text{\quad and} \quad \beta_j = \frac{\lambda_b \sum_{k=0}^{K-1} \pi (k, j) }{j+1}.
\]
Note that $\alpha_j=\ell_1/(j+1)$.
Then \eqref{e15} implies the following recursive relationship for any $j\geq 0$:
\[
g(j+1)\leq \alpha_j g(j) + \beta_j.
\]
We observe that, for $j\geq \ell_1-1$, 
\[
\alpha_j = \frac{\ell_1}{j+1} = 1 - \frac{j - \ell_1+1}{j+1}\leq \exp\left(-\frac{j-\ell_1+1}{j+1}\right).
\]
Moreover, by Theorem~\ref{lem:tB9}, 
\begin{align*}
\sum_{j=\ell_1}^\infty \beta_j & \leq \frac{\lambda_b}{\ell_1 + 1}  \sum_{k=0}^{K-1} \sum_{j=0}^\infty \pi (k, j)
 \leq O(\lambda_a\lambda_b) e^{-\frac{\sigma_a^2}{\lambda_a}}.
\end{align*}
Therefore, by Lemma~\ref{lem:recursion-extra} with $k^\ast=\ell_1-1$ and $\eta = 1$, we have that
\begin{align*}
\sum_{j=\ell_1 + \sigma_b+1}^\infty g(j)
 =  O(\ell_1) e^{-\frac{\sigma_b^2}{2(\sigma_b+\ell_1)}} + O(\ell_1\lambda_a\lambda_b) e^{-\frac{\sigma_a^2}{\lambda_a}}
 =  O(\lambda_b) e^{-\frac{\sigma_b^2}{2(\sigma_b+\lambda_b)}} + O(\lambda_a\lambda^2_b) e^{-\frac{\sigma_a^2}{\lambda_a}},
\end{align*}
since $\ell_1 \leq \lambda_b$.
This proves the lemma.
%\qed
\end{proof}

The above lemma, together with Lemma~\ref{lem:tB9}, implies the second part of Proposition~\ref{prop:G1}.

The balance equation~\eqref{e15} in the proof of Lemma~\ref{lem:tB10} implies the following, which will be used in bounding the loss.

\begin{lemma}\label{lem:1sideG_Z2-sum}
For any $\sigma_b\geq 1$, it holds that
  \begin{align*}
  \sum_{j=\lambda_b +\sigma_b+1}^\infty j \sum_{k=0}^{\infty} \pi (k, j) &= O\left(\lambda^3_b\right) e^{-\frac{\sigma_b^2}{\sigma_b +\lambda_b} }.
  \end{align*}
\end{lemma}
\begin{proof}
The balance equation~\eqref{e15} in the proof of Lemma~\ref{lem:tB10} implies that
\begin{align*}
\lambda_b \sum_{k=0}^{\infty}  \pi (k, j) \geq (j+1)\sum_{k=0}^{\infty} \pi (k, j+1).
\end{align*}
Then, for $j\geq \lambda_b$, it holds that
\[
\frac{\sum_{k=0}^{\infty} \pi (k, j+1)}{\sum_{k=0}^{\infty} \pi (k, j)} 
\leq \frac{\lambda_b}{j+1}
= 1 - \frac{j-\lambda_b+1}{j+1}\leq \exp\left(- \frac{j-\lambda_b+1}{j+1}\right).
\]
By Lemma~\ref{l7} with $k^\ast = \lambda_b-1$ and $\eta=1$, we have
  \begin{align*}
  \sum_{j=\lambda_b +\sigma_b+1}^\infty j \sum_{k=0}^{\infty} \pi (k, j) &= O\left(\lambda^3_b\right) \exp\left(-\frac{\sigma_b^2}{\sigma_b +\lambda_b} \right).
  \end{align*}
  %\qed
\end{proof}

%%%%%%%%%%%%%%%%%%%%%%%%%%%%%%%%%%%%%%%%%%%%%%%%%%%%%%%%
\subsubsection{Bounding the Loss of $\mathsf{Greedy}_1$}
%%%%%%%%%%%%%%%%%%%%%%%%%%%%%%%%%%%%%%%%%%%%%%%%%%%%%%%%

It follows from Proposition~\ref{prop:G1} that the loss of $\mathsf{Greedy}_1$ can be bounded.
This proves Theorem~\ref{t3} for $\mathsf{Greedy}_1$.
We remark that, since $\lambda_a = \frac{d_a}{d_b}\lambda_b$ and $d_a, d_b$ are constants, we may assume that $\lambda_a \leq C \lambda_b$ for some constant $C$.

\begin{theorem}\label{15}
For a bipartite matching market $(d_a, d_b, p)$, 
let $\pi$ be the stationary distribution of the Markov chain of $\mathsf{Greedy}_1$.
Then, if $d_a\geq d_b$, it holds that
\begin{align*}
    \frac{\mathbb{E}_{(A,B)\sim \pi}[A]}{\lambda_a} &\leq  \frac{d_a-  d_b}{d_a}+  \frac{\log (d_b+3)}{d_a} + o(1),\\
    \frac{\mathbb{E}_{(A,B)\sim \pi}[B]}{\lambda_b} &\leq (1+o(1))\frac{\log (d_b+3)}{d_b} + o(1).
\end{align*}
If $d_a < d_b$, then
\begin{align*}
    \frac{\mathbb{E}_{(A,B)\sim \pi}[A]}{\lambda_a} &\leq  \frac{\log (d_b+3)}{d_a}  + o(1),\\
    \frac{\mathbb{E}_{(A,B)\sim \pi}[B]}{\lambda_b} &\leq (1+o(1))\left(\frac{d_b-  d_a}{d_b}+  \frac{\log (d_b+3)}{d_b}\right) + o(1).
\end{align*}
\end{theorem}
\begin{proof}
Define $\sigma_a = \Theta\left(\sqrt{\lambda_a \log (\lambda_a\lambda_b)}\right)$ and $\sigma_b = \Theta\left(\sqrt{\lambda_b \log \lambda_b }\right)$. 
It follows that
\begin{align*}
    \mathbb{E}_{(A, B)\sim \pi}[A]  = \sum_{k\geq 0} k \sum_{j\geq 0} \pi(k, j)
        \leq k_1 +\sigma_a + \sum_{k\geq k_1+\sigma_a+1} k \sum_{j \geq 0} \pi(k, j).
\end{align*}
Since $\sigma_a=O(\lambda_a)$, the last term is $O(1)$ by~\eqref{eq:1sideG_Z2sum}.
Hence $\mathbb{E}_{(A, B)\sim \pi}[A]$ is upper-bounded by    
$k_1 + O(\sqrt{\lambda_a \log (\lambda_a\lambda_b)})$.
Therefore, we obtain
\begin{align}\label{eq:G1lossA}
    \frac{\mathbb{E}_{(A, B)\sim \pi}[A]}{\lambda_a} 
    \leq
    \frac{k_1}{\lambda_a} +o(1).
\end{align}

Similarly, it holds that
\begin{align*}
\mathbb{E}_{(A, B)\sim \pi}[B] &= \sum_{j\geq 0} j \sum_{k\geq 0} \pi(k, j) 
\leq 
\ell_1 + \sigma_b + (\lambda_b +\sigma_b)\sum_{j= \ell_1 + \sigma_b+1}^{\lambda_b +\sigma_b} \sum_{k\geq 0} \pi(k, j) + \sum_{j\geq \lambda_b +\sigma_b+1} j \sum_{k\geq 0} \pi(k, j)\\
&\leq 
\ell_1 + \sigma_b + (\lambda_b +\sigma_b)\sum_{j= \ell_1 + \sigma_b+1}^{\infty} \sum_{k\geq 0} \pi(k, j) + O(1) 
\leq \ell_1 + O(\sqrt{\lambda_b \log \lambda_b}),
\end{align*}
where the second inequality follows from Lemma~\ref{lem:1sideG_Z2-sum} and the third inequality follows since the second term is a constant by Proposition~\ref{prop:G1}. 
Since $\ell_1 = \lambda_b (1-p)^{k_1-\sigma_a}$, we have
\begin{align*}
    \mathbb{E}_{(A, B)\sim \pi}[B] 
    \leq \lambda_b (1-p)^{k_1-\sigma_a} +o(\lambda_b).
\end{align*}
We note that $\lambda_b (1-p)^{k_1} = \lambda_b-\lambda_a+k_1$ by definition of $k_1$, and $(1-p)^{-\sigma_a} \leq \frac{1}{1-p\sigma_a} =1+o(1)$.
Hence we obtain
\begin{align}\label{eq:G1lossB}
    \frac{\mathbb{E}_{(A, B)\sim \pi} [B]}{\lambda_b}
    &\leq (1+o(1)) \frac{\lambda_b-\lambda_a+k_1}{\lambda_b} + o(1). 
\end{align}

We now express~\eqref{eq:G1lossA} and~\eqref{eq:G1lossB} using $d_a$ and $d_b$ with Lemma~\ref{G2:kstarlstar}.
If $\lambda_a\geq \lambda_b$, then $k_1=k_2$, and hence, by Lemma~\ref{G2:kstarlstar}, we see that $k_1\leq \lambda_a -\lambda_b +  \frac{1}{p}\log (d_b+3)$.
Therefore, we have
\begin{align*}
    \frac{\mathbb{E}_{(A, B)\sim \pi} [A]}{\lambda_a}
\leq \frac{d_a - d_b}{d_a} + \frac{\log (d_b+3)}{d_a} + o(1)
\text{\quad and\quad}
    \frac{\mathbb{E}_{(A, B)\sim \pi} [B]}{\lambda_b}
\leq (1+o(1))\frac{\log (d_b+3)}{d_b} + o(1).
\end{align*}
Otherwise, i.e., if $\lambda_a< \lambda_b$, then we see that $k_1=\ell_2  \leq \frac{1}{p} \log (d_b+3)$ by Lemma~\ref{G2:kstarlstar}, implying that
\begin{align*}
    \frac{\mathbb{E}_{(A, B)\sim \pi} [A]}{\lambda_a}
\leq \frac{\log (d_b+3)}{d_a} + o(1)
\text{\quad and\quad}
    \frac{\mathbb{E}_{(A, B)\sim \pi} [B]}{\lambda_b}
 \leq (1+o(1))\left( \frac{d_b-  d_a}{d_b}+ \frac{\log (d_b+3)}{d_b}\right) + o(1).
\end{align*}
This completes the proof.
%\qed
\end{proof}

%%%%%%%%%%%%%%%%%%%%%%%%%%%%%%%%%%%%%%%%%%%%%%%%%%%%%%%%
\subsection{$\mathsf{Patient}_1$}\label{sec:P1detail}
%%%%%%%%%%%%%%%%%%%%%%%%%%%%%%%%%%%%%%%%%%%%%%%%%%%%%%%%

%%%%%%%%%%%%%%%%%%%%%%%%%%%%%%%%%%%%%%%%%%%%%%%%%%%%%%%%
%\subsection{Proof Outline}
%%%%%%%%%%%%%%%%%%%%%%%%%%%%%%%%%%%%%%%%%%%%%%%%%%%%%%%%

Recall that, in $\mathsf{Patient}_1$, we denote by $A_t$ the pool size of inactive agents at time $t$, and by $B_t$ the pool size of greedy agents at time $t$.
Similarly to the other algorithms, it suffices to bound the expected pool sizes in the steady state.

\begin{lemma}\label{lem:P1_Loss2ExpectedSizes}
For a bipartite matching market $(d_a, d_b, p)$, 
let $\pi$ be the stationary distribution of the Markov chain of $\mathsf{Patient}_1$.
For any $\epsilon > 0$ and $T >0$,
\begin{align*}
    \mathbf{L}_a(\mathsf{Patient}_1)
    &\leq \frac{\mathrm{E}_{(A,B)\sim \pi}[A]}{\lambda_a} + \frac{\tau_{\textup{mix}}(\epsilon)}{T} + 6\epsilon+\frac{2^{-6\lambda_a}}{\lambda_a},\\
    \mathbf{L}_b(\mathsf{Patient}_1)
    &\leq \frac{\mathrm{E}_{(A,B)\sim \pi}[B(1-p)^{A}]}{\lambda_b} + \frac{\tau_{\textup{mix}}(\epsilon)}{T} + 6\frac{\epsilon}{p}+\frac{2^{-6\lambda_b}}{\lambda_b},
\end{align*}
%\[
%    \mathbf{L}(\mathsf{Patient}_1)
%    \leq \frac{\mathrm{E}_{(A,B)\sim \pi}[A+B(1-p)^{A}]}{\lambda_a+\lambda_b} + \frac{\tau_{\textup{mix}}(\epsilon)}{T} + 6\frac{\epsilon}{p}+\frac{2^{-6\lambda_a}+2^{-6\lambda_b}}{\lambda_a+\lambda_b}, 
%\]
where $\tau_{\textup{mix}}(\epsilon)$ is the mixing time of the Markov chain.
\end{lemma}
\noindent
The proof of Lemma~\ref{lem:P1_Loss2ExpectedSizes} is given in Section~\ref{sec:P2reduction}.

%%%%%%%%%%%%%%%%%%%%%%%%%%%%%%%%%%%%%%%%%%%%%%%%%%%%%%%%
\paragraph{Markov Chain}
%%%%%%%%%%%%%%%%%%%%%%%%%%%%%%%%%%%%%%%%%%%%%%%%%%%%%%%%

We here define a Markov chain $(A_t, B_t)$ on on $\mathbb{Z}_+\times \mathbb{Z}_+$ for the 1-sided Patient algorithms.
In contrast to $\mathsf{Patient}_2$, agents in $U_t$ are inactive, and hence new agents in $U_t$ stay in the market without making a pair.
Specifically, for any pair of pool sizes $(k, j)$, 
the Markov chain transits only to the states $(k+1, j)$, $(k, j+1)$, $(k-1, j)$, $(k, j-1)$, and $(k-1, j-1)$.
See Figure~\ref{fig:patient2markov}.
It transits to $(k+1, j)$ or $(k, j+1)$ when a new agent arrives, 
to $(k-1, j)$ or $(k, j-1)$ when some agent in the pool leaves the market without getting matched to another agent, and to $(k-1, j-1)$ when some agent in $V_t$ leaves the market with getting matched to another agent.
Thus we have
\begin{align*}
    r_{(k,j) \rightarrow (k+1,j)}&=\lambda_a, \\
    r_{(k,j) \rightarrow (k,j+1)}&=\lambda_b, \\
    r_{(k,j) \rightarrow (k-1,j)}&=k, \\
    r_{(k,j) \rightarrow (k-1,j-1)}&=j(1-(1-p)^k), \\
    r_{(k,j) \rightarrow (k,j-1)}&=j(1-p)^k. 
\end{align*}

%%%%%%%%%%%%%%%%%%%%%%%%%%%%%%%%%%%%%%%%%%%%%%%
\paragraph{Concentration of Pool Sizes}
%%%%%%%%%%%%%%%%%%%%%%%%%%%%%%%%%%%%%%%%%%%%%%%

For $\mathsf{Patient}_1$, we prove that the pool size $(A_t, B_t)$ in the steady state is highly concentrated around $(k_1, \lambda_b)$, where we recall that $k_1$ is defined to be the value that satisfies $\lambda_a=k_1+\lambda_b (1-(1-p)^{k_1})$.
See Section~\ref{sec:G1detail}.

\POneConcentrate*

%Interestingly, for both of $\mathsf{Greedy}_1$ and $\mathsf{Patient}_1$, the size of inactive agents has the same concentrated value $k_1$.

To prove Proposition~\ref{prop:1sideP_notS2}, 
we first show that the pool size $B_t$ of patient agents is highly concentrated around $\lambda_b$ in the steady state.
This means the pool size of patient agents remains large in the steady state.
Therefore, since many patient agents attempt to match inactive agents, the number of inactive agents becomes small, which concentrates around $k_1$.
The details are given in Section~\ref{sec:proof_p1concentration}.

If we set $\sigma_a = \Theta (\sqrt{\lambda_a \log \lambda_a})$ and $\sigma_b = \Theta (\sqrt{\lambda_b \log (\lambda_a\lambda_b}))$, then Proposition~\ref{prop:1sideP_notS2} and Lemma~\ref{lem:P1_Loss2ExpectedSizes} imply that
\begin{align*}
\mathbf{L}_a(\mathsf{Patient}_1) &\approx \frac{1}{\lambda_a}\mathbb{E}_{(A, B)\sim\pi}[A] \leq \frac{1}{\lambda_a}\left(k_1 + o(\lambda_a)\right),\\
\mathbf{L}_b(\mathsf{Patient}_1) &\approx \frac{1}{\lambda_b}\mathbb{E}_{(A, B)\sim\pi}[B(1-p)^A] \leq \frac{1}{\lambda_b}\left(\lambda_b (1-p)^{k_1-o(\lambda_b)}\right),
%\mathbf{L}(\mathsf{Patient}_1) \approx \frac{1}{\lambda_a+\lambda_b}\mathbb{E}_{(A, B)\sim\pi}[A + B(1-p)^A] \leq \frac{1}{\lambda_a+\lambda_b}\left(k_1 +o(\lambda_b)+ \lambda_b (1-p)^{k_1-o(\lambda_b)}\right),
\end{align*}
which shows Theorem~\ref{t3} for $\mathsf{Patient}_1$.
See Section~\ref{sec:p1_loss} for the complete proof.

%%%%%%%%%%%%%%%%%%%%%%%%%%%%%%%%%%%%%%%%%%%%%%%%%%%%%%%%
\subsubsection{Concentration of Pool Sizes: Proof of Proposition~\ref{prop:1sideP_notS2}}\label{sec:proof_p1concentration}
%%%%%%%%%%%%%%%%%%%%%%%%%%%%%%%%%%%%%%%%%%%%%%%%%%%%%%%%

In this subsection, we show Proposition~\ref{prop:1sideP_notS2}.

For $\sigma_b\geq 1$, 
define $\ok$ by the value that satisfies $\lambda_a=\ok+(\lambda_b -\sigma_b)(1-(1-p)^{\ok})$, and 
$\uk$ by the value that satisfies $\lambda_a=\uk+(\lambda_b +\sigma_b)(1-(1-p)^{\uk})$.
Note that $\ok$ and $\uk$ depend on $\sigma_b$.

We first observe the relationship between $\ok$, $\uk$, and $k_1$.
By definition, $\ok >k_1 > \uk$ holds.

\begin{lemma}\label{l4}
For $\sigma_b\geq 1$, define $\ok$ and $\uk$ as above.
Then it holds that
$k_1-\sigma_b < \uk < k_1 < \ok < \min\{k_1+\sigma_b, \lambda_a\}.$
\end{lemma}

\begin{proof}
Define a function $f(k) = k+(\lambda_b -\sigma_b)(1-(1-p)^k)-\lambda_a$, which is a non-decreasing function. Note that $f(\ok)=0$. Since $f(k_1+\sigma_b)=\sigma_b(1-p)^{k_1}>0$ from the definition of $k_1$, we obtain $\ok \leq k_1+\sigma_b$. Since $f(\lambda_a)>0$ and $f(k_1)<0$, we also have $k_1\leq \ok\leq \lambda_a$.

The inequality for $\ok$ holds in a similar way.
If we define a function $f'(k)$ by $k+(\lambda_b +\sigma_b)(1-(1-p)^k)-\lambda_1$, then $f'$ is a non-decreasing function with $f(\uk)=0$, and we have
$f'(k_1-\sigma_b)=-\sigma_b(1-p)^{k_1}<0$. Also, $\uk \leq k_1$ holds since $f'(k_1) = \sigma_b (1-(1-p)^k) >0$.
%\qed
\end{proof}

In the rest of this section, we prove~\eqref{eq::1sideP_notS2} in Proposition~\ref{prop:1sideP_notS2}.
We first show that the pool size $B_t$ of patient agents in the steady state is highly concentrated around $\lambda_b$.

\begin{lemma}\label{lem:1sidePatientCon}
For any $\sigma_b\geq 1$, we have
\begin{align*}
\Pr_{(A, B)\sim \pi} \left[ B \geq \lambda_b + \sigma_b+1 \right]\leq O\left(\lambda_b\right) e^{-\frac{\sigma_b^2}{\sigma_b+\lambda_b}}
\text{\quad and\quad}
\Pr_{(A, B)\sim \pi} \left[ B \leq \lambda_b - \sigma_b-1 \right]\leq O\left(\lambda_b\right) e^{-\frac{\sigma_b^2}{\lambda_b}}.
\end{align*}
\end{lemma}
\begin{proof}
By the balance equation~\eqref{eq:balance} with $X=\{(i, \ell)\mid i\geq 0, 0\leq \ell \leq j\}$ for $j\geq 0$, it holds that
%and $\{(i, \ell)\mid i\geq 0, \ell \geq j+1\}$ for $j\geq 0$, it holds that
\begin{align}\label{eq:P1balance}
\lambda_b \sum_{k=0}^{\infty}  \pi (k, j) = (j+1) \sum_{k=0}^{\infty} \pi (k, j+1).
\end{align}
If $j\geq \lambda_b-1$, then we obtain
\[
\frac{\sum_{k=0}^{\infty} \pi (k, j+1)}{\sum_{k=0}^{\infty} \pi (k, j)} =\frac{\lambda_b}{j+1} = 1- \frac{j - \lambda_b+1}{j+1} \leq \exp\left(- \frac{j - \lambda_b+1}{j+1}\right).
\]
Therefore, by Lemma~\ref{l7}~(i) with $k^\ast=\lambda_b-1$ and $\eta = 1$, we have, for any $\sigma_b \geq 1$, 
\[
\sum_{j=\lambda_b + \sigma_b+1}^\infty \sum_{k=0}^{\infty}  \pi (k, j) = O\left(\frac{\lambda_b+\sigma_b}{\sigma_b}\right) \exp\left(-\frac{\sigma_b^2}{\sigma_b +\lambda_b} \right).
\]
This proves the first part of the lemma as $\sigma_b\geq 1$.

Moreover, if $j\leq \lambda_b$, it holds by~\eqref{eq:P1balance} that
\[
\frac{\sum_{k=0}^{\infty} \pi (k, j-1)}{\sum_{k=0}^{\infty} \pi (k, j)} = \frac{j}{\lambda_b} = 1- \frac{\lambda_b-j}{\lambda_b}
\leq \exp\left(- \frac{\lambda_b-j}{\lambda_b}\right).
\]
By Lemma~\ref{lem:recursion-denom-dec} with $k^\ast = \eta = \lambda_b$, for any $\sigma_b\geq 1$, 
\[
\sum_{j=0}^{\lambda_b - \sigma_b-1}\sum_{k=0}^{\infty}  \pi (k, j) = O(\lambda_b) e^{-\frac{\sigma_b^2}{\lambda_b}},
\]
which proves the second part of the lemma.
%\qed
\end{proof}

We remark that, using Lemma~\ref{l7}~(ii) to~\eqref{eq:P1balance} in the proof of Lemma~\ref{lem:1sidePatientCon}, we obtain
\begin{align}\label{eq:1sideP_Z2-sum}
\sum_{j=\lambda_b + \sigma_b+1}^\infty j \sum_{k=0}^{\infty}  \pi (k, j) = O\left(\lambda^3_b\right) \exp\left(-\frac{\sigma_b^2}{\sigma_b +\lambda_b} \right).
\end{align}

We next prove that, if $B_t$ is larger than about $\lambda_b$, the pool size $A_t$ cannot be larger than about $\ok$ in the steady state.

\begin{lemma}\label{lem:t13}
For any $\sigma_a\geq 1$ and $\sigma_b\geq 1$, we have
\begin{align*}
\Pr_{(A, B)\sim \pi} \left[ A \geq \ok + \sigma_a +1, \ B\geq \lambda_b - \sigma_b \right]&\leq 
O(\lambda_a) e^{-\frac{\sigma_a^2}{2(\sigma_a+\lambda_a)}} + O(\lambda_a^2\lambda_b) e^{-\frac{\sigma_b^2}{\lambda_b}}.
\end{align*}
\end{lemma}

\begin{proof}
By the balance equation~\eqref{eq:balance} with $X=\{(i, j)\mid 0\leq i\leq k, j\geq 0\}$ for $k\geq 0$, it holds that
%and $\{(i, j)\mid i\geq k+1, j\geq 0\}$ for $k\geq 0$, 
\begin{align}
\lambda_a \sum_{j=0}^{\infty}  \pi (k, j) &= \sum_{j=0}^{\infty}(k+1+j(1-(1-p)^{k+1})) \pi (k+1, j)\label{eq:1sideP-balance}\\
&\geq \sum_{j=\lambda_b - \sigma_b}^\infty (k+1+j(1-(1-p)^{k+1})) \pi (k+1, j)\nonumber \\
&\geq (k+1+(\lambda_b -\sigma_b)(1-(1-p)^{k+1})) \sum_{j=\lambda_b - \sigma_b}^\infty \pi (k+1, j).\nonumber 
\end{align}
The LHS is at most
\[
\lambda_a \sum_{j=0}^{\infty}  \pi (k, j) \leq \lambda_a \sum_{j=\lambda_b - \sigma_b}^\infty  \pi (k, j) + \beta_k,
\text{\ \ where\ \ }
\beta_k = \lambda_a \sum_{j=0}^{\lambda_b - \sigma_b-1} \pi (k, j).
\]
Let 
\[
g(k)=\sum_{j=\lambda_b - \sigma_b}^\infty  \pi (k, j)
\text{\quad and\quad}
\alpha_k = \frac{\lambda_a}{k+1+(\lambda_b -\sigma_b)(1-(1-p)^{k+1})}.
\]
Then, if $k\geq \ok$, we have
\[
\alpha_k g(k) + \beta_k \geq g(k+1).
\]
We observe that 
\[
\alpha_k \leq \frac{\lambda_a}{k-\ok+\lambda_a}\leq \exp\left(-\frac{k-\ok}{k-\ok+\lambda_a}\right).
\]
Moreover, by Lemma~\ref{lem:1sidePatientCon}, 
\[
\sum_{k=\ok}^\infty \beta_k \leq \lambda_a \sum_{j=0}^{\lambda_b - \sigma_b-1}\sum_{k=0}^{\infty}  \pi (k, j)  = O(\lambda_a\lambda_b) e^{-\frac{\sigma_b^2}{\lambda_b}}.
\]
Therefore, by applying Lemma~\ref{lem:recursion-extra} with $k^\ast = k_1$ and $\eta = \lambda_a - \ok$, we have 
\begin{align*}
\sum_{k=\ok+\sigma+1}^\infty g(k) 
 =  O(\lambda_a) e^{-\frac{\sigma_a^2}{2(\sigma_a+\lambda_a)}} + O(\lambda_a) \sum_{k=\ok}^\infty \beta_k
 =  O(\lambda_a) e^{-\frac{\sigma_a^2}{2(\sigma_a+\lambda_a)}} + O(\lambda_a^2\lambda_b) e^{-\frac{\sigma_b^2}{\lambda_b}}.
\end{align*}
%\qed
\end{proof}

Moreover, if $B_t$ is smaller than about $\lambda_b$, then $A_t$ cannot be smaller than about $\uk$ in the steady state.

\begin{lemma}
For any $\sigma_a\geq 1$ and $\sigma_b\geq 1$, we have
\begin{align*}
\Pr_{(A, B)\sim \pi} \left[ A \leq \uk - \sigma_a -1, B\leq \lambda_b + \sigma_b \right]
& =  O(\lambda_a) e^{-\frac{\sigma_a^2}{2(\sigma_a+\lambda_a)}} + O(\lambda_a\lambda_b+\lambda^3_b) e^{-\frac{\sigma_b^2}{\sigma_b+\lambda_b}}.
\end{align*}
\end{lemma}

\begin{proof}
The balance equation~\eqref{eq:1sideP-balance} can be transformed as follows.
\begin{align*}
\lambda_a \sum_{j=0}^{\infty}  \pi (k-1, j) 
&= \sum_{j=0}^{\infty}(k+j(1-(1-p)^{k})) \pi (k, j)\\
&\leq (k+(\lambda_b +\sigma_b)(1-(1-p)^{k})) \sum_{j=0}^{\lambda_b + \sigma_b} \pi (k, j)  + \sum_{j=\lambda_b + \sigma_b+1}^\infty (k+j) \pi (k, j).
\end{align*}
The LHS is at least 
\[
\lambda_a \sum_{j=0}^{\lambda_b + \sigma_b}  \pi (k-1, j).
\]
Define 
\[
\alpha_k = \frac{k+(\lambda_b +\sigma_b)(1-(1-p)^{k})}{\lambda_a}
%\]
\text{\quad and\quad}
%\[
\beta_k = \frac{1}{\lambda_a} \sum_{j=\lambda_b + \sigma_b+1}^\infty (k+j) \pi (k, j).
\]
Then, for $k\leq \uk$, we have
\[
\sum_{j=0}^{\lambda_b+\sigma_b}  \pi (k-1, j) \leq \alpha_k \sum_{j=0}^{\lambda_b+\sigma_b}  \pi (k, j) + \beta_k.
\]

We will apply Lemma~\ref{lem:recursion-extra}.
We observe that
\[
\alpha_k \leq \frac{k-\uk+\lambda_a}{\lambda_a}\leq \exp\left(-\frac{\uk-k}{\lambda_a}\right).
\]
It follows from Lemma~\ref{lem:1sidePatientCon} and \eqref{eq:1sideP_Z2-sum} that 
\begin{align*}
\sum_{k=0}^{\uk} \beta_k 
& \leq 
\frac{\uk}{\lambda_a} \sum_{j=\lambda_b + \sigma_b+1}^\infty \sum_{k=0}^{\infty} \pi (k, j)
+ \frac{1}{\lambda_a} \sum_{j=\lambda_b + \sigma_b+1}^\infty j \sum_{k=0}^{\infty} \pi (k, j)\\
& = 
O\left(\lambda_b \right)\exp\left(-\frac{\sigma_b^2}{\sigma_b +\lambda_b} \right)+
O\left(\frac{\lambda^3_b}{\lambda_a} \right)\exp\left(-\frac{\sigma_b^2}{\sigma_b +\lambda_b} \right)\\
& = 
O\left(\lambda_b + \frac{\lambda^3_b}{\lambda_a}\right)\exp\left(-\frac{\sigma_b^2}{\sigma_b +\lambda_b} \right).
\end{align*}
Therefore, by Lemma~\ref{lem:recursion-extra},
\begin{align*}
\sum_{k=0}^{\uk-\sigma_a-1} \sum_{j=0}^{\infty}  \pi (k, j) 
 =  O(\lambda_a) e^{-\frac{\sigma_a^2}{2(\sigma_a+\lambda_a)}} + O(\lambda_a) \sum_{k=0}^{\uk} \beta_k
 =  O(\lambda_a) e^{-\frac{\sigma_a^2}{2(\sigma_a+\lambda_a)}} + O\left(\lambda_a \lambda_b + \lambda^3_b\right) e^{-\frac{\sigma^2_b}{\sigma_b+\lambda_b}}.
\end{align*}
%\qed
\end{proof}

The above two lemmas, together with Lemma~\ref{lem:1sidePatientCon}, imply Proposition~\ref{prop:1sideP_notS2}.

We also present the following lemma, which will be used in bounding the loss in the next subsection.

\begin{lemma}\label{lem:1sideP_Z1-sum}
For any $\sigma_a\geq 1$, we have
  \begin{align*}
  \sum_{k=\lambda_a +\sigma_a+1}^\infty k \sum_{j=0}^{\infty} \pi (k, j) &= O\left(\lambda_a^3\right) e^{-\frac{\sigma_a^2}{\sigma_a +\lambda_a} }.
  \end{align*}
\end{lemma}
\begin{proof}
The balance equation~\eqref{eq:1sideP-balance} implies that 
\begin{align*}
\lambda_a \sum_{j=0}^{\infty}  \pi (k, j) \geq (k+1) \sum_{j=0}^{\infty} \pi (k+1, j).
\end{align*}
If $k\geq \lambda_a-1$, then it holds that
\[
\frac{\sum_{j=0}^{\infty} \pi (k+1, j)}{\sum_{j=0}^{\infty} \pi (k, j)} \leq \frac{\lambda_a}{k+1} = 1- \frac{k - \lambda_a+1}{k+1} \leq \exp\left(- \frac{k - \lambda_a+1}{k+1}\right).
\]
Hence, by applying Lemma~\ref{l7}, we obtain
  \begin{align*}
  \sum_{k=\lambda_a +\sigma_a+1}^\infty k \sum_{j=0}^{\infty} \pi (k, j) &= O\left(\lambda_a^3\right) \exp\left(-\frac{\sigma_a^2}{\sigma_a +\lambda_a} \right).
  \end{align*}
  %\qed
\end{proof}

%%%%%%%%%%%%%%%%%%%%%%%%%%%%%%%%%%%%%%%%%%%%%%%%%%%%%%%%
\subsubsection{Bounding the Loss of $\mathsf{Patient}_1$}\label{sec:p1_loss}
%%%%%%%%%%%%%%%%%%%%%%%%%%%%%%%%%%%%%%%%%%%%%%%%%%%%%%%%

It follows from Proposition~\ref{prop:1sideP_notS2} that the loss of $\mathsf{Patient}_1$ can be upper-bounded as below.
This with Lemma~\ref{lem:P1_Loss2ExpectedSizes} proves Theorem~\ref{t3} for $\mathsf{Patient}_1$.

\begin{theorem}
For a bipartite matching market $(d_a, d_b, p)$, 
let $\pi$ be the stationary distribution of the Markov chain of $\mathsf{Patient}_1$.
Then, if $d_a\geq d_b$, it holds that
\begin{align*}
    \frac{\mathbb{E}_{(A,B)\sim \pi}[A]}{\lambda_a} &\leq  \frac{d_a-  d_b}{d_a}+ \frac{\log (d_b+3)}{d_a} + o(1),\\
    \frac{\mathbb{E}_{(A,B)\sim \pi}[B(1-p)^A]}{\lambda_b} &\leq (1+o(1))\frac{\log (d_b+3)}{d_b} + o(1),
\end{align*}
If $d_a < d_b$, then
\begin{align*}
    \frac{\mathbb{E}_{(A,B)\sim \pi}[A]}{\lambda_a} &\leq  \frac{\log (d_b+3)}{d_a}  + o(1),\\
    \frac{\mathbb{E}_{(A,B)\sim \pi}[B(1-p)^A]}{\lambda_b} &\leq (1+o(1))\left(\frac{d_b-  d_a}{d_b}+  \frac{\log (d_b+3)}{d_b}\right) + o(1).
\end{align*}
\end{theorem}
\begin{proof}
Define $\sigma_a=\Theta (\sqrt{\lambda_a \log \lambda_a})$ and $\sigma_b=\Theta (\sqrt{\lambda_b \log (\lambda_a\lambda_b)})$.
Let $S=\{(k,j)\mid \uk-\sigma_a \leq k \leq \ok+\sigma_a, \lambda_b - \sigma_b\leq j \leq \lambda_b + \sigma_b\}$.
It  holds that
\begin{align}
\mathbb{E}_{(A, B)\sim \pi}[A] &\leq \sum_{(k,j)\in S} k \pi(k, j) +  \sum_{(k,j)\not\in S} k \pi(k, j),\label{eq:P1Loss1}\\
\mathbb{E}_{(A, B)\sim \pi}[B(1-p)^{A}] 
&\leq 
\sum_{(k,j)\in S} j(1-p)^k \pi(k, j)
+ \sum_{(k,j)\not\in S} j \pi(k, j).\label{eq:P1Loss2}
%\mathbb{E}_{(A, B)\sim \pi}[A + B(1-p)^{A}] 
%&\leq 
%\sum_{(k,j)\in S} (k+j(1-p)^k) \pi(k, j)
%+ \sum_{(k,j)\not\in S} j \pi(k, j)
%+ \sum_{(k,j)\not\in S} k \pi(k, j).\label{eq:P1Loss1}
\end{align}
The second term of~\eqref{eq:P1Loss1} is divided as follows.
\begin{align*}
\sum_{(k,j)\not\in S} k \pi(k, j)
&\leq 
(\lambda_a+\sigma_a) \sum_{(k, j)\not\in S} \pi(k, j)
+ \sum_{k\geq \lambda_a+\sigma_a+1}k \sum_{j\geq 0} \pi(k, j).
\end{align*}
The first term is bounded by a constant from Proposition~\ref{prop:1sideP_notS2}, while the second term is bounded by a constant by Lemma~\ref{lem:1sideP_Z1-sum}.
In addition, the second term of~\eqref{eq:P1Loss2} can be transformed to
\begin{align*}
\sum_{(k,j)\not\in S} j \pi(k, j)
&\leq 
(\lambda_b+\sigma_b) \sum_{(k, j)\not\in S} \pi(k, j)
+ \sum_{j\geq \lambda_b+\sigma_b+1}j \sum_{k\geq 0} \pi(k, j).
\end{align*}
Both the terms are bounded by constants from Proposition~\ref{prop:1sideP_notS2} and~\eqref{eq:1sideP_Z2-sum}.
Therefore, we obtain 
\begin{align*}
\mathbb{E}_{(A, B)\sim \pi}[A]  \leq \sum_{(k,j)\in S} k \pi(k, j) + O(1) 
\quad \text{and}\quad
\mathbb{E}_{(A, B)\sim \pi}[B(1-p)^A]  \leq \sum_{(k,j)\in S} j(1-p)^k \pi(k, j) + O(1). 
%\mathbb{E}_{(A, B)\sim \pi}[A + B(1-p)^{A}] 
%&\leq 
%\sum_{(k,j)\in S} (k+j(1-p)^k) \pi(k, j) + O(1) \notag \\
%&\leq \max_{(k,j)\in S} (k+j(1-p)^k) + O(1).\notag
\end{align*}
Since $\lambda_b - \sigma_b\leq j \leq \lambda_b + \sigma_b$ and $\uk-\sigma_a \leq k \leq \ok+\sigma_a$ in $S$, 
it holds that
\begin{align*}
\mathbb{E}_{(A, B)\sim \pi}[A]  &\leq \ok+\sigma_a + O(1) \leq k_1 + \sigma_a + \sigma_b +  O(1)
\end{align*}
since $\ok \leq k_1+\sigma_b$ by Lemma~\ref{l4}.
Moreover, it holds that
\begin{align*}
\mathbb{E}_{(A, B)\sim \pi}[B(1-p)^A]  &\leq (\lambda_b + \sigma_b) (1-p)^{\uk-\sigma_a}+ O(1)\\
&\leq (\lambda_b+o(\lambda_b))(1-p)^{k_1-\sigma_a-\sigma_b}+O(1)\\
&\leq  (1+o(1)) (\lambda_b+o(\lambda_b))(1-p)^{k_1} + O(1)
\end{align*}
where the second inequality follows from that $\uk \geq k_1-\sigma_b$ by Lemma~\ref{l4}, and the third inequality follows since  $(1-p)^{-\sigma_a-\sigma_b}\leq \frac{1}{1-p(\sigma_a+\sigma_b)}=1+o(1)$.

In summary, since $\sigma_a = o(\lambda_a)$ and  $\sigma_b = o(\lambda_b)$, we obtain 
\begin{align*}
    \frac{\mathrm{E}_{(A, B)\sim \pi}[A]}{\lambda_a}
 \leq  \frac{k_1}{\lambda_a} + o(1)
 \text{\quad and\quad}
    \frac{\mathrm{E}_{(A, B)\sim \pi}[B(1-p)^{A}]}{\lambda_b}
 \leq  (1+o(1)) (1-p)^{k_1} + o(1),
\end{align*}
where we note that, since $\lambda_a = \frac{d_a}{d_b}\lambda_b$ and $d_a$ and $d_b$ are constants, we may assume that $\lambda_a \leq C \lambda_b$ for some constant $C$.
They are identical with~\eqref{eq:G1lossA} and~\eqref{eq:G1lossB} in the proof of Theorem~\ref{15}, respectively.
Following the left part in the proof of Theorem~\ref{15}, we complete the proof. 
%\qed
\end{proof}

%%%%%%%%%%%%%%%%%%%%%%%%%%%%%%%%%%%%%%%%%%%%%%%%%%%%%%%%%%%%%%
\subsection{Reducing to Estimating Pool Sizes in the Steady State}\label{sec:appC}
%%%%%%%%%%%%%%%%%%%%%%%%%%%%%%%%%%%%%%%%%%%%%%%%%%%%%%%%%%%%%%

In this section, we prove Lemmas~\ref{lem:G2_Loss2ExpectedSizes},~\ref{lem:P2_Loss2ExpectedSizes},~\ref{lem:G1_Loss2ExpectedSizes}, and~\ref{lem:P1_Loss2ExpectedSizes} in Sections~\ref{sec:UB:G2}--\ref{sec:P1detail}.

We denote by $\tilde{A}_t$ and $\tilde{B}_t$ the pool sizes when an algorithm does nothing and no agents ever get matched~(see also the proof of Theorem~\ref{t5}).
It is easy to see that $A_t \leq \tilde{A}_t$ and $B_t \leq \tilde{B}_t$ for any $t$.

Akbarpour et al.~\cite{akbarpour2020thickness} calculated the probability that the pool size is a given integer $\ell$ for a (non-bipartite) matching market where each agent is inactive.
This can be directly applied to the bipartite matching markets, as $\tilde{A}_t$ and $\tilde{B}_t$ are independent.

\begin{proposition}(Akbarpour et al.~\cite{akbarpour2020thickness})\label{prop:inactive}
For any $t\geq 0$, it holds that
\[
\Pr [\tilde{A}_{t}=\ell]\leq \frac{\lambda_a^\ell}{\ell !}
\quad
\text{\quad and\quad }
\Pr [\tilde{B}_{t}=\ell]\leq \frac{\lambda_b^\ell}{\ell !}.
\]
Therefore, we have
\[
\mathbb{E} [\tilde{A}_{t}] = (1-e^{-t}) \lambda_a\leq \lambda_a
\text{\quad and \quad}
\mathbb{E} [\tilde{B}_{t}] = (1-e^{-t}) \lambda_b \leq \lambda_b.
\]
\end{proposition}

We can bound the expected sizes of $A_t$ and $B_t$ for a Markov chain $(A_t, B_t)$ such that $A_t\leq \tilde{A}_t$ and $B_t\leq \tilde{B}_t$ for any $t$.
%We here use the fact that $\sum_{i=6\lambda_a+1}^\infty \Pr [\tilde{A}_t\geq i]\leq 2^{-6\lambda_a}$.

\begin{proposition}(Akbarpour et al.~\cite{akbarpour2020thickness})\label{prop:G_reduc}
Let $(A_t, B_t)$ be the Markov chain such that $A_t\leq \tilde{A}_t$ and $B_t\leq \tilde{B}_t$ for any $t$.
Let $\pi$ be the stationary distribution, and $\tau_{\textup{mix}}(\epsilon)$ be the mixing time.
%pool sizes at time $t$ under $\mathSf{Greedy}_i$ for some $i=1,2$.
For any $t\geq \tau_{\textup{mix}}(\epsilon)$, it holds that
\[
\mathbb{E} [A_{t}] \leq \mathbb{E}_{(A, B)\sim \pi} [A] + 6 \epsilon \lambda_a + 2^{-6\lambda_a}
\quad\text{and}\quad
\mathbb{E} [B_{t}] \leq \mathbb{E}_{(A, B)\sim \pi} [B] + 6 \epsilon \lambda_b + 2^{-6\lambda_b}.
\]
\end{proposition}

%\begin{proposition}[\citep[Lemma~15]{akbarpour2020thickness}]
%Let $(A_t, B_t)$ be the pool sizes at time $t$ under $\mathSf{Patient}_i$ for some $i=1,2$.
%For any $t\geq \tau_{\text{mix}}(\epsilon)$, it holds that
%\[
%\mathbb{E} [A_{t}] \leq \mathbb{E}_{(A, B)\sim \pi} [A] + 6 \epsilon \lambda_1 + 2^{-6\lambda_1}
%\quad
%\mathbb{E} [B_{t}] \leq \mathbb{E}_{(A, B)\sim \pi} [B] + 6 \epsilon \lambda_2 + 2^{-6\lambda_2}
%\]
%\end{proposition}

\subsubsection{Proofs of Lemmas~\ref{lem:G2_Loss2ExpectedSizes} and ~\ref{lem:G1_Loss2ExpectedSizes}}\label{sec:Greduction}

We prove Lemma~\ref{lem:G2_Loss2ExpectedSizes}, as Lemma~\ref{lem:G1_Loss2ExpectedSizes} can be proved similarly.

%By Proposition~\ref{prop:inactive}, we have $\mathbb{E} [A_t]+\mathbb{E} [B_t]\leq \mathbb{E} [\tilde{A}_t]+\mathbb{E} [\tilde{B}_t]\leq \lambda_a + \lambda_b=m$.
By Proposition~\ref{prop:inactive}, we have $\mathbb{E} [A_t]\leq \mathbb{E} [\tilde{A}_t]\leq \lambda_a$.
This implies that
\begin{align*}
\mathbf{L}_a (\mathsf{Greedy}_2)
= \frac{1}{\lambda_a T} \mathbb{E} \left[\int_{t=0}^T A_t dt \right]
\leq \frac{1}{\lambda_aT} \lambda_a\tau_{\text{mix}} (\epsilon)
+ \frac{1}{\lambda_a T} \int_{t=\tau_{\text{mix}}(\epsilon)}^T \mathbb{E} \left[A_t\right] dt.
\end{align*}
Therefore, it follows from Proposition~\ref{prop:G_reduc} that
\[
\mathbf{L}_a (\mathsf{Greedy}_2) \leq \frac{\tau_{\text{mix}}(\epsilon)}{T} + 6 \epsilon + \frac{1}{\lambda_a}2^{-6\lambda_a} + \frac{\mathbb{E}_{(A, B)\sim \pi}[A]}{\lambda_a}.
\]
The argument for $\mathbf{L}_b(\mathsf{Greedy}_2)$ is symmetrical, and hence Lemma~\ref{lem:G2_Loss2ExpectedSizes} holds.
\qed
%\end{proof}

\subsubsection{Proofs of Lemmas~\ref{lem:P2_Loss2ExpectedSizes} and~\ref{lem:P1_Loss2ExpectedSizes}}
\label{sec:P2reduction}

We first prove Lemma~\ref{lem:P2_Loss2ExpectedSizes}.

\begin{proof}[Proof of Lemma~\ref{lem:P2_Loss2ExpectedSizes}]

It follows from Proposition~\ref{prop:inactive} that $\mathbb{E}[A_t(1-p)^{B_t}] \leq \mathbb{E}[A_t]\leq \mathbb{E}[\tilde{A_t}] \leq \lambda_a$.
%Similarly, we have $\mathbb{E}[B_t(1-p)^{A_t}]  \leq \lambda_b$.
Hence it follows that
\begin{align*}
    \textbf{L}_a(\mathsf{Patient}_2)
    =\frac{1}{\lambda_a}\left( \mathbb{E} \left[ \int_{t=0}^T A_t(1-p)^{B_t}dt \right] \right) 
    =\frac{\tau_{\text{mix}}(\epsilon)}{T} + \frac{1}{\lambda_a T}\int_{t=\tau_{\text{mix}}(\epsilon)}^T \mathbb{E} [A_t(1-p)^{B_t}]dt.
\end{align*}
Let $t\geq \tau_{\text{mix}}(\epsilon)$.
We denote $p_{ij}=\Pr [(A_t, B_t)=(i, j)]$ for simplicity.
It holds that
\begin{align*}
\mathbb{E} [A_t(1-p)^{B_t}] 
&= \sum_{i\geq 0}\sum_{j\geq 0} p_{ij} i (1-p)^j
\leq 
\sum_{i=0}^{6\lambda_a}\sum_{j\geq 0} p_{ij} i (1-p)^j
+ \sum_{i=6\lambda_a+1}^\infty i \sum_{j\geq 0} p_{ij}\\
& \leq 
\sum_{i=0}^{6\lambda_a}\sum_{j\geq 0} (\pi(i,j)+\epsilon) i (1-p)^j
+ \sum_{i=6\lambda_a+1}^\infty  \Pr [A_t \geq i] \\
& \leq 
\mathbb{E}_{(A, B)\sim \pi} [A (1-p)^B] + 6 \epsilon\lambda_a\sum_{j\geq 0} (1-p)^j + \sum_{i=6\lambda_a+1}^\infty  \Pr [\tilde{A}_t \geq i]\\
& \leq 
\mathbb{E}_{(A, B)\sim \pi} [A (1-p)^B] + 6 \epsilon\lambda_a\sum_{j\geq 0} (1-p)^j + 2^{-6\lambda_a}\\
& \leq 
\mathbb{E}_{(A, B)\sim \pi} [A (1-p)^B] + 6 \frac{\epsilon\lambda_a}{p} + 2^{-6\lambda_a},
\end{align*}
where the second inequality follows from the definition of the mixing time, and the third inequality follows since $\Pr [A_t \geq i]\leq \Pr [\tilde{A}_t \geq i]$ for any $i$ by the definition of $\tilde{A}_t$, and the fourth inequality comes from the fact that $\sum_{i=6\lambda_a+1}^\infty \Pr [\tilde{A}_t\geq i]\leq 2^{-6\lambda_a}$~(see \cite{akbarpour2020thickness}).
Therefore, we have 
\[
    \mathbf{L}_a(\mathsf{Patient}_2)
    \leq \frac{\mathrm{E}_{(A,B)\sim \pi}[A(1-p)^{B}]}{\lambda_a} + \frac{\tau_{\text{mix}}(\epsilon)}{T} + 6\frac{\epsilon}{p}+\frac{2^{-6\lambda_a}}{\lambda_a}. 
\]
Similarly, $\mathbb{E} [B_t(1-p)^{A_t}]  \leq \mathbb{E}_{(A, B)\sim \pi} [B (1-p)^A] + 6 \frac{\epsilon\lambda_b}{p} + 2^{-6\lambda_b}$, which proves Lemma~\ref{lem:P2_Loss2ExpectedSizes}.
%\qed
\end{proof}

We can use a similar argument for proving Lemma~\ref{lem:P1_Loss2ExpectedSizes}.
It holds that $\mathbb{E} [B_t(1-p)^{A_t}]  \leq \mathbb{E}_{(A, B)\sim \pi} [B (1-p)^A] + 6 \frac{\epsilon\lambda_b}{p} + 2^{-6\lambda_b}$.
Moreover, by Proposition~\ref{prop:G_reduc}, $\mathbb{E} [A_{t}] \leq \mathbb{E}_{(A, B)\sim \pi} [A] + 6 \epsilon \lambda_a + 2^{-6\lambda_a}$.
This gives Lemma~\ref{lem:P1_Loss2ExpectedSizes}.
\qed

\section{Lower Bounds}\label{sec:LB}

%\LBTwoSide*
%\begin{theorem}
%For a bipartite matching market $(d_a, d_b, p)$ with $d_a \geq d_b$, it holds that
%\begin{align*}
%    \mathbf{L} (\mathsf{OPT}) & \geq
%    \frac{1}{1+2d_a + d_b + 2d_a^2/\lambda_a + d_b^2/\lambda_b}
%\end{align*}
%\end{theorem}

%\LBOneSide*

%%%%%%%%%%%%%%%%%%%%%%%%%%%%%%%%%%%%%%%%%%%%%%%%%%%%%%%%%%%%%%%%%%%%%%%%
%\subsection{Overview: Lower Bounds}\label{s6}
%%%%%%%%%%%%%%%%%%%%%%%%%%%%%%%%%%%%%%%%%%%%

In this section, we prove Theorems~\ref{tt22} and~\ref{tt44}.

We first describe proof sketches of Theorems~\ref{tt22} and~\ref{tt44}.
Recall that $\mathsf{OPT}$ is the optimal algorithm under the assumption that the algorithm does not know the future information~(or even the departure information), and $\mathsf{OMN}$ is the algorithm that can use full information about the future.
Although the optimal solution is computationally hard to obtain, we can lower-bound them by estimating the loss which any matching algorithm suffers.

Let $(\zeta_a, \zeta_b)$ be the expected pool size of an arbitrary algorithm.
%We lower bound $\mathbf{L}(\mathsf{OPT})$ in two ways using $(\zeta_a, \zeta_b)$.
In the market, the expected rate that some agent gets critical is $\zeta_a + \zeta_b$. 
When the planner does not observe critical agents, all critical agents perish with probability one.
Hence, the loss is equal to $\frac{\zeta_a+\zeta_b}{\lambda_a+\lambda_b}$  in this case.

We compute the fraction of agents who form no edges upon arrival and during their sojourn.
If an agent has no edges, then no matching algorithm can match her.
Hence the fraction of such agents gives a lower bound on the loss of any matching algorithm.
As will be seen below, the fraction can be expressed as a function of $\zeta_a$ and $\zeta_b$.
Taking the worst case as $\zeta_a$ and $\zeta_b$ vary, we obtain a lower bound of $\mathbf{L}(\mathsf{OPT})$, which shows the first part of Theorems~\ref{tt22}.
The proof to lower-bound $\mathbf{L}(\mathsf{OMN})$ is similar, except for that we cannot use the fact that $\frac{\zeta_a+\zeta_b}{\lambda_a+\lambda_b}$.
Instead, we utilize the basic observation that $\zeta_a\leq \lambda_a$ and $\zeta_b\leq \lambda_b$.

For the 1-sided matching algorithms, we adopt a similar idea.
Under $\mathsf{Greedy}_1$, each agent in $U$ is inactive, and hence she does not get matched if she has no edges to agents who will arrive after her arrival.
Thus, for an agent in $U$, we compute the fraction of agents who do not form edges to agents arriving after her~(she may form edges to agents in the pool upon her arrival).
Similarly, a greedy agent does not get matched if she forms no edges to agents upon her arrival, and 
a patient agent does not if she has no edges at her departure.
Computing the fraction of such agents gives lower bounds of $\mathbf{L}(\mathsf{Greedy}_1)$ and $\mathbf{L}(\mathsf{Patient}_1)$.
This shows Theorem~\ref{tt44}.

Before going to each case, we prove the following, which was also shown in Jiang~\cite{jiang2018bipartite}.
Since $\mathbf{L}(\mathsf{OPT})\geq \mathbf{L}(\mathsf{OMN})$, this implies that $\frac{d_a-d_b}{d_a+d_b}$ is a lower bound of each algorithm.

\begin{lemma}\label{lem:DeltaLB}
For $d_a\geq d_b$, it holds that
\[
\mathbf{L}(\mathsf{OMN}) \geq \frac{d_a-d_b}{d_a+d_b}.
\]
\end{lemma}
\begin{proof}
We first note that the bipartite graph $G_t = (U_t, V_t, E_t)$ at time $t$ can have a matching of size at most $\min\{A_t,  B_t\}$.
This means that the number of matched agents at time $t$ is at most $2\min\{A_t,  B_t\}$.
Hence the total number of matched agents is at most
\begin{align*}
2\int_{t=0}^T \min\{A_t,  B_t\} dt & \leq 2 \min\left\{\int_{t=0}^T A_t dt, \int_{t=0}^T B_t dt\right\} = 2 \min\left\{\zeta_a, \zeta_b\right\}T
\end{align*}
where we define $\zeta_a:=\mathbb{E}_{t \sim \text{unif}[0, T]}[A_t]$ and $\zeta_b:=\mathbb{E}_{t \sim \text{unif}[0, T]}[B_t]$.
Since $\zeta_a \leq \lambda_a$ and $\zeta_b \leq \lambda_b$, this implies that
the total number of matched agents is at most $2 \lambda_b T$ as $\lambda_a\geq \lambda_b$.
Therefore, the loss is
\[
\mathbf{L}(\mathsf{OMN})\geq 1 - \frac{2\lambda_b T}{(\lambda_a+\lambda_b)T}  =\frac{d_a-d_b}{d_a+d_b}.
\]
This completes the proof.
\end{proof}

%%%%%%%%%%%%%%%%%%%%%%%%%%%%%%%%%%%%%%%%%%%%
\subsection{Lower Bound for $\mathbf{L}(\mathsf{OPT})$}
%%%%%%%%%%%%%%%%%%%%%%%%%%%%%%%%%%%%%%%%%%%%

We here prove the first part of Theorem~\ref{tt22}.
%\begin{proof}[Proof for $\mathsf{OPT}$]

Let $\zeta_a:=\mathbb{E}_{t \sim \text{unif}[0, T]}[A_t]$ and $\zeta_b:=\mathbb{E}_{t \sim \text{unif}[0, T]}[B_t]$.
Since $\mathsf{OPT}$ does not know the departure information, each critical agent perishes with probability $1$.
Therefore, it holds that
\begin{align}
    \textbf{L}(\mathsf{OPT}) = \frac{1}{\lambda_a T + \lambda_b T} \mathbb{E} \left[ \int_{0}^{T} A_t  + B_t  dt \right] = \frac{\zeta_a + \zeta_b}{\lambda_a+\lambda_b}. \label{e8}
\end{align}

Below we bound $\textbf{L}(\mathsf{OPT})$ in a different way.
Suppose that an agent $v$ arrives at time $t_0\sim \text{unif}[0, T]$.
If an agent $v$ does not form any edge during her sojourn, she must perish, which are counted in the loss.
Hence the loss is lower bounded by $\Pr [N(v) = \emptyset]$, where $N(v)$ is the set of neighbors of $v$.
Recall that $s(v)$ denotes the staying time of the agent $v$.
We also denote $U_{t_0, t_0+t}$~(resp., $V_{t_0, t_0+t}$) the set of agents in $U$~(resp., $V$) that arrive in the time interval $[t_0, t_0+t]$ for $t\geq 0$. 
By definition, we obtain 
\begin{align}
    \Pr[N(v) = \emptyset] \geq& \frac{\lambda_a}{\lambda_a+\lambda_b}\int_{0}^{\infty} \Pr [s(v)=t]\mathbb{E}[(1-p)^{|B_t|}]\mathbb{E}[(1-p)^{|V_{t_0, t_0+t}|}]dt \notag\\
    &\quad + \frac{\lambda_b}{\lambda_a+\lambda_b}\int_{0}^{\infty} \mathrm{P}[s(v)=t]\mathbb{E}[(1-p)^{|A_t|}]\mathbb{E}[(1-p)^{|U_{t_0, t_0+t}|}]dt \notag\\
    \geq& \frac{\lambda_a}{\lambda_a+\lambda_b}\int_{0}^{\infty} e^{-t}(1-p)^{\zeta_b}(1-p)^{\lambda_b t}dt 
     + \frac{\lambda_b}{\lambda_a+\lambda_b}\int_{0}^{\infty} e^{-t}(1-p)^{\zeta_a}(1-p)^{\lambda_a t}dt, \notag
\end{align}
where the second inequality follows from the Jensen's inequality. 
Since $1-p \geq e^{-p-p^2}$ for $p < 1/10$, we have 
\[
\int_{0}^{\infty} e^{-t}(1-p)^{\zeta_b}(1-p)^{\lambda_b t}dt 
\geq e^{-\zeta_b(p+p^2)}\int_{0}^{\infty} e^{-t(1+\lambda_b (p+p^2))}dt
\geq \frac{e^{-\zeta_b(p+p^2)}}{1+d_b+d_b^2/\lambda_b}.
%\geq \frac{1 -\zeta_b(p+p^2)}{1+d_b+d_b^2/\lambda_b}.
\]
Similarly, it holds that
\[
\int_{0}^{\infty} e^{-t}(1-p)^{\zeta_a}(1-p)^{\lambda_a t}dt 
\geq \frac{e^{-\zeta_a(p+p^2)}}{1+d_a+d_a^2/\lambda_b}.
\]
Therefore, we obtain
\begin{align}
    \textbf{L}(\mathsf{OPT}) 
    & \geq \frac{\lambda_a}{\lambda_a+\lambda_b}\cdot \frac{e^{-\zeta_b(p+p^2)}}{1+d_b+d_b^2/\lambda_b} + \frac{\lambda_b}{\lambda_a+\lambda_b}\cdot \frac{e^{-\zeta_a(p+p^2)}}{1+d_a+d_a^2/\lambda_a} \label{eq:LBforOMN}\\
    &\geq \frac{\lambda_a}{\lambda_a+\lambda_b}\cdot \frac{1 -\zeta_b(p+p^2)}{1+d_b+d_b^2/\lambda_b} + \frac{\lambda_b}{\lambda_a+\lambda_b}\cdot \frac{1-\zeta_a(p+p^2)}{1+d_a+d_a^2/\lambda_a}\notag\\
    & \geq  \min\left\{\frac{1 -\zeta_b(p+p^2)}{1+d_b+d_b^2/\lambda_b}, \frac{1-\zeta_a(p+p^2)}{1+d_a+d_a^2/\lambda_a} \right\}. \label{eq:OPT_LB1}
\end{align}
Thus $\textbf{L}(\mathsf{OPT})$ is lower bounded by the maximum of the RHSes of~\eqref{e8} and~\eqref{eq:OPT_LB1}.

First suppose that the former term of~\eqref{eq:OPT_LB1} is smaller than the latter.
For a fixed $\zeta_a$, the worst case is attained when
\[
\frac{\zeta_a}{\lambda_a + \lambda_b} + \frac{\zeta_b}{\lambda_a + \lambda_b}  = \frac{1 -\zeta_b(p+p^2)}{1+d_b+d_b^2/\lambda_b} 
\]
which is the case when 
\[
\zeta_b = \frac{\lambda_a + \lambda_b}{1+d_a+2d_b+d_a^2/\lambda_a+2d_b^2/\lambda_b}\left(1 - \frac{\zeta_a(1+d_b+d_b^2/\lambda_b)}{\lambda_a + \lambda_b}\right).
\]
Note that $\zeta_b$ may be negative, but this gives a lower bound.
This implies that
\begin{align*}
    \textbf{L}(\mathsf{OPT}) 
    &\geq 
    \frac{1}{1+d_a+2d_b+d_a^2/\lambda_a+2d_b^2/\lambda_b}
     -
    \frac{\zeta_a}{\lambda_a + \lambda_b}\frac{1+d_b+d_b^2/\lambda_b}{1+d_a+2d_b+d_a^2/\lambda_a+2d_b^2/\lambda_b}
    +
    \frac{\zeta_a}{\lambda_a + \lambda_b}\\
    &\geq 
    \frac{1}{1+d_a+2d_b+d_a^2/\lambda_a+2d_b^2/\lambda_b}
\end{align*}
since $\zeta_a\geq 0$.
Similarly, when the former term of~\eqref{eq:OPT_LB1} is larger than the latter, we have
\begin{align*}
    \textbf{L}(\mathsf{OPT}) 
    &\geq 
    \frac{1}{1+2d_a+d_b+2d_a^2/\lambda_a+d_b^2/\lambda_b}.
\end{align*}
Thus we obtain
\begin{align*}
    \textbf{L}(\mathsf{OPT}) 
    &\geq 
    \min\left\{\frac{1}{1+2d_a+d_b+2d_a^2/\lambda_a+d_b^2/\lambda_b},
    \frac{1}{1+d_a+2d_b+d_a^2/\lambda_a+2d_b^2/\lambda_b}
    \right\}\\
    &= \frac{1}{1+2d_a+d_b+2d_a^2/\lambda_a+d_b^2/\lambda_b}.
\end{align*}
This proves the first part of Theorem~\ref{tt22}.
\qed
%\end{proof}

%%%%%%%%%%%%%%%%%%%%%%%%%%%%%%%%%%%%%%%%%%%%
\subsection{Lower Bound for $\mathbf{L}(\mathsf{OMN})$}
%%%%%%%%%%%%%%%%%%%%%%%%%%%%%%%%%%%%%%%%%%%%

We prove the second part of Theorem~\ref{tt22}.
Since $\mathbf{L}(\mathsf{OMN})\geq \Pr [N(v)=\emptyset]$ for an agent $v$,
we can use the lower bound~\eqref{eq:LBforOMN} as for $\mathbf{L}(\mathsf{OPT})$.
Hence we have

\begin{align*}
    \textbf{L}(\mathsf{OMN}) \geq \mathrm{P}[N(v) = \emptyset] 
    &\geq \frac{\lambda_a}{\lambda_a+\lambda_b}\cdot \frac{e^{-\zeta_b(p+p^2)}}{1+d_b+d_b^2/\lambda_b} + \frac{\lambda_b}{\lambda_a+\lambda_b}\cdot \frac{e^{-\zeta_a(p+p^2)}}{1+d_a+d_a^2/\lambda_a} \label{e10} \\
    &\geq \frac{\lambda_a}{\lambda_a+\lambda_b}\cdot \frac{e^{-\lambda_b(p+p^2)}}{1+d_b+d_b^2/\lambda_b} + \frac{\lambda_b}{\lambda_a+\lambda_b}\cdot \frac{e^{-\lambda_a(p+p^2)}}{1+d_a+d_a^2/\lambda_a} \\
    &\geq \frac{1}{2}\left(\frac{e^{-\lambda_b(p+p^2)}}{1+d_b+d_b^2/\lambda_b} + \frac{e^{-\lambda_a(p+p^2)}}{1+d_a+d_a^2/\lambda_a}\right)
\end{align*}
where the second to the last inequality follows since $\lambda_a\geq \zeta_a$ and $\lambda_b\geq \zeta_b$ by Proposition~\ref{prop:inactive}, and the last inequality follows since $\gamma \alpha + (1-\gamma)\beta \geq (\alpha + \beta)/2$ for any $\alpha\geq \beta\geq 0$ and $1/2\leq \gamma \leq 1$.
Thus the proof of the second part of Theorem~\ref{tt22} is complete.
%\end{proof}
\qed

\subsection{Lower Bound for $\mathbf{L}(\mathsf{Greedy}_1)$}

Let $\zeta_a:=\mathbb{E}_{t \sim \text{unif}[0, T]}[A_t]$ and $\zeta_b:=\mathbb{E}_{t \sim \text{unif}[0, T]}[B_t]$.
Since $\mathsf{Greedy}_1$ does not know the departure information, it holds that
\begin{align}
    \textbf{L}(\mathsf{Greedy}_1) = \frac{\zeta_a + \zeta_b}{\lambda_a+\lambda_b}. 
\end{align}

Suppose that an agent $v$ arrives at time $t_0\sim \text{unif}[0, T]$.
Let $X_v$ be an event that $v$ is matched to no one.
If $v$ is an inactive agent in $U$, $X_v$ happens if $v$ does not form an edge to agents after her arrival.
Hence we can write
\begin{align}
\Pr [X_v] 
& \geq \int_{t=0}^{\infty} \Pr [s(v) = t] \mathbb{E} [(1-p)^{|V_{t_0, t_0+t}|}]dt \notag\\
& \geq  \int_{t=0}^{\infty} e^{-t} (1-p)^{\lambda_b t} dt \tag{by the Jensen's inequality}\nonumber\\
& \geq  \int_{t=0}^{\infty} e^{-t -\lambda_b t (p+p^2)} dt \tag{since $1-p\geq e^{-p-p^2}$ for $p < 1/10$}\nonumber\\
& = \frac{1}{1+\lambda_b (p+p^2)}.\label{eq:LB_G1}
\end{align}
%Hence the expected number of unmatched agents in $U$ is at least $\frac{\lambda_a}{1+\lambda_b (p+p^2)}$.
On the other hand, if $v$ is a greedy agent in $V$, then $X_v$ happens if $v$ does not form an edge to agents upon her arrival.
In this case, we can write
\begin{align*}
\Pr [X_v] 
 = \int_{t=0}^{\infty} \Pr [s(a) = t] \mathbb{E} [(1-p)^{A_{t_0}}]dt \geq  \int_{t=0}^{\infty} e^{-t} (1-p)^{\zeta_a} dt   = e^{- \zeta_a (p+p^2)},
\end{align*}
where the inequality follows from the Jensen's inequality.
%\begin{align*}
%\Pr [X_v] 
%& = \int_{t=0}^{\infty} \Pr [s(a) = t] \mathbb{E} [(1-p)^{A_{t_0}}]dt\\
%& \geq  \int_{t=0}^{\infty} e^{-t} (1-p)^{\zeta_a} dt \tag{by the Jensen's inequality}\\
%&  = e^{- \zeta_a (p+p^2)}.
%\end{align*}
%Hence the expected number of unmatched agents in $V$ is at least $\lambda_be^{- \zeta_a (p+p^2)}$.

%\begin{comment}
Therefore, the total loss is lower bounded by the following.
\begin{align}
\mathbf{L}(\mathsf{Greedy}_1) 
&\geq \frac{1}{\lambda_a+\lambda_b} \max\left\{\zeta_a,\frac{\lambda_a}{1+\lambda_b (p+p^2)} \right\}  + \frac{1}{\lambda_a+\lambda_b}\max\left\{\zeta_b, \lambda_b e^{- \zeta_a (p+p^2)}\right\}. \notag
\end{align}
%where we use the fact that $\max\{a+b, c+d\}\geq \frac{1}{2}(\max\{a, c\}+\max\{b, d\})$ for any real numbers $a, b, c, d$.
The worst case is attained when 
\[
\zeta_a = \frac{\lambda_a}{1+\lambda_b (p+p^2)} \quad\text{and}\quad \zeta_b = \lambda_b e^{- \zeta_a (p+p^2)}.
\]

First suppose that $\lambda_a \geq \lambda_b$.
Then, since $\zeta_b\geq 0$, we have  
\begin{align}\label{eq:LBG1case1}
\mathbf{L} (\mathsf{Greedy}_1) &\geq \frac{\lambda_a}{\lambda_a+\lambda_b} \frac{1}{1+\lambda_b (p+p^2)} = \frac{1}{2}\frac{1}{1+\lambda_b (p+p^2)}.
\end{align}
Next suppose that $\lambda_a < \lambda_b$.
In this case, 
\begin{align*}
\zeta_b  = \lambda_b e^{- \zeta_a (p+p^2)}
= \lambda_b \exp \left(- \frac{(p+p^2)\lambda_a}{1+\lambda_b (p+p^2)}\right)
\geq \lambda_b \left(1 - \frac{(p+p^2)\lambda_a}{1+\lambda_b (p+p^2)}\right)
= \lambda_b \left(\frac{1+(p+p^2)(\lambda_b-\lambda_a)}{1+\lambda_b (p+p^2)}\right).\label{eq:lb_e}
\end{align*}
Therefore, the total loss is 
\begin{align*}
\mathbf{L} (\mathsf{Greedy}_1) 
&\geq  \frac{\lambda_a}{\lambda_a+\lambda_b} \frac{1}{1+\lambda_b (p+p^2)} + \frac{\lambda_b}{\lambda_a+\lambda_b} \frac{1+(p+p^2)(\lambda_b-\lambda_a)}{1+\lambda_b (p+p^2)}\\
& =  \frac{1}{1+\lambda_b (p+p^2)} + \frac{\lambda_b(p+p^2)}{1+\lambda_b (p+p^2)} \frac{\lambda_b-\lambda_a}{\lambda_a+\lambda_b}
 \geq  \frac{1}{1+\lambda_b (p+p^2)} + \frac{d_b}{1+d_b} \frac{\lambda_b-\lambda_a}{\lambda_a+\lambda_b}\\
& \geq   \frac{1}{1+\lambda_b (p+p^2)} + \frac{1}{2} \frac{\lambda_b-\lambda_a}{\lambda_a+\lambda_b}
\end{align*}
since $\frac{d_b}{1+d_b}\geq \frac{1}{2}$ for $d_b\geq 1$.
%Thus the loss of the optimal algorithm is lower bounded.
This, together with~\eqref{eq:LBG1case1}, proves the first part of Theorem~\ref{tt44}.
\qed
%\end{comment}

\subsection{Lower Bound of $\mathbf{L}(\mathsf{Patient}_1)$}

Let $\zeta_a:=\mathbb{E}_{t \sim \text{unif}[0, T]}[A_t]$ and $\zeta_b:=\mathbb{E}_{t \sim \text{unif}[0, T]}[B_t]$.
Suppose that an agent $v$ arrives at time $t_0\sim \text{unif}[0, T]$.
Let $X_v$ be an event that $v$ is matched to no one.
For an inactive agent $v$ in $U$, $X_v$ happens if $v$ does not form an edge to any agent.
In this case, we can write
\begin{align*}
\Pr [X_v] 
 \geq \int_{t=0}^{\infty} \Pr [s(v) = t] \mathbb{E} [(1-p)^{B_{t_0}}]\mathbb{E} [(1-p)^{|V_{t_0, t_0+t}|}]dt
 \geq  \int_{t=0}^{\infty} e^{-t} (1-p)^{\zeta_b+\lambda_b t} dt 
%\tag{by the Jensen's inequality}\\
  \geq \frac{e^{- \zeta_b (p+p^2)}}{1+\lambda_b (p+p^2)},
%\tag{since $1-p\geq e^{-p-p^2}$ for $p < 1/10$}.
\end{align*}
where the second inequality follows from the Jensen's inequality, and the last one holds since $1-p\geq e^{-p-p^2}$ for $p < 1/10$.
Hence the expected number of unmatched agents in $U$ is at least $\lambda_a \frac{e^{- \zeta_b (p+p^2)}}{1+\lambda_a (p+p^2)}$.
Also, $\mathbb{E}[A_t]\geq \zeta_a$ holds, since agents in $U$ are inactive.
For a patient agent $v$ in $V$, $X_v$ happens if $v$ does not form an edge to any agent at her departure.
Hence, it holds that
\begin{align*}
\Pr [X_v] 
 \geq \int_{t=0}^{\infty} \Pr [s(v) = t] \mathbb{E} [(1-p)^{A_{t_0+t}}]dt
 \geq  \int_{t=0}^{\infty} e^{-t} (1-p)^{\zeta_a} dt 
 %\tag{by the Jensen's inequality}\\
  \geq e^{- \zeta_a (p+p^2)},
  %\tag{since $1-p\geq e^{-p-p^2}$ for $p < 1/10$}.
\end{align*}
where we again use the Jensen's inequality and the fact that $1-p\geq e^{-p-p^2}$ for $p < 1/10$.
Hence the expected number of unmatched agents in $V$ is at least $\lambda_b e^{- \zeta_a (p+p^2)}$.

Therefore, we obtain a lower bound of the loss:
\begin{align*}
\mathbf{L}(\mathsf{Patient}_1) &\geq \frac{1}{\lambda_a+\lambda_b} \max\left\{\zeta_a, \lambda_a \frac{e^{- \zeta_b (p+p^2)}}{1+\lambda_b (p+p^2)}\right\} + \frac{\lambda_b}{\lambda_a+\lambda_b}e^{- \zeta_a (p+p^2)}.
\end{align*}
The worst case is attained when
\[
\zeta_a = \lambda_a \frac{e^{- \zeta_b (p+p^2)}}{1+\lambda_b (p+p^2)}.
\]
In this case, 
\begin{align*}
\mathbf{L}(\mathsf{Patient}_1) &\geq \frac{1}{\lambda_a+\lambda_b} \left(\zeta_a + \lambda_b e^{- \zeta_a (p+p^2)} \right).
\end{align*}
The RHS is minimized when $\zeta_a = \frac{1}{p+p^2}\log (d_b+d_bp)$.
Hence, we have 
\begin{align*}
\mathbf{L}(\mathsf{Patient}_1) &\geq \frac{\log (d_b+d^2_b/\lambda_b) + 1}{d_a+d_b+d_ap+d_bp}.
\end{align*}
\qed

\section{Conclusion}\label{s7}
%\dznote{Maybe we could revise some in this section.}
%\nknote{not some. ALL}
In this work, we studied a bipartite matching market model with arrivals and departures.
We proposed 1-sided/2-sided local algorithms with different timing properties, the Greedy and Patient algorithms. We achieved both upper and lower bounds on the performance of these algorithms, which are shown to be almost tight.
In addition, we provided lower bounds on the performance of any matching algorithms.

Our results indicate that waiting to thicken the market is highly valuable for the balanced 2-sided market, which is a similar conclusion to Akbarpour et al.~\cite{akbarpour2020thickness} and Beccara et al.~\cite{baccara2020optimal}.
Even when the market is not balanced, the loss of the smaller side can be made much smaller by waiting.
%We note that when the market is unbalanced, it is unavoidable to have a constant loss $\frac{d_a-d_b}{d_a+d_b}$ in the larger side.
On the other hand, waiting is not valuable for the 1-sided market.
It means that, to improve the loss in the 1-sided market, the departure information is not so beneficial, and other information, such as the time that her neighbors arrive/depart, is necessary to obtain smaller loss.
%For example, an agent has to take care of new agents arriving after her.

Our models are simple and developed from theoretical interest.
Although our models carry practical implications on waiting, they ignore some aspects of practical settings in reality.
For example, the probability that two agents are compatible is set to be constant over time $[0, T]$ and among agents.
However, the probability may be affected by past information~\cite{parker2016platform}.
Also, our model ignores waiting costs.
It would be interesting to incorporate such practical settings into our model and analyze the performance of matching algorithms.

\bibliographystyle{plain}

%\bibliographystyle{plainnat}
%\bibliography{reference}  %%% Uncomment this line and comment out the ``thebibliography'' section below to use the external .bib file (using bibtex) .

\clearpage
\appendix
\section{Lemmas for Solving Recursions}\label{sec:recursion}

In this section, we prove lemmas used in Section~\ref{sec:UB}.
We first provide the following lemmas by Akbarpour et al.~\cite{akbarpour2020thickness} with basic observation. 
Then we prove Lemmas~\ref{l7},~\ref{lem:recursion-denom-dec}, and~\ref{lem:recursion-extra} in the subsequent sections.

%\nknote{The paragraph has no meaning. Using "some", "different", gives us no specific explanation}
%In this section, first we provide some useful lemmas which have been used throughout the paper. Then we provide some useful lemmas used in different sections. We also proved some lemmas or theorems in the text.

%\subsection{Auxiliary Inequalities and propositions}

%\nknote{I found that there are several versions of Akbarpour's paper which has different numbering. So it may be better to remove the numbers unless the corresponding equations appear in the final journal version}

\begin{lemma}(Akbarpour et al.~\cite{akbarpour2020thickness})\label{l9} 
  For any $a, b\geq 0$, we have
  \[
  \sum_{i=a}^\infty e^{-bi^2} = \frac{e^{-ba^2}}{1-e^{-2ab}}\leq \frac{e^{-ba^2}}{\min\{ab, 1/2\}}.
  \]
\end{lemma}

\begin{lemma}(Akbarpour et al.~\cite{akbarpour2020thickness})\label{l6}
  For any $a, b\geq 0$, we have
  \[
  \sum_{i=a}^\infty i e^{-bi} \leq \frac{e^{-ba}(2ab+4)}{b^2}.
  \]
\end{lemma}

\begin{lemma}\label{l5}
 It holds that, for any $x\geq 0$, 
 \[
 \frac{1}{1-e^{-1/x}} \leq \frac{x^{2}}{x-\frac{1}{2}} = O(x).
 \]
\end{lemma}
\begin{proof}
 Since
 $e^{-x}\leq 1-x+\frac{x^2}{2}$,
 we have
 \[
 \frac{1}{1-e^{-1/x}} \leq \frac{1}{\frac{1}{x}-\frac{1}{2x^2}} = \frac{x^{2}}{x-\frac{1}{2}} = O(x).
 \]
 %\qed
\end{proof}

%\section{Appendix}
%In the appendix, we provide some inequalities that we use throughout the paper. In addition, we provide proofs of some key steps in the text.

%%%%%%%%%%%%%%%%%%%%%%%%%%%%%%%%%%%%%%%%%%%%%%%%%%%%%%%%%%%%%
%\subsection{Solving Recursions}
%%%%%%%%%%%%%%%%%%%%%%%%%%%%%%%%%%%%%%%%%%%%%%%%%%%%%%%%%%%%%

%\recursionOne*

%%%%%%%%%%%%%%%%%%%%
\subsection{Proof of Lemma~\ref{l7}}
%%%%%%%%%%%%%%%%%%%%%%%%%%%%%%%%%%%%%%%%%
%\begin{proof}
By definition, it holds that
  \begin{align*}
  f(k+1)&\leq \exp\left(-\frac{k-k^\ast}{k+\eta}\right) f(k)
  \leq \cdots \leq  \exp\left(-\sum_{i=k^\ast}^{k}\frac{i-k^\ast}{i+\eta}\right)f(k^\ast)
  \leq  \exp\left(-\sum_{i=0}^{k-k^\ast}\frac{i}{i+k^\ast+\eta}\right) \\
  &\leq  \exp\left(-\frac{1}{k+\eta} \sum_{i=0}^{k-k^\ast}i \right)
   \leq  \exp\left(-\frac{1}{k+\eta} (k-k^\ast)^2\right).
  \end{align*}
  This implies that, for any $\sigma\geq 1$, we have
  \begin{align*}
  \sum_{k=k^\ast+ \sigma+1}^\infty f(k)
  &\leq \sum_{k=k^\ast+ \sigma}^\infty \exp\left(-\frac{1}{k+\eta} (k-k^\ast)^2\right)
  \leq \sum_{k=\sigma}^\infty \exp\left(-\frac{1}{k+k^\ast+\eta} k^2\right)
  \leq \sum_{k=\sigma}^\infty \exp\left(-\frac{1}{1+\frac{k^\ast+\eta}{\sigma}} k\right)\\
  & = \frac{1}{1-\exp\left(-\frac{\sigma}{\sigma+k^\ast+\eta}\right)} \exp\left(-\frac{\sigma^2}{\sigma +k^\ast+\eta} \right).
  \end{align*}
  
  Since Lemma~\ref{l5} implies that
  \[
  \frac{1}{1-\exp\left(-\frac{\sigma}{\sigma+k^\ast+\eta}\right)} = O\left(\frac{\sigma +k^\ast+\eta}{\sigma}\right),
  % = O\left(k^\ast+\eta\right)
  \]
  this proves the first part of the lemma.
  %when $\sigma\geq 1$, we have the first part of the lemma.

  Similarly, it holds that
  \begin{align*}
  \sum_{k=k^\ast+ \sigma+1}^\infty k f(k)
  &\leq \sum_{k=k^\ast+ \sigma}^\infty (k+1) \exp\left(-\frac{1}{k+\eta} (k-k^\ast)^2\right)\\
  &\leq 
  \sum_{k=\sigma}^\infty k \exp\left(-\frac{1}{k+k^\ast+\eta} k^2\right)
  + (k^\ast+1) \sum_{k=\sigma}^\infty  \exp\left(-\frac{1}{k+k^\ast+\eta} k^2\right)\\
  & = 
  \sum_{k=\sigma}^\infty k \exp\left(-\frac{1}{k+k^\ast+\eta} k^2\right)
  +  O\left(k^\ast(k^\ast+\eta)\right) \exp\left(-\frac{\sigma^2}{\sigma +k^\ast+\eta} \right),
  \end{align*}
  where the last equality follows from~\eqref{eq:recursion-denom}.
  The first term is bounded by Lemma~\ref{l6} as follows.
  \begin{align*}
  \sum_{k=\sigma}^\infty k \exp\left(-\frac{1}{k+k^\ast+\eta} k^2\right)
  &\leq 
  \sum_{k=\sigma}^\infty k \exp\left(-\frac{1}{1+\frac{k^\ast+\eta}{\sigma}} k\right)
  \leq 
  \left(\frac{\sigma+k^\ast +\eta}{\sigma}\right)^2
      \left(2 \frac{\sigma^2}{ \sigma+k^\ast+\eta}+4\right)
    e^{-\frac{\sigma^2}{\sigma+k^\ast+\eta}}\\
  & = 
  \frac{(\sigma + k^\ast+\eta)\left(2\sigma^2+4\sigma + 4k^\ast+4\eta\right)}{\sigma^2}
    e^{-\frac{\sigma^2}
             {\sigma +k^\ast+ \eta}
      }.
  \end{align*}
  Therefore, when $\sigma =O(k^\ast+ \eta)$ and $\sigma\geq 1$, it is  $O\left((k^\ast+\eta)^3\right)e^{-\frac{\sigma^2}
             {\sigma +k^\ast+ \eta}
      }$.
  This proves the second part of the lemma.
  \qed
%\end{proof}

%%%%%%%%%%%%%%%%%%%%%%%%%%%%%%%%%%%%
\subsection{Proof of Lemma~\ref{lem:recursion-denom-dec}}
%%%%%%%%%%%%%%%%%%%%%%%%%%%%%%%%%%%%

%We next prove Lemma~\ref{lem:recursion-denom-dec}.

%\recursionTwo*

%\begin{proof}
  The proof is similar to Lemma~\ref{l7}.
  By definition, it holds that 
  \begin{align*}
  f(k-1)&\leq \exp\left(-\frac{k^\ast-k}{\eta}\right) f(k) 
  \leq \cdots \leq  \exp\left(-\sum_{i=k}^{k^\ast}\frac{k^\ast-i}{\eta}\right)f(k^\ast)\\
  %&\leq  \exp\left(-\sum_{i=0}^{k^\ast-k}\frac{i}{\eta}\right)\\
  &\leq  \exp\left(-\frac{1}{\eta} \sum_{i=0}^{k^\ast-k}i \right)
  \leq  \exp\left(-\frac{1}{\eta} (k^\ast-k)^2\right).
  \end{align*}
  This implies by Lemma~\ref{l9} that, for any $\sigma\geq 1$, we have
  \begin{align*}
  \sum_{k=0}^{k^\ast - \sigma-1} f(k)
  \leq \sum_{k=1}^{k^\ast- \sigma} \exp\left(-\frac{1}{\eta} (k^\ast-k)^2\right)
  \leq \sum_{k=\sigma}^\infty \exp\left(-\frac{1}{\eta} k^2\right)
  \leq e^{-\frac{\sigma^2}{\eta}}\frac{1}{\min\{ \sigma/\eta, 1/2\}}= O(\eta) e^{-\frac{\sigma^2}{\eta}}.
  \end{align*}
  \qed
%\end{proof}

%%%%%%%%%%%%%%%%%%%%%%%%%%%%%%%%%%%%%%%%%%%%%%%%%%%%%%%%%%%%%
\subsection{Proof of Lemma~\ref{lem:recursion-extra}}
%\subsection{Solving Recursions with Additive Terms}\label{sec:RecursionWAddition}
%%%%%%%%%%%%%%%%%%%%%%%%%%%%%%%%%%%%%%%%%%%%%%%%%%%%%%%%%%%%%

%In this section, we prove Lemma~\ref{lem:recursion-extra}.

%\recursionExtra*

%\begin{proof}
  Applying the given inequality repeatedly, we have
  \begin{align*}
  g(k+1) & \leq \alpha_k g(k) + \beta_k \leq \alpha_k \alpha_{k-1} g(k-1) + \alpha_k \beta_{k-1} + \beta_k\\
  & \leq \alpha_k \alpha_{k-1} \alpha_{k-2}g(k-2) + \alpha_k \alpha_{k-1}\beta_{k-2} +\alpha_k \beta_{k-1} + \beta_k\\
  & \leq \dots \leq \left(\prod_{i=k^\ast}^k \alpha_i\right) g(k^\ast) + \sum_{i=k^\ast}^k \left(\prod_{j=i+1}^k \alpha_j\right) \beta_i.
  \end{align*}
  For $k_1, k_2\in \mathbb{N}$, we denote 
  \[
  \gamma (k_1, k_2) = \prod_{i=k_1}^{k_2} \alpha_i.
  \]
  When $k_1>k_2$, define $\gamma (k_1, k_2)=0$.
  Since $g(k^\ast)\leq 1$, we have for any $k\geq k^\ast$, 
  \begin{align}\label{eq:g}
  g(k+1) & \leq \gamma(k^\ast, k) + \sum_{i=k^\ast}^k \gamma (i+1, k) \beta_i.
  \end{align}

  Summing up $g(k)$'s of \eqref{eq:g}, we have
  \begin{align*}
  \sum_{k=k^\ast+\sigma+1}^\infty g(k) \leq \sum_{k=k^\ast+\sigma}^\infty \gamma (k^\ast, k) + \sum_{k=k^\ast+\sigma}^\infty \sum_{i=k^\ast}^k \gamma (i+1, k) \beta_i.
  \end{align*}
  Since the second term is equal to 
  \[
  \sum_{k=k^\ast+\sigma}^\infty \sum_{i=k^\ast}^k \gamma (i+1, k) \beta_i
  = \sum_{i=k^\ast}^\infty  \sum_{k=k^\ast+\sigma}^\infty \gamma (i+1, k) \beta_i,
  \]
  we obtain 
  \begin{align}\label{eq:g_sum}
  \sum_{k=k^\ast+\sigma+1}^\infty g(k) \leq \sum_{k=k^\ast+\sigma}^\infty \gamma (k^\ast, k) + \sum_{i=k^\ast}^\infty  \sum_{k=k^\ast+\sigma}^\infty \gamma (i+1, k) \beta_i.
  \end{align}

 To show Lemma~\ref{lem:recursion-extra}, we provide a sequence of claims to bound each term of~\eqref{eq:g_sum}.
%\end{proof}

  \begin{claim}\label{clm:gamma}
     For any $j\geq 0$, it holds that
     \begin{align*}
     \gamma (k^\ast+j, k) \leq \exp \left(- \frac{1}{2(k + \eta)}(k-k^\ast+j)(k-k^\ast-j) \right).
     \end{align*}
  \end{claim}
  \begin{proof}[Proof of Claim~\ref{clm:gamma}]
     We have
     \begin{align*}
     \prod_{i=k^\ast+j}^k \alpha_k & \leq \prod_{i=k^\ast+j}^k \exp \left(- \frac{i- k^\ast}{i + \eta}\right) 
     = \exp \left(- \sum_{i=k^\ast+j}^k \frac{i- k^\ast}{i + \eta}\right)
      \leq \exp \left(- \frac{1}{k + \eta}\sum_{i=k^\ast+j}^k (i- k^\ast) \right)\\
     & \leq \exp \left(- \frac{1}{2(k + \eta)}(k-k^\ast+j)(k-k^\ast-j) \right).
     \end{align*}
     %\qed
  \end{proof}

  The first term of \eqref{eq:g_sum} is then bounded as follows.

  \begin{claim}\label{clm:gamma_sum1}
  It holds that
  \begin{align}\label{eq:gamma_sum1}
  \sum_{k=k^\ast+\sigma}^\infty \gamma (k^\ast, k) 
  = O\left(k^\ast + \eta\right)e^{-\frac{\sigma^2}{2(\sigma+k^\ast + \eta)}}.
  \end{align}
  \end{claim}
  \begin{proof}[Proof of Claim~\ref{clm:gamma_sum1}]
    Using Claim~\ref{clm:gamma} with $j=0$, we have
    \begin{align*}
    \sum_{k=k^\ast+\sigma}^\infty \gamma (k^\ast, k) 
    & \leq  \sum_{k=k^\ast+\sigma}^\infty \exp \left(- \frac{1}{2(k + \eta)}(k- k^\ast)^2 \right)
     = \sum_{k=\sigma}^\infty \exp \left(- \frac{k}{2(k +k^\ast + \eta)} k \right)\\
    & \leq \sum_{k=\sigma}^\infty \exp \left(- \frac{1}{2\left(1 +\frac{k^\ast + \eta}{\sigma}\right)} k \right)
     = \frac{e^{-b\sigma}}{1-e^{-b}} = O\left(\frac{1}{b} \right)e^{-b\sigma}.
    \end{align*}
    where $b=\frac{1}{2\left(1 +\frac{k^\ast + \eta}{\sigma}\right)}=\frac{\sigma}{2(\sigma+k^\ast+\eta)}$.
    Since $\sigma\geq 1$, 
    \[
    O\left(\frac{1}{b} \right) =O\left(\frac{\sigma+k^\ast+\eta}{\sigma}\right)= O(k^\ast + \eta),
    \]
    which proves the claim.
    %\qed
  \end{proof}

We next bound the second term of~\eqref{eq:g_sum}.
To do it, we provide the following claims.

  \begin{claim}\label{clm:gamma_sum2}
  For any $j\geq 0$, we have 
  \begin{align}\label{eq:gamma_sum2}
  \sum_{k=k^\ast+\sigma+j}^\infty \gamma (k^\ast+\sigma+j, k) 
  & 
%  \leq \frac{1}{1-e^{-b_j}}= O\left(\frac{\sigma+ j + k^\ast + \eta}{\sigma+j}\right)
  = O\left(k^\ast + \eta\right),
  \end{align}
  where $b_j = \frac{\sigma+j}{2(\sigma+ j + k^\ast + \eta)}$.
%  Moreover, it holds that
%  \begin{align*}
%  \sum_{k=k^\ast+\sigma+j}^\infty k \gamma (k^\ast+\sigma+j, k) 
%  &   =   O\left((k^\ast+\eta)^2(\sigma+j)\right).
%  \end{align*}
  \end{claim}
  \begin{proof}[Proof of Claim~\ref{clm:gamma_sum2}]
    The proof is similar to Claim~\ref{clm:gamma_sum1}.
    It follows from Claim~\ref{clm:gamma} that 
    \begin{align*}
    \sum_{k=k^\ast+\sigma+j}^\infty \gamma (k^\ast+\sigma+j, k) 
    & \leq  \sum_{k=k^\ast+\sigma+j}^\infty \exp \left(- \frac{1}{2(k + \eta)}((k-k^\ast)^2 - (\sigma +j)^2) \right)\\
    & = \sum_{k=\sigma+j}^\infty \exp \left(- \frac{1}{2(k + k^\ast + \eta)} (k^2- (\sigma +j)^2) \right)\\
    & =  \sum_{k=\sigma+j}^\infty \exp \left( \frac{(\sigma +j)^2}{2(k + k^\ast + \eta)}  \right) \exp \left(- \frac{1}{2(k + k^\ast + \eta)} k^2 \right)\\
    & \leq  \exp \left( \frac{(\sigma +j)^2}{2(\sigma +j + k^\ast + \eta)}  \right) \sum_{k=\sigma+j}^\infty  \exp \left(- \frac{1}{2\left(1 + \frac{k^\ast + \eta}{\sigma+j}\right)} k \right)\\
    & = e^{b_j(\sigma +j)} \frac{e^{-b_j(\sigma+j)}}{1-e^{-b_j}}
     = \frac{1}{1-e^{-b_j}},
    \end{align*}
    where $b_j = \frac{1}{2\left(1 +\frac{k^\ast + \eta}{\sigma+j}\right)}=\frac{\sigma+j}{2(\sigma+ j + k^\ast + \eta)}$.
    By Lemma~\ref{l5}, we see that
    \[
    \frac{1}{1-e^{-b_j}} = O\left( \frac{\sigma+ j + k^\ast + \eta}{\sigma+j}\right) = O(k^\ast +\eta),
    \]
    since $\sigma+j\geq 1$.
    This proves the claim.
    %\qed
  \end{proof}

  We are ready to bound the second term of \eqref{eq:g_sum}.

  \begin{claim}\label{clm:gamma_beta_sum}
  It holds that 
  \begin{align*}
     \sum_{i=k^\ast}^\infty  \sum_{k=k^\ast+\sigma}^\infty \gamma (i+1, k) \beta_i 
     &\leq O(k^\ast+\eta) \sum_{k=k^\ast}^\infty \beta_i.
  \end{align*}
  \end{claim}
  \begin{proof}[Proof of Claim~\ref{clm:gamma_beta_sum}]
     Since $\gamma (i+1, k)\leq \gamma (k^\ast+\sigma, k)$ if $i\leq k^\ast+\sigma -1$, it follows that
     \begin{align*}
     \sum_{i=k^\ast}^\infty  \sum_{k=k^\ast+\sigma}^\infty \gamma (i+1, k) \beta_i 
     & = \sum_{i=k^\ast}^{k^\ast+\sigma-1} \sum_{k=k^\ast+\sigma}^\infty \gamma (i+1, k) \beta_i
     + \sum_{i=k^\ast+\sigma}^{\infty} \sum_{k=i+1}^\infty \gamma (i+1, k) \beta_i\nonumber\\
     &\quad  \leq \sum_{k=k^\ast+\sigma}^\infty \gamma (k^\ast+\sigma, k) \sum_{i=k^\ast}^{k^\ast+\sigma-1}  \beta_i
     + \sum_{j=\sigma}^{\infty} \sum_{k=k^\ast+j+1}^\infty \gamma (k^\ast+j+1, k) \beta_{k^\ast+j}.
     \end{align*}
     By Claim~\ref{clm:gamma_sum2}, 
     \[
     \sum_{k=k^\ast+\sigma}^\infty \gamma (k^\ast+\sigma, k) = O(k^\ast+\eta)
     \quad \text{and}\quad
     \sum_{k=k^\ast+j+1}^\infty \gamma (k^\ast+j+1, k)= O(k^\ast+\eta).
     \]
     Hence, it holds that
     \begin{align*}
     \sum_{i=k^\ast}^\infty  \sum_{k=k^\ast+\sigma}^\infty \gamma (i+1, k) \beta_i 
     \leq O(k^\ast+\eta) \sum_{i=k^\ast}^{k^\ast+\sigma-1}  \beta_i
     + O(k^\ast+\eta) \sum_{j=\sigma}^{\infty}\beta_{k^\ast+j}
     = O(k^\ast+\eta) \sum_{i=k^\ast}^\infty  \beta_i.
     \end{align*}
     %\qed
  \end{proof}
  Lemma~\ref{lem:recursion-extra} follows from Claims~\ref{clm:gamma_sum2} and~\ref{clm:gamma_beta_sum}.
  \qed

\end{document}